\documentclass[acmsmall,nonacm,screen]{acmart}
\usepackage{booktabs}
\usepackage{stmaryrd}
\usepackage{tikz}
\usepackage{adjustbox}
\usepackage{tabularx}
\usepackage{multirow}
\usepackage{amssymb} 
\usepackage{bbold}
\usepackage{colortbl}
\usetikzlibrary{graphs,chains,matrix,calc}

\usepackage{forest} 
\usepackage{tikz-qtree} 

\newcommand{\NRMDatalogSN}{\texorpdfstring{$\text{NRMD}^{\neg}$}{NRMD¬}}
\newcommand{\MDatalogSN}{\texorpdfstring{$\text{MD}^{\neg}$}{MD¬}}
\newcommand{\MRA}{{MRA}}
\newcommand{\SPARQL}{{SPARQL}}

\newcommand{\Variable}[1]{{{?}{#1}}}
\newcommand{\varX}{\Variable{X}}
\newcommand{\varY}{\Variable{Y}}
\newcommand{\varZ}{\Variable{Z}}

\newcommand{\sI}{\mathbf{I}}
\newcommand{\sL}{\mathbf{L}}
\newcommand{\sV}{\mathbf{V}}
\newcommand{\sT}{\mathbf{T}}
\newcommand{\sC}{\mathbf{C}}
\newcommand{\sA}{\mathbf{A}}
\newcommand{\sN}{\mathbf{N}}

\newcommand{\calL}{\mathcal{L}}
\newcommand{\calQ}{\mathcal{Q}}
\newcommand{\calD}{\mathcal{D}}
\newcommand{\calS}{\mathcal{S}}

\newcommand{\Eval}{\operatorname{Eval}}

\newcommand{\unull}{\operatorname{null}}

\newcommand{\ans}{\operatorname{atoms}}

\newcommand{\gpattern}{\operatorname{gp}}

\newtheorem{claim}{Claim}

\newcommand{\colCopy}[2]{\langle#1,#2\rangle}

\newcommand\rsb{\rule{0pt}{2.6ex}} 
\newcommand\rsa{\rule[-1.2ex]{0pt}{0pt}} 
\newcommand\rs{\rsb \rsa} 

\newcommand{\NullRel}{\operatorname{Null}}
\newcommand{\CompRel}{\operatorname{Comp}}
\newcommand{\TripRel}{\operatorname{Trip}}

\newcommand{\dt}{\operatorname{dt}}

\newcommand{\col}{\operatorname{coloring}}

\newcommand{\equal}{\operatorname{eq}}

\newcommand{\triple}{\operatorname{triple}}
\newcommand{\term}{\operatorname{term}}
\renewcommand{\terms}{\operatorname{terms}} 
\newcommand{\iri}{\operatorname{iri}}

\newcommand{\dom}{\operatorname{dom}}

\newcommand{\inScope}{\operatorname{inScope}}
\newcommand{\range}{\operatorname{range}}
\newcommand{\false}{\operatorname{\emph{false}}}
\newcommand{\true}{\operatorname{\emph{true}}}
\newcommand{\error}{\operatorname{\emph{error}}}
\newcommand{\comp}{\operatorname{comp}}

\newcommand{\NotNull}{\operatorname{NotNull}}
\newcommand{\Error}{\operatorname{Error}}
\newcommand{\subs}{\operatorname{subs}}

\newcommand{\NULL}{\mathtt{NULL}}

\newcommand{\SELECT}{\operatorname{SELECT}}

\newcommand{\WHERE}{\operatorname{WHERE}}
\newcommand{\AAND}{\operatorname{AND}}
\newcommand{\UNION}{\operatorname{UNION}}

\newcommand{\OPTIONAL}{\operatorname{OPT}}
\newcommand{\FILTER}{\operatorname{FILTER}}
\newcommand{\MINUS}{\operatorname{MINUS}}

\newcommand{\bound}{\operatorname{bound}}
\newcommand{\EXCEPT}{\operatorname{EXCEPT}}

\def\lojoin{\hbox{\raise -.2em\hbox to-.32em{$\urcorner$} \hbox to-.08em{$\lrcorner$} $\Join$}\,}
\newcommand{\ev}[2]{\llbracket #1 \rrbracket_{#2}}

\newcommand{\var}{\operatorname{var}}

\newcommand{\T}{\mathcal{T}}
\newcommand{\TN}{\mathcal{N}}
\newcommand{\TE}{\mathcal{E}}
\newcommand{\TL}{\mathcal{L}}

\newcommand{\card}{\operatorname{card}} 
\newcommand{\set}{\operatorname{set}}

\newcommand{\Individual}[1]{\mathsf{#1}}
\newcommand{\iAlice}{\Individual{Alice}}
\newcommand{\iBob}{\Individual{Bob}}
\newcommand{\iCarol}{\Individual{Carol}}
\newcommand{\iSantiago}{\Individual{Santiago}}
\newcommand{\iLima}{\Individual{Lima}}
\newcommand{\Predicate}[1]{\mathsf{#1}}
\newcommand{\pKnows}{\Predicate{knows}}
\newcommand{\pLivesIn}{\Predicate{livesIn}}

\newcommand{\iriFor}[1]{\mathit{iri}\_{#1}}
\newcommand{\iriT}[1]{\mathit{iri}\_t^{#1}}
\newcommand{\iriB}{\mathit{iri}\_b}

\newcommand{\iriA}{\mathit{iri}\_A}
\newcommand{\litA}{\mathit{lit}\_A}

\newcommand{\eliminar}[1]{}
\makeatletter
\@ifpackageloaded{stix}{%
}{%
	\DeclareFontEncoding{LS2}{}{\noaccents@}
	\DeclareFontSubstitution{LS2}{stix}{m}{n}
	\DeclareSymbolFont{stix@largesymbols}{LS2}{stixex}{m}{n}
	\SetSymbolFont{stix@largesymbols}{bold}{LS2}{stixex}{b}{n}
	\DeclareMathDelimiter{\lBrace}{\mathopen} {stix@largesymbols}{"E8}%
	{stix@largesymbols}{"0E}
	\DeclareMathDelimiter{\rBrace}{\mathclose}{stix@largesymbols}{"E9}%
	{stix@largesymbols}{"0F}
}
\makeatother

\begin{document}

\title[The multiset semantics of SPARQL patterns]{Multiset semantics in SPARQL, Relational Algebra and Datalog}
\titlenote{This is the authors' preprint of: R. Angles, C. Gutierrez, and D. Hern\'andez. Multiset semantics in SPARQL, Relational Algebra and Datalog. \emph{Semantic Web}, 2026. DOI: \href{https://doi.org/10.1177/22104968261439426}{10.1177/22104968261439426}.}

\author{Renzo Angles}
\orcid{0000-0002-6740-9711}
\affiliation{%
  \institution{Department of Computer Science, Faculty of Engineering, Universidad de Talca}
  \country{Chile}
}
\affiliation{%
  \institution{Millenium Institute for Foundational Research on Data (IMFD)}
  \country{Chile}
}
\email{rangles@utalca.cl}

\author{Claudio Gutierrez}
\orcid{0000-0002-4559-6544}
\affiliation{%
  \institution{Department of Computer Science, Universidad de Chile}
  \country{Chile}
}
\affiliation{%
  \institution{Millenium Institute for Foundational Research on Data (IMFD)}
  \country{Chile}
}
\email{cgutierr@dcc.uchile.cl}

\author{Daniel Hernández}
\authornote{Corresponding author.}
\orcid{0000-0002-7896-0875}
\affiliation{%
  \institution{Institute for Artificial Intelligence, University of Stuttgart}
  \country{Germany}
}
\email{daniel.hernandez@ki.uni-stuttgart.de}

\renewcommand{\shortauthors}{Angles et al.}

\begin{abstract}
The paper analyzes and characterizes the algebraic and logical structure of the multiset semantics for SPARQL patterns involving AND, UNION, FILTER, EXCEPT, and SELECT. To do this, we align SPARQL with two well-established query languages: Datalog and Relational Algebra. Specifically, we study (i) a version of non-recursive Datalog with safe negation extended to support multisets, and (ii) a multiset relational algebra comprising projection, selection, natural join, arithmetic union, and except. We prove that these three formalisms are expressively equivalent under multiset semantics.
\end{abstract}


\maketitle

\section{Introduction}%
\label{sec:introduction}

Informally speaking, multisets are sets in which each element could occur multiple times, that is, the number of ``copies'' of each element matters. In the field of databases, the notion of multisets (also called ``duplicates'' or ``bags'')\footnote{There seems to be no agreement on the best terminology~\cite[p. 27]{Melton2002}. In this paper, we will use the word ``multiset''.} has been studied in several contexts, including programming languages \cite{91640,91641}, bag languages \cite{90818,90191,90190,91631,91629,90189}, relational algebra \cite{91030,91632,91627}, Datalog \cite{90820,91368,91634,90960,91630}, SQL \cite{91635,91642}, {\SPARQL} \cite{91643,91636,91302,91353} and data integration \cite{91639}.

The incorporation of multisets in query languages is essentially due to practical concerns: duplicate elimination is expensive, and duplicates might be required for some applications, e.g., for aggregation. Although this design decision may be debatable (e.g., see~\cite{Date2006}), today multisets are an established reality in database systems~\cite{Green09,Lamperti}.

The classical theory behind declarative query languages includes formalisms (relational algebra or relational calculus) that for sets have a clear and intuitive semantics for users, developers and theoreticians~\cite{Abiteboul-Book}.  The same cannot be said for their extensions to multisets, whose theory is complex (particularly the containment of queries), and their practical use not always clear~\cite{Green09}. Worst, there exist several possible ways of extending set relational operators to multisets, which makes the study and design of multiset semantics for query languages challenging.

\begin{table}
  \caption{\emph{Possible ways of extending set operators with multiset semantics in SQL and {\SPARQL}.}\hspace{0.5em} The table shows several extended relational algebra operations for multisets currently present (or possible to implement) in SQL and {\SPARQL}.
  Let $R$, $S$ and $T$ be multiset relations satisfying that $R$ and $S$ have the same attributes, and $T$ does not have attributes in common with $R$.
  The cardinality of an element $x$ in a relation $R$ is represented as $R(x)$.
  Note that SPARQL works with multisets of bindings, whose corresponding schema is a set of variables.
  }
  \label{multiset-operations-sql-sparql}
  \resizebox{\textwidth}{!}{%
  \begin{tabular}{lllll}
    \toprule
    \textbf{Operation}
    & \textbf{Operator}
    & \textbf{Cardinality for $x$}
    & \textbf{SQL}
    & \textbf{SPARQL} \\
    \midrule
    Selection
    & $\sigma_\varphi(R)$
    & $\left\{
      \begin{array}{ll}
        R(x) & \text{ if }x\text{ satisfies }\varphi, \\
        0 & \text{ otherwise.} \\
      \end{array}\right.$
    & \texttt{SELECT * FROM $R$ WHERE $\varphi$}
    & \texttt{$R$ FILTER ($\varphi$)}
    \\
    \midrule
    Cartesian product
    & $R \times T$
    & $R(x) \times T(x)$ 
    & \texttt{$R$ CROSS JOIN $T$}
    & \texttt{$R$ AND $T$}
    \\
    \midrule
    Join
    & $R \Join_{\varphi} T$
    & $ R(x) \times T(x)$ 
    & \texttt{($R$ CROSS JOIN $T$) WHERE $\varphi$}
    & \texttt{($R$ AND $T$) FILTER ($\varphi$)}
    \\
    \midrule
    Max-union
    & $R \sqcup S$
    & $\max(R(x), S(x))$ 
    & \texttt{($R$ UNION ALL $S$) EXCEPT ALL} 
    & -- 
    \\
    & & & \texttt{($R$ INTERSECT ALL $S$)}
    \\
    \midrule
    Arithmetic union
    & $R \cup S$
    & $R(x) + S(x)$ 
    & \texttt{$R$ UNION ALL $S$}
    & \texttt{$R$ UNION $S$}
    \\ 
    \midrule
    Min-intersection
    & $R \cap S$
    & $\min(R(x), S(x))$
    & \texttt{$R$ INTERSECT ALL $S$}
    & --
    \\ 
    \midrule
    Max-intersection
    & $R \sqcap S$
    & $S(x) \times S(x)$
    & \texttt{$R$ NATURAL JOIN $S$}
    & \texttt{$R$ AND $S$}
    \\ 
    \midrule
    Arithmetic difference
    & $R - S$
    & $\max(0, R(x) - S(x))$
    & \texttt{$R$ EXCEPT ALL $S$}
    & --
    \\
    \midrule
    Existential negation
    & $R \setminus S$
    & $\left\{
      \begin{array}{ll}
        R(x) & \text{ if } S(x) = 0,\\
        0 & \text{ otherwise.}
      \end{array}\right.$
    & 
     $ 
      \begin{array}{l}
        \texttt{SELECT * FROM $R$} \\
        \texttt{WHERE $x$ NOT IN ($S$)} 
      \end{array}
     $
    & \texttt{$R$ MINUS $S$} 
    \\
    \midrule
    Projection
    & $\pi_{Atts}(R)$
    & $\sum_{t \in R,\, t[Atts]=x} R(x)$
    & \texttt{SELECT $Atts$ FROM $R$}
    & \texttt{SELECT $Atts$}
    \\
    \bottomrule
   \\
  \end{tabular}%
  }
\end{table}

To illustrate the variety of possible semantics, we will show the different extensions to multisets of set operators found in the literature.
Consider the following multiset relations:
$R(W,X) = \lBrace (a,b),(a,b),(a,d) \rBrace$, 
$S(W,X) = \lBrace (a,b) \rBrace$ and 
$T(Y,Z) = \lBrace (b,c),(b,c) \rBrace$.
For the first relation, $R$ is the name of the relation, $W$ and $X$ are the attributes that conform the schema of $R$, $R$ contains three tuples, and the tuple $(a,b)$ is duplicated (i.e., its cardinality is $2$). A similar description can be given for the relations $S$ and $T$. Note that $R$ and $S$ have the same attributes, while $T$ does not have attributes in common with $R$ and $S$.

  
\begin{itemize}
\item 
The \emph{selection} returns the tuples satisfying a given condition but keeping cardinalities. 
For example, $\sigma_{X='b'}(R)$ returns the multiset $\lBrace (a,b),(a,b) \rBrace$ with schema $(W,X)$.

\item 
The \emph{cartesian product} results in the multiplication of the cardinalities. 
For example, $R \times T$ returns the multiset $\lBrace$ $(a,b,b,c)$, $(a,b,b,c)$, $(a,b,b,c)$, $(a,b,b,c)$, $(a,d,b,c)$, $(a,d,b,c)$ $\rBrace$ with schema $(W,X,Y,Z)$.

\item 
The \emph{join} results in the multiplication of the cardinalities, as it is expressed as a cartesian product followed by a selection. 
For example, $R \Join_{X = Y} T$ returns the multiset $\lBrace$ $(a,b,b,c)$, $(a,b,b,c)$, $(a,b,b,c)$, $(a,b,b,c)$ $\rBrace$ with schema $(W,X,Y,Z)$.

\item 
The \emph{max-union} takes the maximum number of occurrences of an element. 
For example, $R \sqcup S$ returns the multiset $\lBrace$ $(a,b)$, $(a,b)$, $(a,d)$ $\rBrace$ with schema $(W,X)$.

\item 
The \emph{arithmetic union} adds up cardinalities. 
For example, $R \cup S$ returns the multiset $\lBrace$ $(a,b)$, $(a,b)$, $(a,d)$, $(a,b)$ $\rBrace$ with schema $(W,X)$.

\item
The \emph{min-intersection} takes the minimum number of occurrences of each element in the intersection. 
For example, $R \cap S$ returns the multiset $\lBrace$ $(a,b)$ $\rBrace$ with schema $(W,X)$.

\item
The \emph{max-intersection} returns the product of the cardinalities of each element in the intersection.
For example, $R \sqcap S$ returns the multiset $\lBrace$ $(a,b)$, $(a,b)$ $\rBrace$ with schema $(W,X)$.

\item 
The \emph{arithmetic difference} subtracts the cardinalities of the elements up to zero. 
For example $R - S$ returns the multiset $\lBrace$ $(a,b)$, $(a,d)$ $\rBrace$ with schema $(W,X)$.

\item 
The \emph{existential negation} returns the elements in the first multiset that do not occur in the second one,
and keeping the cardinalities of such elements in the first multiset. 
For example, the expression $R \setminus S$ returns the multiset $\lBrace$ $(a,d)$ $\rBrace$ with schema $(W,X)$.

\item 
The \emph{projection} reduces the number of attributes in each tuple, and gives rise to new cardinalities for the resulting tuples.
For example $\pi_{W} (R)$ returns the multiset $\lBrace$ $(a)$, $(a)$, $(a)$ $\rBrace$ with schema $(W)$.

\end{itemize}


Table \ref{multiset-operations-sql-sparql} shows a summary of the above operators, and their corresponding implementation in SQL and {\SPARQL}. Note that SQL can express all the operators, whereas {\SPARQL} does not support max-union, min-intersection, and arithmetic difference.
Also note that SPARQL uses the AND operator to implement cartesian product and max-intersection. The first case occurs because $R$ and $T$ do not have variables in common, and the second case occurs because $R$ and $S$ have the same set of variables. 

The landscape of operators over multisets poses important challenges for integrating multisets in query languages.
First, as shown in Table \ref{multiset-operations-sql-sparql}, some operators exhibit different semantics when applied to multisets. 
Second, while relational algebra and SQL support all the semantics listed, SPARQL and Datalog only support a subset. 
Third -- and this is the main motivation for our research -- it remains unclear whether there exists an optimal set of multiset operators for SPARQL, and if so, which one it is. To tackle these questions, it is essential to understand how formalisms that are ``closed'' with respect to SPARQL behave, and how their design and behavior can inform or be translated into SPARQL. In technical terms, this means analyzing the expressive power of SPARQL regarding multisets. To this end, we focus on two natural and well-studied reference points: relational algebra and Datalog. That is the aim of this article. Next, we review the existing literature about multisets. 


\paragraph{Related Work.}
First, we consider the research works that define general algebras for manipulating bags.
Albert~\cite{90818} extended typical set operations (union, intersection, difference, and boolean selection) to bags, and demonstrated that some of the algebraic properties for sets fail for multisets.  
Grumbach et al.~\cite{90190} introduced a bag algebra, called BALG, that extends relational operations to handle duplicates. This paper shows that BALG is more expressive than standard relational algebra because it can count duplicates, but it still has low data complexity (LOGSPACE).
Grumbach and Milo~\cite{91631} focused on designing bag algebras that are both expressive and computationally tractable. They introduce restricted forms of projection and join to maintain tractable data complexity.
Libkin and Wong~\cite{90191,90189} introduced BQL, a query language for handling bags and aggregate functions (sum, count, avg.). They show that BQL is more expressive than traditional set-based languages, and shows that after incorporating structural recursion to BQL, it is able to express all primitive recursive functions, significantly increasing its computational power.
Ricciotti and Cheney~\cite{91642} explored how to mix set and bag semantics in query languages, addressing practical needs found in SQL (e.g., SELECT versus SELECT DISTINCT). They propose a formal model that supports both semantics and allows translation between them. 



The first attempt to extend the relational algebra to include multisets was made by Dayal et al.~\cite{91030}. In this work, the authors introduced a multiset relational algebra (formed by the operators of projection, selection, join, max-union, arithmetic union, min-intersection and arithmetic difference) and studied their algebraic properties. This work laid the groundwork for formalizing bag semantics in relational query languages.
%
Klauser and Goodman~\cite{91632} provided a semantic framework for understanding the role of multirelations (relations with duplicates) at the conceptual level. The authors explain how any query language can be extended consistently to have full multirelational expressiveness. 
Afrati et al.~\cite{90960} studied query containment in relational databases under bag semantics and bag-set semantics (duplicates allowed in intermediate steps but not in final output). The authors identify conditions under which containment is decidable and provide complexity results.
Console et al.~\cite{91627} investigated fragments of bag relational algebra, focusing on their expressive power. The authors also study query answering over bags with nulls (i.e. under incomplete data).

Multisets have also been the subject of study in the context of Datalog, with various extensions proposed to support bag semantics.
Mumick et al.~\cite{91368} defined the Magic Sets transformation for optimizing recursive queries, and described how to adapt this technique to support duplicates. They also showed how to efficiently evaluate recursive queries under multiset semantics. 
In a subsequent work~\cite{90820}, Mumick et al. extended the Magic Sets technique to support duplicates and aggregate functions in recursive queries. The authors also studied the challenges of preserving correct bag semantics when applying recursion and aggregation.
Cohen~\cite{91634} studied the problem of query equivalence under bag semantics. This work includes complexity results and demonstrates that equivalence checking is significantly harder under bag semantics.
Bertossi et al.~\cite{91368} developed a translation of Datalog under bag semantics into warded Datalog$^\pm$, a well-behaved extension under set semantics. The authors investigated the properties of the resulting Datalog$^\pm$ programs, the problem of deciding multiplicities, and expressibility of some bag operations.

For SPARQL -- the standard query language for RDF databases -- Pérez et al.~\cite{10145} provided the first formal treatment of its multiset semantics.
This work influenced the definition of SPARQL 1.0 \cite{10155} and SPARQL 1.1 \cite{90699}, whose semantics are based on operations over multisets of mappings (although a database is a set of RDF triples).
Schmidt et al.~\cite{90380} presented a formal framework for SPARQL query optimization, addressing both set and bag semantics. The authors analyzed the algebraic properties of SPARQL operations like OPTIONAL, UNION, and FILTER under multisets, and introduced equivalence rules and normal forms for optimizing queries. 
Kaminski et al.~\cite{Kaminski} presented a formal investigation of subqueries and aggregate functions in SPARQL 1.1, focusing on their semantics under multisets. The authors analyzed the expressive power of these constructs, showing that SPARQL 1.1 is strictly more expressive than SPARQL 1.0 due to these features. 

Finally, we review research articles that present comparisons and translations among SPARQL, relational algebra, and Datalog.
Cyganiak~\cite{10140} was among the first to translate a core fragment of SPARQL into relational algebra. 
Polleres~\cite{10150} proved the inclusion of the fragment of {\SPARQL} patterns with safe filters into Datalog by providing a precise and correct set of rules.
Schenk~\cite{10158} proposed a formal semantics for {\SPARQL} based on Datalog, but concentrated on complexity more than expressiveness issues. Both Polleres and Schenk did not consider the multiset semantics of {\SPARQL} in their translations.  
Angles and Gutierrez~\cite{90006} studied the expressive power of SPARQL by providing a translation to non-recursive safe Datalog with negation.
Chebotko et al.~\cite{90827} addressed the problem of translating SPARQL queries into SQL while preserving bag semantics. The authors proposed a formal translation framework that captures the subtleties of OPTIONAL, UNION, and FILTER, and ensures that duplicates in the result sets are handled correctly when mapped to relational databases.
Angles and Gutierrez \cite{91353} studied the multiset semantics of SPARQL patterns by translating its patterns into two languages: a version of multiset relational
algebra and multiset non-recursive Datalog with safe negation.
Angles et al. \cite{sparqlog} implemented the translation from {\SPARQL} to Datalog within the Vadalog system \cite{vadalog}. 




\paragraph{Objectives and Contributions.}


The main objective of this article is to examine the theoretical foundations of SPARQL's multiset semantics. To do so, we compare it with classical algebraic and logical frameworks -- specifically, Relational Algebra and Datalog. We focus on the SPARQL fragment built from AND, UNION, FILTER, EXCEPT, and SELECT, characterizing its structure and proving its expressive equivalence with corresponding fragments of Relational Algebra and Datalog.

The specific contributions of our research are as follows:

\paragraph{(1)}
Based on the work of Mumick et al~\cite{90820}, who defined the multiset semantics for Datalog without negation, we defined a version called {\em Non-Recursive Multiset Datalog with Safe Negation} ({\NRMDatalogSN}).
The definition of \NRMDatalogSN includes negation and follows a proof-theoretic semantics.


\paragraph{(2)} 
Based on the work of Dayal et al.~\cite{91030}, who extended the relational algebra to include multiset relations, we defined a {\em Multiset Relational Algebra} ({\MRA}).
The definition of {\MRA} includes the operators of projection ($\pi$), selection ($\sigma$), natural join ($\Join$), arithmetic union ($\cup$) and filter difference ($\setminus$), all of them working under multiset semantics.


\paragraph{(3)} 
We show the equivalence among the aforementioned {\SPARQL} fragment, {\MRA} and {\NRMDatalogSN} by providing translations for databases, queries, and answers. 
Table~\ref{table:equivalences} shows a glimpse of these translations, whose details are developed in this paper.


\medskip

This paper extends a previously published conference paper~\cite{91353}.  Herein, we provide extended discussion throughout, we extend the study to some operators that were introduced in the version 1.1 of {\SPARQL} after the publication of our previous work, and we extend the analysis to also consider bag semantics. Some of the additional contributions of this paper come from Hernandez's Ph.D. thesis~\cite{phd-hernandez}.

The rest of the article is organized as follows.  Section~\ref{sec:preliminaries} presents basic concepts and notations. 
The {\SPARQL} query language is defined in  Section~\ref{sec:sparql}.
Non-recursive Multiset Datalog with Safe Negation ({\NRMDatalogSN}) is defined in Section~\ref{sec:datalog}.
The Multiset Relational Algebra ({\MRA}) is defined in Section~\ref{sec:mra}.
The equivalence between {\SPARQL} and {\NRMDatalogSN} is presented in Section~\ref{sec:sparql-nrmd}.
The equivalence between {\MRA} and {\NRMDatalogSN} is presented in  Section~\ref{sec:mra-nrmd}.
The equivalence between {\MRA} and {\SPARQL} is presented in Section~\ref{sec:mra-sparql}.
Conclusions are presented in Section~\ref{sec:conclusions}.

\newcommand{\tetoprule}{%
\cmidrule[\heavyrulewidth](r){1-2}
\cmidrule[\heavyrulewidth](l){4-6}}
\newcommand{\temidrule}{%
\cmidrule[0.2pt](r){1-2}
\cmidrule[0.2pt](l){4-6}}
\newcommand{\tebotrule}{%
\cmidrule[\heavyrulewidth](r){1-2}
\cmidrule[\heavyrulewidth](l){4-6}}

\begin{table}[t]
\caption{{\sc Schema of correspondences among:} {\SPARQL} graph patterns, Multiset Relational Algebra ({\MRA}) expressions, Non-Recursive Datalog with safe Negation ({\NRMDatalogSN}) rules, and SQL expressions. 
The operator $\EXCEPT$ is not part of {\SPARQL}, but it replaces the standard operators $\MINUS$ and $\OPTIONAL$ without changing the expressiveness of the fragment. In {\MRA}, $\uplus$ is the arithmetic union and $\setminus$ is the multiset filter difference. {\SPARQL} patterns are assumed normalized, that is, variables in the filter condition are in the schema of the filtered pattern, and operators $\EXCEPT$ AND $\UNION$ assume operands with the same schema. $P_1$ and $P_2$ are {\SPARQL} patterns that are associated to atoms $L_1$ and $L_2$ in the {\NRMDatalogSN} translation, and relations $r_1$ and $r_2$ in the {\MRA} translations, respectively.}
\label{table:equivalences}
\resizebox{\textwidth}{!}{%
\begin{tabular}{llll}
  \toprule
  {\SPARQL}
  & {\NRMDatalogSN}
  & {\MRA}
  & SQL
  \\ \toprule
  $\SELECT \mathcal{X} \; P_1$
  & $L \gets
    L_1,\unull(\mathcal{X}\setminus\mathcal{X}_1)$
  & $\pi_{\mathcal{X}}(r_1) \Join
    \unull(\mathcal{X}\setminus\mathcal{X}_1)$
  & $\begin{aligned}[t]
    &\mathtt{SELECT}\;\mathcal{X}\\[-6pt]
    &\mathtt{FROM}\;r_1
    \;\mathtt{NATURAL\;JOIN}\;
    \unull(\mathcal{X}\setminus\mathcal{X}_1)
    \end{aligned}
    $
  \\ \midrule
  $P_1 \FILTER \,X=a$
  & $L \gets L_1,X=a$
  & $\sigma_{X=a}(r_1)$
  & {\tt FROM $r_1$ WHERE  $X=a$ }
  \\ \midrule
  $P_1 \AAND P_2$
  & $\begin{aligned}[t]
    L \gets &v_1(L_1), v_2(L_2),\\[-6pt]
    &\comp(v_1,v_2,\mathcal{X})
    \end{aligned}
    $
  & $\begin{aligned}[t]
    \pi_{\bar{\mathcal{X}}}(
    \rho_{v_1}(r_1) \Join
    \rho_{v_2}(r_2) \Join \\[-6pt]
    \comp(v_2,v_2,\mathcal{X}))
    \end{aligned}
    $
  & $\begin{aligned}[t]
    &\mathtt{SELECT}\; \mathcal{X}\\[-6pt]
    &\begin{aligned}[t]
      \mathtt{FROM}\; &r_1\;
      \mathtt{NATURAL\;JOIN}\; r_2\;
      \mathtt{NATURAL\;JOIN}\\[-6pt]
      &\comp(v_2,v_2,\mathcal{X})
      \end{aligned}
    \end{aligned}
    $
  \\ \midrule
  $P_1 \UNION P_2$
  & $L \gets L_1\,;\; L \gets L_2$
  & $r_1 \uplus r_2$
  & {\tt $r_1$ UNION  ALL  $r_2$}
  \\ \midrule
  $P_1 \EXCEPT P_2$
  & $L \gets L_1, \neg L_2$
  & $r_1 \setminus r_2$
  & {\tt $r_1$ EXCEPT  $r_2$}
  \\ \bottomrule
\end{tabular}%
}
\end{table}

\section{Preliminaries}
\label{sec:preliminaries}

This section provides the concepts and formal notation we will follow regarding multisets
and the expressive power of query languages.

\subsection{Multisets}

Informally, a multiset is an unordered collection of elements where each element may occur more than once.  Formally, a \emph{multiset} is a tuple $M = (S,\mathit{card})$ where $S$ is the underlying set of $M$ (containing the distinct elements), and $\mathit{card} : S \to \mathbb{N}^+$ is a function that defines the cardinality in $M$ of each element $a \in S$.  
We write $\set(M) = S$ to denote that the underlying set of $M$ is $S$. Given a positive natural number $n$, $\card(a,M) = n$ denotes $a \in \set(M)$ and the cardinality of $a$ in $M$ is $n$, and is usually written as $(a,n) \in M$.
Abusing notation, we write $\card(a,M)=0$ if $a\notin\set(M)$ and $a \in M$ when $\card(a,M) \ge 1$.
In what follows, we will prefer these formal notions over the informal and intuitive $\lBrace a, a, a, b \rBrace$.

When dealing with multisets, formally describing each ``copy'' of an element is challenging. The notion of \emph{colored set} \cite{90820} is a formalism to do this. Indeed, assuming that the set of colors is $\mathbb{N}^+$, the \emph{colored set} of a multiset $M$, denoted $\col(M)$, is the set $\{ \colCopy{a}{i} \mid a \in \set(M) \mbox{ and } 1 \leq i \leq \card(a,M) \}$. In the contrary direction, we write $\col^{-1}(C)$ for the multiset $M$ with $\set(M)=\set(C)$ and defined by a colored set $C$ when forgetting the colors, that is, when for every element $a \in \set(M)$, $\card(a,M)$ is the number of colored copies of $a$ in $C$. Abusing notation, we write $\col^{-1}(\colCopy{a}{i}) = a$.

\begin{example}
Let $A$ be the set $\{a,b,c\}$, $M$ be a multiset with $\set(M) = A$, $\card(a,M) = 1$, $\card(b,M) = 2$, and $\card(c,M) = 3$. Then, $\col(M)$ is the colored set
\[
C = \{ \colCopy{a}{1}, \colCopy{b}{1}, \colCopy{b}{2}, \colCopy{c}{1}, \colCopy{c}{2}, \colCopy{c}{3}\}.
\]
\end{example}

\subsection{Comparing the expressive power of query languages}%
\label{sec:comparing-expressive-power}

Next, we present the notion of query language and two notions of expressive power used in this paper.

\begin{definition}[Query language]
  A query language $\calL$ is a quadruple $(\calQ, \calD, \calS, \Eval)$, where $\calQ$ is the set of queries in $\calL$, $\calD$ is the set of databases in $\calL$, $\calS$ is the set of query answers in $\calL$, and $\Eval: \calQ \times \calD \to \calS$ is the query evaluation function of $\calL$.
\end{definition}

Let $\calL = (\calQ, \calD, \calS, \Eval)$ be a query language.  Two queries $Q_1, Q_2 \in \calQ$ are said to be \emph{equivalent}, denoted $Q_1 \equiv Q_2$, if for every database $D \in \calD$, it holds that $\Eval(Q_1,D) = \Eval(Q_2,D)$, i.e., they return the same query answer for all input databases.

Given a query language $(\calQ, \calD, \calS, \Eval)$, a query $Q \in \calQ$ determines a function $q:\calD\to\calS$ defined as $q(D)=\Eval(Q,D)$, called the {\em query function} of $Q$. Two queries $Q_1$ and $Q_2$ are thus equivalent, denoted $Q_1 \equiv Q_2$, if they determine the same query function.

In this context, the {\em expressive power} of a query language $\calL$ is understood as the set of all query functions that are expressible by $\calL$.  Abiteboul et al.~\cite{Abiteboul-Book} summarizes how this notion is used to compare the expressive power of relational algebra, Datalog, and relational calculus.  In the context of SPARQL, Zhang and Van den Bussche~\cite{DBLP:journals/ipl/ZhangB14}, Kontchakov et al.~\cite{91309}, and Angles and Gutierrez~\cite{91301} use this notion to compare different fragments of SPARQL.

The query languages studied in this paper do not satisfy the aforementioned property of having a common set of databases and query answers. Thus, we need an extended version of the notion of expressive power as in Definition~\ref{def:expressive-power} below.  

\begin{definition}[Generalized expressive power]%
  \label{def:expressive-power}
  Given two query languages $\calL_1 = (\calQ_1, \calD_1, \calS_1, \Eval_1)$ and $\calL_2 = (\calQ_2, \calD_2, \calS_2, \Eval_2)$, we say that $\calL_1$ is contained in $\calL_2$ if and only if there exist functions $g:\calD_1 \to \calD_2$ (called the \emph{database translation}), $f:\calQ_1 \to \calQ_2$ (called the \emph{query translation}), and $h:\calS_2 \to \calS_1$ (called the \emph{query answer translation}), such that for every $Q \in \calQ_1$ and database $D \in \calD_1$ it holds that
  \begin{align*}
    \Eval_1(Q, D) = h(\Eval_2(f(Q), g(D))).
  \end{align*}
  If that is the case, we say that the triple $(f,g,h)$ is a \emph{simulation} of $\calL_1$ in $\calL_2$. We say that the languages $\calL_1$ and $\calL_2$ have the same expressive power, denoted $\calL_1 \cong \calL_2$, if and only if $\calL_1$ is contained in $\calL_2$ and $\calL_2$ is contained in $\calL_1$.
\end{definition}

The above definition of generalized expressive power is implicit in the translations by Polleres~\cite{10150}, Angles and Gutierrez~\cite{90006,91353}, and Polleres and Wallner~\cite{91643}.

Observe that the extended notion defined above defines a partial order: the containment relation on the equivalence classes over the relation $\cong$.  In fact, reflexivity and antisymmetry follow directly from the definition, while transitivity is shown in Figure~\ref{fig:transitivity-containment}.

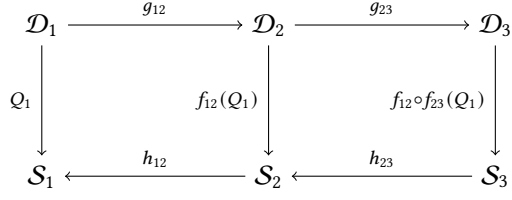
\begin{figure}[t]
  \begin{tikzpicture}
    \node (d1) at (0,2) {$\calD_1$};
    \node (d2) at (3,2) {$\calD_2$};
    \node (d3) at (6,2) {$\calD_3$};
    \node (s1) at (0,0) {$\calS_1$};
    \node (s2) at (3,0) {$\calS_2$};
    \node (s3) at (6,0) {$\calS_3$};
    \path (d1) edge [->]
    node[left] {${\scriptstyle Q_1}$} (s1);
    \path (d2) edge [->]
    node[left] {$\scriptstyle f_{12}(Q_1)$} (s2);
    \path (d3) edge [->]
    node[left] {$\scriptstyle f_{12} \circ f_{23}(Q_1)$} (s3);
    \path (d1) edge [->]
    node[above] {${\scriptstyle g_{12}}$} (d2);
    \path (d2) edge [->]
    node[above] {${\scriptstyle g_{23}}$} (d3);
    \path (s3) edge [->]
    node[above] {${\scriptstyle h_{23}}$} (s2);
    \path (s2) edge [->]
    node[above] {${\scriptstyle h_{12}}$} (s1);
  \end{tikzpicture}
  \caption{Transitivity of language containment. The figure represents three languages $\calL_i=(\calQ_i,\calD_i,\calS_i,\Eval_i)$ where $i\in\{1,2,3\}$. The containment of a language $\calL_i$ in $\calL_{i+1}$ is given by the simulation $(f_{i,i+1},\, g_{i,i+1},\, h_{i,i+1})$. The transitive containment of $\calL_1$ in $\calL_3$ is given by the simulation $(f_{12}\circ f_{23},\, g_{12}\circ g_{23},\, h_{23}\circ h_{12})$ where $\circ$ denotes the composition of functions (e.g., $g_{12}\circ g_{23}$ denotes the function from $D_1$ to $D_3$ that results from composing $g_{12}$ and $g_{23}$).}
  \label{fig:transitivity-containment}
\end{figure}

\subsection{\texorpdfstring{Comparing {\SPARQL}, {\NRMDatalogSN} and {\MRA}}{Comparing SPARQL, NRMD¬ and MRA}}

In the remainder of this paper, we define three families of query languages: Non-recursive Multiset Datalog with Safe Negation ({\NRMDatalogSN}), Multiset Relational Algebra ({\MRA}) and a core fragment of SPARQL. After defining these languages, we present simulations that show the equivalence among these three families of query languages.  These simulations are depicted in Figure~\ref{fig:triangle}.

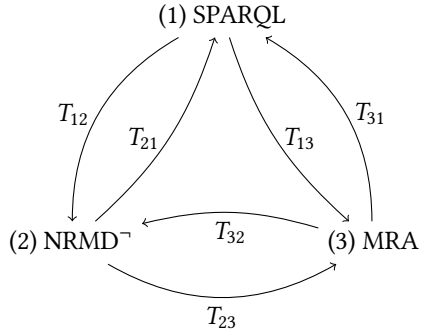
\begin{figure}[h!]
  \centering
  \begin{tikzpicture}
    \node (sparql) at (0, 4) {(1) SPARQL};
    \node (md) at (-2, 1) {(2) {\NRMDatalogSN}};
    \node (mra) at (2, 1) {(3) {\MRA}};
    \path [bend right,->] (sparql) edge node[left] {$T_{12}$} (md);
    \path [bend right=15,->] (md) edge node[left] {$T_{21}$} (sparql);
    \path [bend right,->] (md) edge node[below] {$T_{23}$} (mra);
    \path [bend right=15,->] (mra) edge node[below] {$T_{32}$} (md);
    \path [bend right,->] (mra) edge node[right] {$T_{31}$} (sparql);
    \path [bend right=15,->] (sparql) edge node[right] {$T_{13}$} (mra);
  \end{tikzpicture}
  \caption{The triangle of simulations among {\SPARQL},
  Non-Recursive Multiset Datalog with Safe Negation ({\NRMDatalogSN}), and Multiset Relational Algebra ({\MRA})
  described in this paper. The query languages are identified by
  numbers, and $T_{ij}$ denotes the simulation of
  language $i$ using language $j$.}
  \label{fig:triangle}
\end{figure}

\section{Multiset {\SPARQL}}
\label{sec:sparql}

{\SPARQL} \cite{10155,90699} is the standard query language for RDF.  In this paper we study a fragment of {\SPARQL}, the ``relational core'', described by Angles and Gutierrez~\cite{91353}, which considers the operators FILTER, SELECT, AND, UNION, and EXCEPT. This fragment captures essentially the graph pattern queries in {\SPARQL}. In fact, it has been proved \cite{91353,91309} that it is mutually expressible with the standard-core consisting of the operators FILTER, SELECT, AND, UNION, OPTIONAL, and MINUS. (In what follows when speaking of ``SPARQL'' we will mean this fragment).

\subsection{RDF Graphs}
Assume two disjoint infinite sets $\sI$ and $\sL$, called IRIs and literals, respectively. An \emph{RDF term} is an element in the set $\sT = \sI \cup \sL$.  An \emph{RDF triple} is a triple $(s,p,o) \in \sI \times \sI \times \sT$ where $s$ is called the \emph{subject}, $p$ is called the \emph{predicate} and $o$ is called the \emph{object}.  An \emph{RDF graph} (just graph from now on) is a set of RDF triples.
Given a graph $G$, the function $\terms(G)$ returns the RDF terms that occur in $G$, that is, 
$\terms(G)$ = 
$\{ s \mid (s,p,o) \in G \}$ $\cup$ 
$\{ p \mid (s,p,o) \in G \}$ $\cup$
$\{ o \mid (s,p,o) \in G \}$
The \emph{union} of graphs, $G_1 \cup G_2$, is the theoretical union of their sets of triples.

A {\SPARQL} database will be a set of RDF triples. 

\paragraph{Note:}
In addition to $\sI$ and $\sL$, {\SPARQL} admits as terms anonymous resources called blank nodes. In this paper, we do not include them to help focus on the issues arising from multisets. Avoiding blank nodes does not affect the results presented in this paper. Indeed, in {\SPARQL}, blank nodes in the data can be consistently replaced by IRIs and produce equivalent query results, and blank nodes in queries can be replaced by fresh variables without changing the semantics of the query~\cite{91033}.

\subsection{{\SPARQL} Syntax}%
\label{sec:sparql-syntax}
Assume the existence of an infinite set $\sV$ of variables disjoint from $\sT$ (RDF terms).  A \emph{filter condition} is defined recursively as follows: (i) If $\varX,\varY \in \sV$ and $c \in \sT$ then $(\varX = c)$, $(\varX = \varY)$ and $\bound(\varX)$ are atomic filter conditions; (ii) If $\varphi_1$,$\varphi_2$ are filter conditions then $(\varphi_1 \land \varphi_2)$,$(\varphi_1 \lor \varphi_2)$ and $\neg \varphi_1$ are complex filter conditions. We denote by $\var(\varphi)$ the 
set of variables occurring in $\varphi$.
  
A {\SPARQL} \emph{pattern} is defined recursively as follows:
\begin{itemize}
\item
  A triple from $(\sI \cup \sV) \times (\sI \cup \sV) \times (\sI \cup \sL \cup \sV)$ is a pattern called a \emph{triple pattern}. We will assume that a triple pattern has at least one variable. 
\item
  If $P_1$ and $P_2$ are patterns then $( P_1 \AAND P_2 )$, $( P_1 \UNION P_2 )$, and $( P_1 \EXCEPT P_2 )$ are patterns.
\item
  If $P$ is a pattern and $\varphi$ is a filter condition then $(P \FILTER \varphi)$ is a pattern.
\item
  If $W$ is a set of variables and $P_1$ is a pattern then $(\SELECT W\, P_1)$ is a pattern.
\end{itemize}

\subsection{{\SPARQL} Semantics}
A \emph{solution mapping} (or just \emph{mapping} from now on) is a partial function $\mu : \sV \to \sT$ where the domain of $\mu$, denoted $\dom(\mu)$, is the subset of $\sV$ where $\mu$ is defined.  We write $\mu_\emptyset$ to denote the mapping with empty domain (i.e., $\dom(\mu_\emptyset)=\emptyset$).  Given $\varX \in \sV$ and $c \in \sT$, we write $\mu(\varX) = c$ to denote that $\mu$ maps the variable $\varX$ to the term $c$.  Given a finite set of variables $W$, the restriction of a mapping $\mu$ to $W$, denoted $\mu_{|W}$, is a mapping $\mu'$ that satisfies $\dom(\mu') = W \cap \dom(\mu)$ and $\mu'(\varX) = \mu(\varX)$ when $\varX \in \dom(\mu')$. 
Two solution mappings $\mu_1, \mu_2$ are \emph{compatible}, denoted $\mu_1 \sim \mu_2$, when for all $\varX \in \dom(\mu_1) \cap \dom(\mu_2)$ they satisfy $\mu_1(\varX)=\mu_2(\varX)$, that is, when $\mu_1 \cup \mu_2$ is also a mapping.  Note that two mappings with disjoint domains are always compatible.

Let $\Omega$ be a multiset of solution mappings.  The \emph{domain of variables} in $\Omega$, denoted $\dom(\Omega)$, is defined as the set union of the domains of the variables that occur in the solution mappings of $\Omega$.  Given a mapping $\mu$, the cardinality of $\mu$ in $\Omega$ will be denoted as $\card(\mu,\Omega)$.  If $\mu \notin \Omega$ then $\card(\mu,\Omega) = 0$.

The evaluation of a filter condition $\varphi$ under a mapping $\mu$, denoted $\mu(\varphi)$, is defined in a three-valued logic with values $\true$, $\false$ and $\error$.  We say that $\mu$ satisfies $\varphi$ when $\mu(\varphi) = \true$.  The semantics of $\mu(\varphi)$ is defined recursively as follows:

\begin{itemize}
\item
  If $\varphi$ is $\varX = c$ and $c \in \sT$, then: (a) If $\varX \in \dom(\mu)$ then $\mu(\varphi) = \true$ when $\mu(\varX) = c$ and $\mu(\varphi) = \false$ otherwise; (b) If $\varX \notin \dom(\mu)$ then $\mu(\varphi) = \error$.  
\item
  If $\varphi$ is ${\varX} = {\varY}$ and $\varX,\varY \in \dom(\mu)$, then $\mu(\varphi) = \true$ when $\mu(\varX) = \mu(\varY)$, and $\mu(\varphi) = \false$ otherwise. If $\varX \notin \dom(\mu)$ or $\varY \notin \dom(\mu)$ then $\mu(\varphi) = \error$.
\item
  If $\varphi$ is $\bound(\varX)$ and $\varX \in \dom(\mu)$ then $\mu(\varphi) = \true$; otherwise $\mu(\varphi) = \false$.
\item
  If $\varphi$ is a complex filter condition, then it is evaluated following the three valued logic shown in Table~\ref{table:3-valued}.
\end{itemize}

The evaluation of a pattern $P$ on a graph $G$ is defined as a function $\ev{P}{G}$, which returns a multiset of mappings.  Let $P_1,P_2$ be {\SPARQL} patterns, $\varphi$ be a filter condition and $W$ be a set of variables. For simplicity of reading, denote $M = \ev{P}{G}$, $M_1= \ev{P_1}{G}$, and $M_2 = \ev{P_2}{G}$.  The evaluation $\ev{P}{G}$ is defined recursively as follows:

\begin{itemize}
\item
  If $P$ is a triple pattern $t$ then $\set(M) = \{ \mu \mid \dom(\mu) = \var(t),\, \mu(t) \in G \}$, where $\mu(t)$ is the triple obtained by replacing the variables in $t$ according to $\mu$, and $\card(\mu,M) = 1$.
\item
  If $P$ is $(P_1 \AAND P_2)$ then
  $\set(M) = \{
    \mu_1 \cup \mu_2 \mid\; \mu_1 \in M_1,\, \mu_2 \in M_2,
    \text{and }\mu_1 \sim \mu_2 \}$ and
  $\card(\mu,M) =  \sum_{\mu = \mu_1 \cup \mu_2} \card(\mu_1,M_1)
  \times \card(\mu_2,M_2)$.
\item
  If $P$ is $(P_1 \UNION P_2)$ then
  $\set(M) = \{ \mu \mid \mu \in M_1 \lor \mu \in M_2 \}$ and
  $\card(\mu,M) = \card(\mu,M_1) + \card(\mu,M_2)$.
\item
  If $P$ is $(P_1 \EXCEPT P_2)$ then
  $\set(M) = \{ \mu \mid \mu \in M_1,\, \mu \notin M_2 \}$ and
  $\card(\mu,M) = \card(\mu,M_1)$.
\item
  If $P$ is $(P_1 \FILTER \varphi)$ then
  $\set(M) = \{ \mu \mid \mu \in M_1,\, \mu(\varphi)=\true\}$ and
  $\card(\mu,M) = \card(\mu,M_1)$.
\item
  If $P$ is $(\SELECT W\;P_1)$ then\\ 
  $\set(M) = \{\mu' \mid \mu' = \mu_{|W} \land \mu \in M_1 \}$ and\\
  $\card(\mu',M) = \sum_{\mu' = \mu_{|W}} \card(\mu,M_1)$.
\end{itemize}

\begin{table}[t]
  \caption{Evaluation of complex filter conditions
  \cite[\S 17.2]{10155}, where $\mu$ is a solution mapping, and $\varphi_1$,$\varphi_2$ are filter conditions.}
  \label{table:3-valued}
  \begin{tabular}{cccc}
    \toprule
    $\mu(\varphi_1)$
    & $\mu(\varphi_2)$
    & $\mu(\varphi_1) \land \mu(\varphi_2)$
    & $\mu(\varphi_1) \lor \mu(\varphi_2)$
    \\ \midrule

    true & true & true & true \\
    true & false & false & true \\
    true & error & error & true \\
    false & true & false & true \\
    false & false & false & false \\
    false & error & false & error \\
    error & true & error & true \\
    error & false & false & error \\
    error & error & error & error \\
    \bottomrule
  \end{tabular}
  \hspace{4em}
  \begin{tabular}{cc}
    \toprule
    $\mu(\varphi_1)$ & $\neg(\mu(\varphi_1))$ \\ \midrule
    true & false \\
    false & true \\
    error & error \\
    \bottomrule
  \end{tabular}
\end{table}

To facilitate the translation from {\SPARQL} to relational algebra and Datalog, we use the difference operator $\EXCEPT$ in {\SPARQL}, called SetMinus by Kontchakov et al.~\cite{91309}.  Kontchakov et al.~\cite{91309} proved that, over this fragment, the operator $\EXCEPT$ and the pair of standard operators $\{\MINUS,\operatorname{OPTIONAL}\}$ are mutually expressible.

\subsection{Normalization of SPARQL patterns}%
\label{sec:normalization-of-patterns}
The solution mappings of a {\SPARQL} pattern $P$ may have different domains. To translate {\SPARQL} to languages built upon relations, we require representing multisets of mappings as relations whose tuples have the same set of attributes.  This set of attributes has to contain all variables that can appear in the solution mappings of $P$.  The {\SPARQL} specification~\cite{90699} defines a finite set of variables, called {\em in-scope}, that include all variables of a {\SPARQL} pattern $P$ that can occur in the solution mappings of $P$.  To complete the relation, unbound values need to be denoted with a distinguished constant of the target languages.

\begin{example}%
  \label{ex:representing-mappings-as-relation}
  Assume a pattern $P$ with in-scope variables $\varX$, $\varY$, and $\varZ$ that returns the multiset of mappings $\Omega = \lBrace \{\varX \mapsto a\}, \{\varX \mapsto b, \varY \mapsto c\}, \{\varY \mapsto d\} \rBrace$.  Since all variables in the solution mappings are ensured to be in-scope variables of $P$, we can represent this multiset of mappings as the following relation ($\bot$ denotes the distinguished constant to denote unbound values):
  \begin{center}
  $
    \left[
    \begin{array}{ccc}
      \varX & \varY & \varZ \\ \midrule
      a & \bot & \bot \\
      b & c & \bot \\
      \bot & d & \bot \\
    \end{array}
    \right].
    $
  \end{center}
\end{example}

In-scope variables are defined as follows.  Let $P_1$, $P_2$ and $P_3$ be patterns, $\varphi$ be a filter condition, and $W$ be a set of variables.  The set of \emph{in-scope variables} of a pattern $P$, denoted $\inScope(P)$, is defined recursively as follows:
\begin{enumerate}
\item
  If $P$ is a triple pattern then 
$\inScope(P)$ is the set of variables occuring in $P$.
\item
  If $P$ is
  $(P_1 \AAND P_2)$ or $(P_1 \UNION P_2)$
  then $\inScope(P)=\inScope(P_1) \cup \inScope(P_2)$;
\item
  If $P$ is
  $(P_1 \FILTER \varphi)$ or $(P_1 \EXCEPT P_2)$
  then $\inScope(P)=\inScope(P_1)$;
\item
  If $P$ is $(\SELECT~W~P_1)$ then
  $\inScope(P) = W$.
\end{enumerate}

So far, we have described how to translate the results of {\SPARQL} queries to relations.  However, languages built upon relations have some restrictions that difficult a straightforward translation of the {\SPARQL} operations.  The relational selection operation requires all attributes in the selection formula being attributes of the relation; the relational union is done over relations of the same schema; and the relational difference requires all variables in the subtrahend be instanced in the minuend.  Conversely, {\SPARQL} does not have these restrictions.  We next present a normal form to simplify the translation from {\SPARQL} to relational languages by satisfying the constraints of the target languages.

\begin{definition}[{\SPARQL} normal form]
  A pattern $P$ is said to be in {\em normalized} or in {\em normal
  form} if the following conditions hold:
  \begin{enumerate}
  \item
    For every sub-pattern $(P_1 \FILTER \varphi)$ in $P$ it holds that
    $\var(\varphi)\subseteq\inScope(P_1)$;
  \item
    For every sub-pattern $(P_1 \UNION P_2)$ in $P$ it holds that
    $\inScope(P_1)=\inScope(P_2)$;
  \item
    For every sub-pattern $(P_1 \EXCEPT P_2)$ in $P$ it holds that
    $\inScope(P_1)=\inScope(P_2)$.
  \end{enumerate}
\end{definition}

\begin{lemma}
  Every {\SPARQL} query (in the fragment described in Section~\ref{sec:sparql-syntax}) can be rewritten as an equivalent normalized {\SPARQL} query.
\end{lemma}

\begin{proof}
  The conditions that make a pattern normalized refer to restrictions to the in-scope variables of patterns. Patterns that are not normalized include at least one sub-pattern that has either the form $(P_1 \; \FILTER\, \varphi)$, $(P_2 \;\UNION P_3)$, or $(P_2 \;\EXCEPT P_3)$, where $\varphi$ contains a variable $\varX \notin \inScope(P_1)$, and $\inScope(P_2) \neq \inScope(P_3)$. We next present a method to normalize these patterns.

  Given a pattern $P$, and a finite set of variables $X$, $P \equiv (\SELECT~(\inScope(P) \cup X)~P)$. Indeed, a mapping $\mu$ is a solution of the pattern $(\SELECT~(\inScope(P) \cup X)~P)$ if and only there exists a solution mapping $\mu'$ of pattern $P$ such that $\mu = \mu'|_{\inScope(P) \cup X}$. By the definition of the in-scope variables, $\dom(\mu') \subseteq \inScope(P)$. Then, $\dom(\mu') \subseteq \inScope(P) \cup X$. Then, $\mu = \mu'$. Hence, $P \equiv (\SELECT~(\inScope(P) \cup X)~P)$.

  Let $P'_1$, $P'_2$, and $P'_3$ be the patterns defined as follows:
  \begin{align*}
    P'_1 &= (\SELECT~(\inScope(P_1) \cup \var(\varphi))~P_1),\\[-5pt]
    P'_2 &= (\SELECT~(\inScope(P_2) \cup \inScope(P_3))~P_2),\\[-5pt]
    P'_3 &= (\SELECT~(\inScope(P_2) \cup \inScope(P_3))~P_3).
  \end{align*}
  Since $P'_1 \equiv P_1$, $P'_2 \equiv P_2$, and $P'_3 \equiv P_3$,
  the following equivalences  hold:
  \begin{align*}
    & (P_1 \; \FILTER\, \varphi) \equiv (P'_1 \; \FILTER\, \varphi), \\[-5pt]
    & (P_2 \;\UNION P_3) \equiv\ (P'_2 \UNION P'_3), \\[-5pt]
    & (P_2 \;\EXCEPT P_3) \equiv\ (P'_2 \EXCEPT P'_3).
  \end{align*}
  Unlike the patterns on the left side of these equivalences, the patterns on the right side are normalized. Indeed, by the definition of the $\inScope$ function, $\var(\varphi) \subseteq \inScope(P'_1)$ and $\inScope(P'_2) = \inScope(P'_3)$. Hence, these equivalences can be used to normalize {\SPARQL} patterns.
\end{proof}

\begin{example}
Let $P$ be the pattern $(P_1 \UNION P_2)$ where $P_1$ is the triple pattern $(\varX,\text{\rm is},\text{\rm person})$ and $P_2$ is the triple pattern $(\varX, \text{\rm email},\varY)$, and $G$ be the RDF graph that includes the triples $(a,\text{\rm is}, \text{\rm person})$ and $(a,\text{\rm email}, \text{\rm a@ex.org})$.  The pattern $P$ is not in normal form because variable $\varY$ is in $\inScope(P_2)$, but not in $\inScope(P_1)$.  
The normal form of the pattern $P$ is a pattern $P'$ that results from replacing $P_1$ by the pattern $P_1' = (\SELECT\;\varX\,\varY\, (\varX,\text{\rm is},\text{\rm person}))$.  
The patterns $P$ and $P'$ are equivalent because the patterns $P_1$ and $P'_1$ return the same multiset of solution mappings $\Omega_1 = \lBrace \{\varX \mapsto a\}\rBrace$.  Note that variable $\varY$ is not in the solutions of $P_1$ nor $P'_1$. However, variable $\varY$ is in $\inScope(P'_1)$ but not in $\inScope(P_1)$. Using the in-scope variables of the patterns to translate the results of patterns $P_1$ and $P'_1$ as relations we get the respective relations
\[
  \left[
    \begin{array}{c}
      \varX \\ \midrule
      a \\
    \end{array}
  \right]
  \text{ and }
  \left[
    \begin{array}{cc}
      \varX & \varY \\ \midrule
      a & \bot \\
    \end{array}
  \right].
\]
  Although both relations represent the same multiset of mappings, just the second relation has the same attributes as the result of pattern $P_2$, and thus can be operated with the relational union.
\end{example}

\section{Non-Recursive Multiset Datalog with Safe Negation ({\NRMDatalogSN})}
\label{sec:datalog}

This section presents an extension of Datalog to support multiset semantics.  Based on the work of Mumick et al.~\cite{90820}, a database is defined to allow duplicate facts, and the evaluation of a fact is given by the number of different proofs for that fact. We extended Mumick's formalism in~\cite{91353} to provide a more complete formalism including negation, which we call {\MDatalogSN}. Furthermore, we follow the work of Bertossi et al.~\cite{91630} for the semantics of {\MDatalogSN}. We call \emph{Non-Recursive Multiset Datalog with Safe Negation} ({\NRMDatalogSN}) to the fragment of {\MDatalogSN} restricted to non-recursive queries.

\subsection{{\NRMDatalogSN} Syntax}
\label{def:syntax-datalog-programs}
Assume three disjoint sets: {\em variables}, {\em constants} and {\em predicate names}.  A \emph{term} is either a variable or a constant.  An \emph{atom} is an expression $p(t_1,\dots,t_n)$ where $p$ is a predicate name and each $t_i$ is a term.  An equality expression will be represented by an atom of the form $eq(t_1,t_2)$.  A \emph{literal} is either an atom (i.e. a \emph{positive literal} $A$) or the negation of an atom (i.e. a \emph{negative literal} $\neg A$).  Given a literal $L$, we use $\var(L)$ to denote the variables in $L$.  A Horn Clause, or simply clause, is an expression containing at most one positive literal. There are three types of clauses, named facts, rules and goals.

A \emph{fact} is a positive literal that does not contain any variables.  A \emph{{\MDatalogSN} Database} is a finite multiset of facts.  The {\em vocabulary} of a {\MDatalogSN} database $D$ is a pair $(P,\alpha)$ where $P$ is the set of predicate names occurring in the facts of $D$, and $\alpha$ is a function defining the arity of each predicate name in $P$, i.e. if $p(c_1,\dots,c_n) \in D$ then $\alpha(p)=n$.  The predicate names occurring in $D$ are called \emph{extensional}.

A \emph{program} $\Pi$ is a finite set of rules.
A \emph{rule} is an expression $L_{n+1} \gets L_1, \dots, L_n$ where $L_{n+1}$ is a positive literal with no constants called the head, and $L_1, \dots, L_n$ ($n \geq 1$) is a set of literals called the \emph{body}.  
The predicate names occurring in the head of the rules of $\Pi$ are called \emph{intensional}.

A variable $X$ occurs positively in a rule $R$ if and only if $X$ occurs in a positive literal in the body of $R$.  A rule $R$ is said to be \emph{safe} if all its variables occur positively.  Additionally, we will assume that every literal in the body of a rule has a variable at least.
A program is \emph{safe} if all its rules are safe. 
A \emph{{\MDatalogSN} program} is a safe program.

The \emph{dependency graph} of a program $\Pi$ is a digraph $(N,E)$ where the set of nodes $N$ is the set of predicates names that occur in the literals of $\Pi$, and there is an edge $(p_1,p_2)$ in $E$ if there is a rule in $\Pi$ whose body contains the predicate name $p_1$, and whose head contains the predicate name $p_2$.  A program is said to be \emph{non-recursive} if its dependency graph is acyclic. A {\NRMDatalogSN} program is a {\MDatalogSN} that is non-recursive.

A \emph{goal clause} is an atom without constants. A {\MDatalogSN} \emph{query} is a pair $(L,\Pi)$ where $L$ is a goal clause, and $\Pi$ is a {\MDatalogSN} program. A {\NRMDatalogSN} query is a {\MDatalogSN} query $(L,\Pi)$ such that $\Pi$ is non-recursive. A {\NRMDatalogSN} database is a {\MDatalogSN} database.
We assume that extensional and intensional predicate symbols are disjoint; in particular, a predicate cannot appear both as a fact in the database and as the head of a rule.

\begin{figure*}[t]
  \begin{adjustbox}{width=\linewidth}
    \begin{tikzpicture}[anchor=base]
      \matrix (ps) [matrix of nodes, column sep=0mm, row sep=5mm] {
      \node (Xr1) {$r(a)$};
      & \node (Xr2) {$r(a)$};
      & \node (Xpr1) {$p(a)$};
      & \node (Xpr2) {$p(a)$};
      & \node (q11a) {};
      & \node (q11b) {};
      & \node (q12a) {};
      & \node (q12b) {};
      & \node (q21a) {};
      & \node (q21b) {};
      & \node (q22a) {};
      & \node (q22b) {};
      \\
      \node (rX1) {$\colCopy{r(a)}{1}$};
      & \node (rX2) {$\colCopy{r(a)}{2}$};
      & \node (pXr1) {$r(a)$};
      & \node (pXr2) {$r(a)$};
      & \node (q11Xr1) {$r(a)$};
      & \node (q11Xpr1) {$p(a)$};
      & \node (q12Xr1) {$r(a)$};
      & \node (q12Xpr2) {$p(a)$};
      & \node (q21Xr2) {$r(a)$};
      & \node (q21Xpr1) {$p(a)$};
      & \node (q22Xr2) {$r(a)$};
      & \node (q22Xpr2) {$p(a)$};
      \\
      &
      & \node (prX1) {$\colCopy{r(a)}{1}$};
      & \node (prX2) {$\colCopy{r(a)}{2}$};
      & \node (q11rX1) {$\colCopy{r(a)}{1}$};
      & \node (q11pXr1) {$r(a)$};
      & \node (q12rX1) {$\colCopy{r(a)}{1}$};
      & \node (q12pXr2) {$r(a)$};
      & \node (q21rX2) {$\colCopy{r(a)}{2}$};
      & \node (q21pXr1) {$r(a)$};
      & \node (q22rX2) {$\colCopy{r(a)}{2}$};
      & \node (q22pXr2) {$r(a)$};
      \\
      &
      &
      &
      &
      & \node (q11prX1) {$\colCopy{r(a)}{1}$};
      &
      & \node (q12prX2) {$\colCopy{r(a)}{2}$};
      &
      & \node (q21prX1) {$\colCopy{r(a)}{1}$};
      &
      & \node (q22prX2) {$\colCopy{r(a)}{2}$};
      \\
      };
      \node(q11) at ($(q11a)!0.5!(q11b)$) {$q(a)$};
      \node(q12) at ($(q12a)!0.5!(q12b)$) {$q(a)$};
      \node(q21) at ($(q21a)!0.5!(q21b)$) {$q(a)$};
      \node(q22) at ($(q22a)!0.5!(q22b)$) {$q(a)$};
      \path[->] (Xr1) edge node[right] {$\scriptscriptstyle F$} (rX1);
      \path[->] (Xr2) edge node[right] {$\scriptscriptstyle F$} (rX2);
      \path[->] (Xpr1) edge node[right] {$\scriptscriptstyle S$} (pXr1);
      \path[->] (Xpr2) edge node[right] {$\scriptscriptstyle S$} (pXr2);
      \path[->] (pXr1) edge node[right] {$\scriptscriptstyle F$} (prX1);
      \path[->] (pXr2) edge node[right] {$\scriptscriptstyle F$} (prX2);
      \path[->] (q11) edge node[left] {$\scriptscriptstyle R$} (q11Xr1);
      \path[->] (q11) edge node[right] {$\scriptscriptstyle R$} (q11Xpr1);
      \path[->] (q12) edge node[left] {$\scriptscriptstyle R$} (q12Xr1);
      \path[->] (q12) edge node[right] {$\scriptscriptstyle R$} (q12Xpr2);
      \path[->] (q21) edge node[left] {$\scriptscriptstyle R$} (q21Xr2);
      \path[->] (q21) edge node[right] {$\scriptscriptstyle R$} (q21Xpr1);
      \path[->] (q22) edge node[left] {$\scriptscriptstyle R$} (q22Xr2);
      \path[->] (q22) edge node[right] {$\scriptscriptstyle R$} (q22Xpr2);
      \path[->] (q11Xr1) edge node[right] {$\scriptscriptstyle F$} (q11rX1);
      \path[->] (q12Xr1) edge node[right] {$\scriptscriptstyle F$} (q12rX1);
      \path[->] (q21Xr2) edge node[right] {$\scriptscriptstyle F$} (q21rX2);
      \path[->] (q22Xr2) edge node[right] {$\scriptscriptstyle F$} (q22rX2);
      \path[->] (q11Xpr1) edge node[right] {$\scriptscriptstyle S$} (q11pXr1);
      \path[->] (q12Xpr2) edge node[right] {$\scriptscriptstyle S$} (q12pXr2);
      \path[->] (q21Xpr1) edge node[right] {$\scriptscriptstyle S$} (q21pXr1);
      \path[->] (q22Xpr2) edge node[right] {$\scriptscriptstyle S$} (q22pXr2);
      \path[->] (q11pXr1) edge node[right] {$\scriptscriptstyle F$} (q11prX1);
      \path[->] (q12pXr2) edge node[right] {$\scriptscriptstyle F$} (q12prX2);
      \path[->] (q21pXr1) edge node[right] {$\scriptscriptstyle F$} (q21prX1);
      \path[->] (q22pXr2) edge node[right] {$\scriptscriptstyle F$} (q22prX2);
    \end{tikzpicture}
  \end{adjustbox}
\caption{Example of derivation trees. 
Let $D$ be a {\NRMDatalogSN} database, $F = r(a)$ be a fact in $D$ with $\card(F,D) = 2$, $\Pi = \{R,S\}$ be a {\NRMDatalogSN} program where $R$ is the rule $q(X) \gets r(X),p(X)$ and $S$ is the rule $p(X) \gets r(X)$.
This figure shows the derivation trees of $\Pi$ with respect to $D$.}
\label{fig:derivation-trees}
\end{figure*}
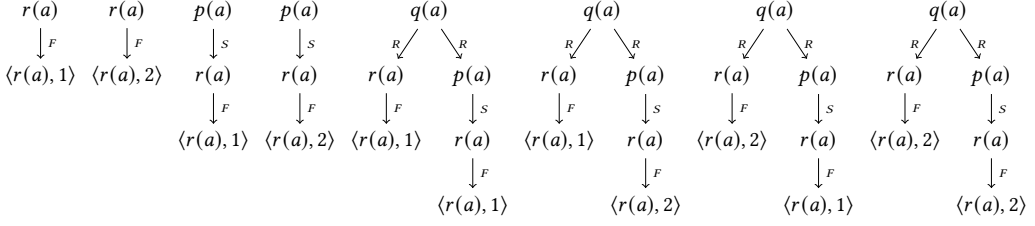

\subsection{{\NRMDatalogSN} Semantics}
\label{ss:datalogsemantics}
We follow the formalisms by Mumick et al.~\cite{90820} and Bertossi et al.~\cite{91630} that use a proof-theoretic semantics for {\NRMDatalogSN} programs.
The semantics is based on the notions of  ``substitution'' and ``derivation tree''.  

A \emph{substitution} is a partial function $\theta$ from variables to constants. Given a literal $L$ (positive or negative), and a substitution $\theta$, we write $\theta(L)$ to denote the literal $L'$ obtained by replacing all variables $x$ occurring in $L$ with $\theta(x)$. Informally, the answer for a query $(L,\Pi)$ where $\Pi$ is a {\NRMDatalogSN} program, over a database $D$, will be a multiset of substitutions with the same domain, each obtained from one proof showing that this substitution works.
 We will use 
the notion of ``colored set'' \cite{90820} to identify the different copies of an element, as introduce in 
Section 2.1. 


The notion of ``derivation tree'' \cite{90820} will be used to count the number of proofs for an atom.
Formally, a \emph{Derivation Tree} is a connected, undirected graph, with no cycles, represented as a tuple $\T = (\TN, \TE, \TL, \epsilon, \lambda)$ where $\TN$ is a set of nodes, $\TE$ is a set of edges, $\TL$ is a set of labels (for nodes and edges), $\epsilon : \TE \to \TN \times \TN$ is a total function that assigns a pair of nodes to each edge, and $\lambda : (\TN \cup \TE) \to \TL$ is a total function that assigns a label to each node and edge. 
The function $root(\T)$ will be used to obtain the root node of $\T$.

Let $R$ be a rule of the form $L_{n+1} \gets L_1,\dots,L_m, L_{m+1},\dots, L_n$ where $L_1,\dots,L_m$ are positive literals, and $L_{m+1},\dots, L_n$ are negative literals, $DT$ be a set of derivation trees, and $ST = ( \T_1, \dots, \T_m )$ be a sequence of derivation trees that satisfy that every derivation tree in $ST$ is also in $DT$.
We say that $DT$ matches $R$ with $ST$, denoted $DT \models^{ST} R$, if there is a substitution $\theta$ satisfying:
(i) for every positive literal $L_i \in R$ it applies that $\theta(L_i) = root(\T_i)$ where $\T_i \in ST$; and
(ii) for every negative literal $L_j \in $ it applies that $DT$ does not contain a derivation tree whose root node has the label $\theta(L_j)$.

Assume that $DT \models^{ST} R$ where $ST = ( \T_1, \dots, \T_m )$,  $\T_1 = (\TN_1, \TE_1, \TL_1, \epsilon_1, \lambda_1)$, 
$\dots$, and
$\T_m = (\TN_m, \TE_m, \TL_m, \epsilon_m, \lambda_m)$. 
The derivation tree $\T_R = (\TN_R, \TE_R, \TL_R, \epsilon_R, \lambda_R)$ for the rule $R$ is defined as follows: 
$\TN_R = \{ n_r \} \cup \TN_1 \cup \dots \cup \TN_m$, 
$\TE_R = \{ e_1, \dots, e_m \} \cup \TE_1 \cup \dots \cup \TE_m$,
$\TL_R = \{ R, \theta(L_{n+1}) \} \cup \TL_1 \cup \dots \cup \TL_m$,
every assignment in $\epsilon_i$ is also in $\epsilon_R$,
$\epsilon_R(e_1) = (n_r, root(\T_1))$,
$\dots$,
$\epsilon_R(e_m) = (n_r, root(\T_m))$,
every assignment in $\lambda_i$ is also in $\lambda_R$,
$\lambda_R(n_r) = \theta(L_{n+1})$, 
and $\lambda_R(e_1) = R$.

Let $D$ be a {\NRMDatalogSN} database and $\Pi$ a {\NRMDatalogSN} program. The set of derivation trees of $\Pi$ with respect to $D$, denoted $\dt(\Pi,D)$, is defined as follows:
\begin{enumerate}
\item 
For every fact $F \in D$ of the form $p(t_i,\dots,t_n)$, and for every colored copy $\colCopy{p(t_i,\dots,t_n)}{i}$, it applies that $\dt(\Pi,D)$ contains a derivation tree $\T^i_F = (\TN^i_F, \TE^i_F, \TL^i_F, \epsilon^i_F, \lambda^i_F)$ where 
$\TN^i_F = \{ n_1, n_2 \}$, 
$\TE^i_F = \{ e_1 \}$,
$\TL^i_F = \{ p(t_i,\dots,t_n), \colCopy{p(t_i,\dots,t_n)}{i}, F \}$,
$\epsilon^i_F(e_1) = (n_1,n_2)$,
$\lambda^i_F(n_1) = p(t_i,\dots,t_n)$, 
$\lambda^i_F(n_2) = \colCopy{p(t_i,\dots,t_n)}{i}$, 
and $\lambda^i_F(e_1) = F$.
\item
Assume that $DT$ is the set of derivation trees obtained for the facts in $D$ as defined above.
Given a rule $R$ in $\Pi$ and a sequence of derivation trees $ST$ satisfying $DT \models^{ST} R$, the derivation tree for $R$ is added to
$DT$. This process is repeated for every rule $R$ in $\Pi$, until no more derivation trees are generated. Finally, $\dt(\Pi,D) = DT$.  
\end{enumerate}

Let $\Pi$ be a {\NRMDatalogSN} program, $D$ be a {\NRMDatalogSN} database and $F$ be a fact. A derivation tree $\T \in \dt(\Pi,D)$ is said to be a \emph{proof} for the fact $F$ if the label of the root node is $F$. The multiset of \emph{atoms} of $\Pi$ in $D$, denoted $\ans(\Pi, D)$, is the multiset of facts $F$ such that there is a proof for $F$ in $\dt(\Pi, D)$, and the cardinality of $F$ in $\ans(\Pi, D)$ is the number of proofs of $F$.  Figure~\ref{fig:derivation-trees} shows the derivation trees that are proofs of the facts derived from an example {\NRMDatalogSN} program.
The facts $r(a)$, $p(a)$ and $q(a)$ belong to $\ans(\Pi, D)$ with cardinalities 2, 2, and 4.

The {\NRMDatalogSN} query language over a vocabulary $\tau$ is the query language $(\calQ, \calD, \calS, \ev{\cdot}{\cdot})$ where:
\begin{enumerate}
\item
  $\calQ$ is the set of {\NRMDatalogSN} queries over $\tau$;
\item
  $\calD$ is the set of {\NRMDatalogSN} databases over $\tau$;
\item
  $\calS$ is the set of {\NRMDatalogSN} query answers (i.e., pairs $(V, M)$ where $V$ is a set of variables and $M$ is a multiset of substitutions $\theta$ with $\dom(\theta)=V$); and
\item
  {\sloppy $\ev{\cdot}{}$ is the function that receives a {\NRMDatalogSN} query $(L,\Pi)$ and a {\NRMDatalogSN} database $D$, and returns a {\NRMDatalogSN} query answer $(V,M)$ where $V = \var(L)$, $\set(M) = \{ \theta \mid \theta(L)\in\ans(\Pi,D) \text{ and} \dom(\theta)=V\}$, and $\card(\theta,M) = \card(\theta(L),\ans(\Pi,D))$.\par}
\end{enumerate}

Observe that the domain of the query answer for a query $(L,\Pi)$ is $\var(L)$. Abusing notation, we will say that it is also the domain of the query $(L,\Pi)$, denoted $\dom((L,\Pi))=\var(L)$.

The Multiset Datalog query language presented here, {\NRMDatalogSN}, differs from the version proposed by Bertossi et al.~\cite{91630} in that we do not allow recursive programs nor constants in the head of rules. These restrictions permit to match the expressive power of the {\SPARQL} fragment studied here.

\subsection{Normalization of {\NRMDatalogSN} programs}
To simplify the translations from {\NRMDatalogSN} to SPARQL and $\MRA$, we assume that every {\NRMDatalogSN} query is \emph{normalized} into a query that contains only rules of the three following types:

\begin{tabbing}
  $L_0 \leftarrow L_1$,
  \hspace{4em}\= where $\var(L_0) \subseteq \var(L_1);$
  \hspace{14em}\= {\normalfont(projection rule)} \\
  $L_0 \leftarrow L_1, L_2$,
  \> where $\var(L_0) = \var(L_1) \cup \var(L_2)$;
  \> {\normalfont(join rule)} \\
  $L_0 \leftarrow L_1, \neg L_2$,
  \> where $\var(L_2) = \var(L_1)$ and $\var(L_0) = \var(L_1)$.
  \> {\normalfont(negation rule)}
\end{tabbing}
Next, we show the feasibility of this normalization.

\begin{lemma}\label{lemma:normalized-md}
  Every {\NRMDatalogSN} query is equivalent to a normalized {\NRMDatalogSN} query.
\end{lemma}

\begin{proof}
  We provide a normalization algorithm that replaces every rule in the query by a set of rules that do not change the semantics of the query. Given a {\NRMDatalogSN} query $(L,\Pi)$, every rule $R \in \Pi$ has the form
  \[p(\bar{X}) \gets A_1,\dots,A_m, \neg B_{1},\dots,\neg B_n,\] 

  where $A_1,\dots,A_m$ are positive literals, and $\neg B_{1},\dots,\neg B_n$ are negative literals. For $1 \leq i \leq m$, let $\bar{Y}_i$ be the set of variables that consists of the variables occurring in the atoms $A_1,\dots,A_i$. Then, we replace rule $R$ by the minimal set of rules $\Pi_R$ that includes the following rules:  
  \begin{enumerate}
  \item
    Rules $R_i^A$, for $2 \leq i \leq m$, defined recursively as follows:
    \begin{enumerate}
    \item
      $R_2^A = q_2^A(\bar{Y}_2) \gets A_1,A_2$.
    \item
      $R_i^A = q_i^A(\bar{Y}_{i}) \gets q_{i-1}^A(\bar{Y}_{i-1}),A_i$.
    \end{enumerate}
  \item
    Rules $R_j^B$ and $R_j^{B'}$ for $1 \leq j \leq n$,
    defined recursively as follows:
    \begin{enumerate}
    \item
      $R_0^B=  r_0^B(\bar{Y}_m) \gets q_m^A(\bar{Y}_m)$,
    \item
      $R_j^B = r_j^B(\bar{Y}_m) \gets r_{j-1}^B(\bar{Y}_m), \neg B'_j(\bar{Y}_m) $,
    \item
      $R_j^{B'} =  B'_j(\bar{Y}_m) \gets r^B_{j-1}(\bar{Y}_m), B_j$ .
    \end{enumerate}
  \item
    A rule $R' = p(\bar{X}) \gets r_n^B(\bar{Y}_m).$
  \end{enumerate}
  Let $(L,\Pi')$ be the query resulting from replacing rule $R$ with the rules in $\Pi_R$. It is clear that the program is normal (recall that the original program is safe).
  Need to show that both programs are equivalent, that is, that the solutions of query $(p(\bar{X}),\Pi)$ after the replacement are the same and have the same cardinalities. These two conditions follow from Claim~\ref{claim:normalized-md} in the Appendix.
\end{proof}

\section{Multiset Relational Algebra ({\MRA})}
\label{sec:mra}

The multiset relational algebra used in this paper is based on the semantics defined by Dayal et al.~\cite{91030}.  This algebra considers the operations of {\em selection}, {\em projection}, {\em natural join} and {\em arithmetic union}.  Additionally, we include operators for {\em renaming} and {\em filter difference} (or ``except'').

\subsection{Multiset relations}

Assume that $\sN, \sA, \sC$ are disjoint infinite sets, where $\sN$ is the domain of relation names, $\sA$ is the domain of {\em attributes}, and $\sC$ is the domain of {\em constants} or {\em values}.

A \emph{relation schema} is given by a relation name $R \in \sN$ and a set of attributes $\{ A_1, \dots, A_n \}$ where $A_i \in \sA$ for $1 \leq i \leq n$. To simplify the notation, we will use the relation name $R$ to denote the relation schema, and $\widehat{R}$ to denote the attributes of $R$. A \emph{relational database schema} is a finite set of relation schemas.

A \emph{tuple} over a relation schema $R$ with attributes $\widehat{R} = \{ A_1, \dots, A_n \}$ is a total mapping $t$ from $\widehat{R}$ to $\sC$.  The value of tuple $t$ on an attribute $A_i \in \widehat{R}$ will be denoted as $t(A_i)$.  Given a set of attributes $U \subseteq \widehat{R}$ and a tuple $t$, we write $t[U]$ to denote the tuple $t'$ with attributes $U$ such that $t'(A) = t(A)$ for every attribute $A \in U$.

A {\em multiset relation} $r$ over a relation schema $R$ is a multiset of tuples over $\widehat{R}$.  We write $\hat{r}$ to denote the relation schema $R$ where the multiset relation $r$ is defined.  Given a tuple $t \in r$, we will use $\card(t,r)$ to denote the cardinality of tuple $t$ in $r$.

A \emph{relational database schema} is a set of relation schemas.  Given a relational database schema $T = \{R_1,\dots,R_n\}$, a \emph{multiset relational database} over $T$ is a set of multiset relations $\{r_1,\dots,r_n \}$ where each relation $r_i$ is defined over the schema $R_i$. Sometimes we will write {\MRA} database, emphasizing that the multiset relational database is in the context of {\MRA}.

Let $r_1, r_2$ be two multiset relations, and $t_1 \in r_1$ and $t_2 \in r_2$ be tuples. We say that $t_1$ and $t_2$ are compatible, denoted $t_1 \sim t_2$, if (i) for every attribute $A \in \hat{r}_1 \cap \hat{r}_2$ it holds that $t_1(A) = t_2(A)$, or (ii) $\hat{r}_1 \cap \hat{r}_2 = \emptyset$. If $t_1$ and $t_2$ are compatible, then the merge of them, denoted $t_1 \cup t_2$, is the tuple $t$ with attributes $\hat{r}_1 \cup \hat{r}_2$ where $t(A)=t_1(A)$ for each attribute $A \in \hat{r}_1$, and $t(B)=t_2(B)$ for each attribute $B \in \hat{r}_2 \setminus \hat{r}_1$.

\subsection{Syntax of {\MRA}}

The multiset relational algebra defined in this paper includes the operators of selection ($\sigma$), projection ($\pi$), renaming ($\rho$), join ($\Join$), union ($\cup$), and except ($\setminus$).  Next we describe the syntax of {\MRA} expressions containing the above operators.

A \emph{selection formula} $\psi$ is a Boolean combination of equality expressions of the form $x=y$ where $x,y \in \sA\cup\sC$.  We define a \emph{{\MRA} expression} $E$ over a relational database schema $T$, and the attributes of $E$, denoted $\widehat{E}$, by mutual recursion as follows:
\begin{itemize}
\item
  A relation name $R \in T$ is a {\MRA} expression $E$, and $\widehat{E}=\widehat{R}$.
\item
  If $E_1$ is a {\MRA} expression and $\psi$ is a selection formula where the attributes occurring in $\psi$ are included in $\widehat{E}_1$, then $\sigma_{\psi}(E_1)$ is a {\MRA} expression $E$, and $\widehat{E}=\widehat{E}_1$.
\item
  If $E_1$ is a {\MRA} expression and $S \subseteq \widehat{E}_1$ is a set of attributes, them $\pi_{S}(E_1)$ is a {\MRA} expression $E$, and $\widehat{E}=S$.
\item
  If $E_1$ is a {\MRA} expression, $A\in\widehat{E}_1$ and $B \in \sA$ are attributes, then $\rho_{A/B}(E_1)$ is a {\MRA} expression $E$, and $\widehat{E}=(\widehat{E}_1\setminus A) \cup B$.
\item
  If $E_1$ and $E_2$ are {\MRA} expressions, then $(E_1 \Join E_2)$ is an {\MRA} expression $E$, and $\widehat{E}=\widehat{E}_1\cup\widehat{E}_2$.
\item
  If $E_1$ and $E_2$ are {\MRA} expressions and $\widehat{E}_1=\widehat{E}_2$, then $(E_1 \setminus E_2)$ is a {\MRA} expression $E$, and $\widehat{E}=\widehat{E}_1$.
\item
  If $E_1$ and $E_2$ are {\MRA} expressions and $\widehat{E}_1=\widehat{E}_2$, then $(E_1 \cup E_2)$ is a {\MRA} expression $E$, and $\widehat{E}=\widehat{E}_1$.
\end{itemize}

Note that a selection operation $\sigma_\psi(E_1)$ requires that attributes in the selection formula $\psi$ be attributes of the {\MRA} expression $E_1$; the projection operation $\pi_S(E_1)$ requires that $S$ be a subset of the attributes of the {\MRA} expression $E_1$; and that the union $E_1 \cup E_2$ and difference $E_1 \setminus E_2$ expressions require that expressions $E_1$ and $E_2$ have the same set of attributes.

\subsection{Semantics of {\MRA}}

Given a selection formula $\psi$ and a tuple $t$ over a relation schema $R$, we will use $t \models \psi$ to denote that $t$ satisfies $\psi$, and its evaluation is given as follows:
\begin{enumerate}

\item
  if $\psi$ is $A = B$ where $A,B \in \widehat{R}$ are attributes, then $t \models \psi$ iff $t(A) = t(B)$;

\item
  if $\psi$ is $A = c$ where $A \in \widehat{R}$ is an attribute and $c \in \sC$ is a constant, then $t \models \psi$ iff $t(A) = c$;

\item
  if $\psi$ is $c_1 = c_2$ where $c_1, c_2 \in \sC$ are constants, then $t \models \psi$ iff $c_1$ is the same constant as $c_2$;

\item
  if $\psi$ is $\psi_1 \land \psi_2$, then $t \models \psi$ iff $t \models \psi_1$ and $t \models \psi_2$;

\item
  if $\psi$ is $\psi_1 \lor \psi_2$, then $t \models \psi$ iff $t \models \psi_1$ or $t \models \psi_2$;

\item
  if $\psi$ is $\neg \psi_1$, then $t \models \psi$ iff $t \models \psi_1$ does not hold.
\end{enumerate}

Now, the evaluation of a {\MRA} expression $E$ over a multiset relational database $D$ (of the same schema as $E$) is defined as a function $\Eval(E,D)$ that returns a multiset relation $r$ with the same schema as $E$.

Let $D$ be a {\MRA} database over a schema $T$ and $E, E_1, E_2$ be {\MRA} expressions over $T$.  The evaluation of $\Eval(E,D)$ is the multiset relation $r$ defined recursively as follows (assume that $\Eval(E_1,D) = r_1$, and $\Eval(E_2,D) = r_2$):
\begin{itemize}
\item
  If $E$ is a relation name $R_1 \in T$, then $r$ is the relation for the relation name $R_1$ in the database $D$.
\item
  If $E$ is $\sigma_{\psi}(E_1)$ then $\set(r) = \{ t \mid t \in r_1 \text{ and } t \models \psi \}$ and $\card(t, r) = \card(t, r_1)$.
\item
  If $E$ is $\pi_S(E_1)$ then $\set(r) = \{ t' \mid t' = t[ S] \text{ and } t \in r_1 \}$ and
  \[
  \card(t',r) = \sum_{t \mbox{ \footnotesize{with} } t[S] =t'} \card(t,r_1).
  \]
\item
  If $E$ is $\rho_{A/B} (E_1)$ then $r$ is the result from renaming in $r_1$ the attribute $A$ as $B$.
\item
  If $E$ is $(E_1 \Join E_2)$ then $\set(r) = \{ t_1 \cup t_2 \mid t_1 \in r_1,\; t_2 \in r_2, \text{ and } t_1 \sim t_2 \}$ and $\card(t_1\cup t_2,r) = \card(t_1,r_1) \times \card(t_2,r_2)$.
\item
  If $E$ is $(E_1 \cup E_2)$ then $\set(r) = \{ t \mid t \in r_1 \text{ or } t \in r_2 \}$ and $\card(t,r) = \card(t,r_1) + \card(t,r_2)$.
\item
  If $E$ is $(E_1 \setminus E_2)$ then $\set(r) = \{ t \mid t \in r_1 \text{ and } t \notin r_2 \}$ and $\card(t,r) = \card(t,r_1)$.
\end{itemize}

Hence, in {\MRA}, the set of queries is the set of {\MRA} expressions, the set of databases is the set of multiset relational databases, the set of results is the set of multiset relations, and the evaluation procedure is the aforementioned function $\Eval$.

\section{Equivalence between {\SPARQL} and {\NRMDatalogSN}}
\label{sec:sparql-nrmd}
This section presents the simulations that prove that {\SPARQL} and Non-Recursive Multiset Datalog with Safe Negation ({\NRMDatalogSN}) have the same expressive power. 
Specifically, we show that {\SPARQL} can be simulated by {\NRMDatalogSN} (Section \ref{sec:sparql2datalog}), and {\NRMDatalogSN} can be simulated by {\SPARQL} (Section \ref{sec:datalog2sparql}). 

\subsection{From {\SPARQL} to {\NRMDatalogSN}}%
\label{sec:sparql2datalog}
This section shows that {\SPARQL} can be simulated by Non-Recursive Multiset Datalog with Safe Negation ({\NRMDatalogSN}).
To support this, we describe the following translation functions:
\begin{itemize}
\item function $f_{12}$, that translates {\SPARQL} queries into {\NRMDatalogSN} queries; 
\item function $g_{12}$, that translates {\SPARQL} databases into {\NRMDatalogSN} databases; and
\item function $h_{12}$, that translates {\NRMDatalogSN} query answers into {\SPARQL} query answers.
\end{itemize}

\subsubsection{Translating databases from {\SPARQL} to {\NRMDatalogSN}}
Recall that a {\SPARQL} database is a set of RDF triples and a {\NRMDatalogSN} database is a multiset of facts.
Given this, the basic idea is to translate each RDF triple into a Datalog atom. Additionally, we create an atom to encode all RDF terms, and an atom to encode the unbound value.

\begin{definition}[Function $g_{12}$]%
\label{def:graphs-to-facts}
Let $\bot$ be an RDF term which will exclusively used to encode the unbounded value.
Let $(P,\alpha)$ be a Datalog vocabulary where $P = \{ \term, \equal, \comp, \triple, \unull \}$ and  
$\alpha(\term)=1$, $\alpha(\equal)=2$, $\alpha(\comp)=3$, $\alpha(\triple)=3$, and $\alpha(\unull)=1$.
Given an RDF graph $G$, the function $g_{12}(G)$ returns a {\NRMDatalogSN} database $D$ with vocabulary $(P,\alpha)$, where the multiset of facts in $D$ is defined as follows:
\begin{itemize}
\item
for each RDF term $t$ in $\terms(G)$ (i.e. the set of terms in $G$) , $D$ contains the facts $\term(t)$, $\equal(t,t)$, 
$\comp(t, t, t)$, $\comp(t, \bot, t)$, $\comp(\bot, t, t)$;
\item
for each RDF triple $(s,p,o) \in G$, $D$ contains the fact $\triple(s,p,o)$; 
\item 
$D$ contains the fact $\comp(\bot, \bot, \bot)$;
\item
$D$ contains the fact $\unull(\bot) \in D$.
\end{itemize}
\end{definition}

Intuitively: 
a fact of the form $\term(t)$ is used to represent that $t$ is an RDF term;
a fact of the form $\equal(t,t)$ is used to represent the equality of the RDF term $t$;
for each term $t$ in $G$, the notion of compatible is represented by the facts $\comp(t, t, t)$, $\comp(t, \bot, t)$ and $\comp(\bot, t, t)$; 
a fact of the form $\triple(s,p,o)$ is used to represent an RDF triple $(s,p,o)$.



\begin{example}
Let $G$ be the RDF graph defined as follows
  \[
    G = \{
    \begin{aligned}[t]
      &(\iAlice, \pLivesIn, \iSantiago), (\iAlice, \pKnows, \iBob),
      (\iBob, \pLivesIn, \iSantiago), (\iBob, \pKnows, \iCarol),\\[-4pt]
      &(\iCarol, \pLivesIn, \iLima) \}.
    \end{aligned}
  \]
The translation of the above RDF triples into Datalog facts is given as follows:
\[
g_{12}(G) = \lBrace
\begin{aligned}[t]
 &\term(\iAlice), \equal(\iAlice, \iAlice),\\[-4pt]
 &\comp(\iAlice, \iAlice, \iAlice), \comp(\iAlice, \bot, \iAlice), \comp(\bot, \iAlice, \iAlice),\\[-4pt] 
 & \quad \vdots \\[-4pt]      
 &\term(\iLima), \equal(\iLima, \iLima),\\[-4pt]
 &\comp(\iLima, \iLima, \iLima), 
  \comp(\iLima, \bot, \iLima),
  \comp(\bot, \iLima, \iLima),\\[-4pt] 
 &\triple(\iAlice, \pLivesIn, \iSantiago),
  \triple(\iAlice, \pKnows, \iBob),\\[-4pt]
 &\triple(\iBob, \pLivesIn, \iSantiago),
  \triple(\iBob, \pKnows, \iCarol),\\[-4pt]
 &\triple(\iCarol, \pLivesIn, \iLima),\\[-4pt]
 &\comp(\bot, \bot, \bot),\\[-4pt]
 &\unull(\bot) \\[-4pt]
 \rBrace.
 \end{aligned}
\]
\end{example}

\subsubsection{Translating queries from {\SPARQL} to {\NRMDatalogSN}}
\label{sec:sparql-to-md}
In general terms, any {\SPARQL} graph pattern can be translated into a set of {\NRMDatalogSN} rules. However, there are some subtleties that need to be discussed before presenting the general translation rules. 

An initial issue is the translation of a filter graph pattern $ P = (P_1 \FILTER \varphi)$ where $\varphi$ is a complex filter condition.  In order to simplify the translation to Datalog, we need to transform $P$ into a collection of filter graph patterns where every filter condition is an atomic filter condition.

Consider the following equivalences:
\begin{align}
  &(P_1\FILTER \varphi_1 \land \varphi_2) \equiv
    ((P_1 \FILTER \varphi_1) \FILTER \varphi_2).
    \label{formula:conj}
  \\
  &(P_1\FILTER \varphi_1 \lor \varphi_2) \equiv
    (P_1 \FILTER \varphi_1) \UNION
    (P_1 \FILTER \varphi_2).
    \label{formula:disj1}
  \\
  &(P_1 \FILTER \neg (\varphi_1)) \equiv 
    (P_1 \EXCEPT (P_1 \FILTER \varphi_1) ).
    \label{formula:neg}
\end{align}
Intuitively, these equivalences seem to be true, since similar equivalences are valid in set relational algebra, namely $\sigma_{\varphi_1 \land \varphi_2}(R) = \sigma_{\varphi_1}(\sigma_{\varphi_2}(R))$, $\sigma_{\varphi_1 \lor \varphi_2}(R) = \sigma_{\varphi_1}(R) \cup \sigma_{\varphi_2}(R)$, and $\sigma_{\neg\varphi_1}(R) = R \setminus \sigma_{\varphi_1}(R)$. Under set semantics, these three equivalences are valid. However, under bag semantics, just equivalence (\ref{formula:conj}) is valid, and equivalences (\ref{formula:disj1}) and (\ref{formula:neg}) present problems.  Let us analyze them and provide valid equivalences.

\begin{itemize}
\item
  To see why equivalence (\ref{formula:disj1}) is not valid, consider the case where for a solution $\mu$ of the pattern $P_1$ the evaluation of formulas $\varphi_1$ and $\varphi_2$ are true. Then, $\mu$ is a solution of the queries in both sides of equivalence (\ref{formula:disj1}). However, the cardinality differs. Indeed, the cardinality of $\mu$ for the query on the right side is twice the cardinality for the query on the left side. Hence, equivalence (\ref{formula:disj1}) is valid for set semantics but not for bag semantics.
\item
To see why equivalence (\ref{formula:neg}) is not valid, consider the case where for a solution mapping $\mu$ of the pattern $P_1$, formula $\varphi_1$ produces an error. Then, formula $\neg\varphi_1$ also produces an error, and hence $\mu$ is not a solution to the query on the left side. On the other hand, since $\mu$ is a solution mapping for $P_1$ but not a solution to the pattern $(P_1 \FILTER \varphi_1)$, $\mu$ is a solution mapping to the query on the right side. Hence, this equivalency is not valid because the queries do not have the same solution mappings.
\end{itemize}
Intuitively, equivalence (\ref{formula:disj1}) is no longer valid when we change from set semantics to bag semantics, whereas equivalence (\ref{formula:neg}) is no longer valid when we change from 2-valued logic to 3-valued logic.  In the following, we show how to solve these problems.

\begin{lemma}[Rewriting of disjoint filter conditions]%
  \label{disjunction-rewriting-simple}
  We say that two filter conditions $\varphi_1$ and $\varphi_2$ are {\em disjoint}, if for every mapping $\mu$ it does not hold that $\mu(\varphi_1)$ and $\mu(\varphi_2)$ are simultaneously $\true$.  Equivalence (\ref{formula:disj1}) is true when $\varphi_1$ and $\varphi_2$ are disjoint.
\end{lemma}
\begin{proof}
  Given that $\varphi_1$ and $\varphi_2$ are disjoint, it applies that $\mu(\varphi_1)$ is $\true$ when $\mu(\varphi_2)$ is not $\true$ (and vice versa). So, it holds that $\mu(\varphi_1 \lor \varphi_2) = \true$ if and only if $\mu(\varphi_1)=\true$ or $\mu(\varphi_2)=\true$, and the cardinality of $\mu$ on the left hand side is the sum of the cardinalities of $\mu$ in each of the terms of the right hand side.
\end{proof}

Now, consider the following equivalence:
\begin{align}
  &\begin{aligned}[t]
    (P_1\FILTER \varphi_1 \lor \varphi_2) \equiv 
    &(P_1 \FILTER \varphi_1 \land \neg \varphi_2) \UNION \\[-5pt]
    &(P_1 \FILTER \neg \varphi_1 \land  \varphi_2) \UNION \\[-5pt]
    &(P_1 \FILTER \varphi_1 \land \varphi_2).
     \end{aligned}
  \label{formula:disj2}
\end{align}

Equivalence (\ref{formula:disj2}) solves one of the problems of equivalence (\ref{formula:disj1}), but it still has problems in evaluating formulas with errors. 
In order to solve them, we introduce the notion of ``error filter condition.''

\begin{definition}[Error filter condition]%
\label{def:error-formula}
Let $\varphi$, $\varphi_1$ and $\varphi_2$ be filter conditions.
The function $\Error(\varphi)$ returns a filter condition defined recursively as follows:
\begin{itemize}
\item if $\varphi$ is $\bound(?X)$ then $\Error(\varphi) = \false;$
\item if $\varphi$ is $?X=a$ then $\Error(\varphi) = \neg\bound(?X)$;
\item if $\varphi$ is $?X={?Y}$ then\\ $\Error(\varphi)$ =
 $(\neg\bound(?X) \land \bound(?Y))$ $\lor$
 $(\bound(?X) \land \neg\bound(?Y))$ $\lor$
 $(\neg\bound(?X) \land \neg\bound(?Y))$;
\item if $\varphi$ is $\varphi_1 \land \varphi_2$ then\\
 $\Error(\varphi)$ =
 $(\varphi_1 \land \Error(\varphi_2))$  $\lor$
 $(\Error(\varphi_1) \land \varphi_2)$ $\lor$
 $(\Error(\varphi_1) \land \Error(\varphi_2))$;
\item if $\varphi$ is $\varphi_1 \lor \varphi_2$ then\\ 
 $\Error(\varphi)$ =
 $(\neg\varphi_1 \land \Error(\varphi_2))$ $\lor$
 $(\Error(\varphi_1) \land \neg\varphi_2)$ $\lor$
 $(\Error(\varphi_1) \land \Error(\varphi_2))$;
\item if $\varphi$ is $\neg \varphi_1$ then $\Error(\varphi)$ = $\Error(\varphi_1)$.
\end{itemize}
\end{definition}

\begin{lemma}%
  \label{lemma:about-error-formula}
  For every filter condition $\varphi$ and mapping $\mu$ it holds that $\mu(\varphi)=\error$ if and only if $\mu(\Error(\varphi))=\true$.
\end{lemma}

\begin{proof}
  This lemma is proved by induction on the structure of the filter condition (see Claim~\ref{claim:error-formula} in the appendix).
\end{proof}

\begin{example}
  Let $\varphi$ be the filter condition $L \lor \neg L$ where $L$ is the equality $?X = a$.  According to Definition~\ref{def:error-formula}, $\Error(\varphi)$ will be the filter condition $(\neg L \land \Error(\neg L)) \lor (\Error(L) \land \neg \neg L) \lor (\Error(L) \land \Error(\neg L))$.  Since $\neg\neg L$ is equivalent to $L$ and $\Error(\neg L)$ is equivalent to $\Error(L)$, then $\Error(\varphi)$ is equivalent to $(\neg L \land \Error(L)) \lor (\Error(L) \land L) \lor (\Error(L))$, which is equivalent to $(\varphi \land \Error(L)) \lor (\Error(L))$, and then, equivalent to $\Error(L)$.  According to Definition~\ref{def:error-formula}, we conclude that $\Error(\varphi)$ is equivalent to $\neg\bound(?X)$.

  There are three possible values for variable $?X$ in a mapping $\mu$, namely $\mu(?X)=a$, $\mu(?X)=b$ (for a term $b\neq a$), and variable $?X$ is unbound in $\mu$ (denoted $\mu(?X)=\bot$).  The following table shows the values for $\mu(\varphi)$ and $\mu(\Error(\varphi))$ for these three cases.
  \begin{center}
    \begin{tabular}{ccc}
      \toprule
      $\mu(?X)$ & $\mu(\varphi)$ & $\mu(\Error(\varphi))$ \\ \midrule
      $a$ & $\true$ & $\false$ \\
      $b$ & $\true$ & $\false$ \\
      $\bot$ & $\error$ & $\true$ \\
      \bottomrule
    \end{tabular}
  \end{center}
  As defined by Lemma~\ref{lemma:about-error-formula}, the filter condition $\varphi$ produces error for mappings $\mu$ where $\mu(\Error(\varphi))=\true$, and $\mu(\Error(\varphi))$ is either $\true$ or $\false$.
\end{example}

Now we present an equivalence for the disjunction that works in all cases.

\begin{lemma}[Disjunction rewriting]%
  \label{disjunction-rewriting-fixed}
  Given two filter conditions $\varphi_1$ and $\varphi_2$, and a pattern $P$, the following equivalence holds for bag semantics:
  \begin{align}
    &\begin{aligned}[t]
      (P \FILTER \varphi_1 \lor \varphi_2) \equiv\;
      &(P \FILTER \varphi_1 \land \varphi_2) \UNION \\[-5pt]
      &(P \FILTER \varphi_1 \land \neg \varphi_2) \UNION \\[-5pt]
      &(P \FILTER \neg \varphi_1 \land \varphi_2) \UNION \\[-5pt]
      &(P \FILTER \varphi_1 \land \Error(\varphi_2)) \UNION \\[-5pt]
      &(P \FILTER \Error(\varphi_1) \land \varphi_2) . \\
    \end{aligned}
    \label{formula:disj-fixed}
  \end{align}
\end{lemma}

\begin{proof}
  Since $\varphi \lor \neg\varphi \lor \Error(\varphi)$ is a tautology for every filter condition $\varphi$, the following equivalences hold:
  \begin{align*}
    \varphi_1 & \equiv \varphi_1 \land
                (\varphi_2 \lor \neg\varphi_2 \lor \Error(\varphi_2)) 
                \equiv (\varphi_1 \land \varphi_2) \lor
                (\varphi_1 \land \neg\varphi_2) \lor
                (\varphi_1 \land \Error(\varphi_2)), \\
    \varphi_2 & \equiv \varphi_2 \land
                (\varphi_1 \lor \neg\varphi_1 \lor \Error(\varphi_1))
                \equiv (\varphi_2 \land \varphi_1) \lor
                (\varphi_2 \land \neg\varphi_1) \lor
                (\varphi_2 \land \Error(\varphi_1)).
  \end{align*}
  Hence, the following equivalence holds:
  \begin{align*}
    \varphi_1 \lor \varphi_2
    & \equiv
      (\varphi_1 \land \varphi_2) \lor
      (\varphi_1 \land \neg\varphi_2) \lor
      (\neg\varphi_1 \land \varphi_2) \lor
      (\varphi_1 \land \Error(\varphi_2)) \lor
      (\Error(\varphi_1) \land \varphi_2).
  \end{align*}
  Since all filter conditions in the disjunction of the right side of this equivalence are disjoint, by Lemma~\ref{disjunction-rewriting-simple}, we got equivalence~(\ref{formula:disj-fixed}).
\end{proof}

Finally, we provide a translation for filter graph patterns which have a negation.  Under two-valued logic, the evaluation of a pattern $P$ of the form $(P_1 \FILTER \neg \varphi)$ may be understood as ``all solutions $\mu$ of $P_1$ except those where $\mu(\varphi)$ is true.''  Under 3-valued logic, the evaluation of $P$ means ``all solutions $\mu$ of $P$ except those where $\mu(\varphi)$ is true or $\mu(\varphi)$ is error.''  Thus according to the latter meaning we have:

\begin{lemma}[Negation rewriting]%
  \label{lemma:negation-rewriting-fixed}
  Given a filter condition $\varphi$, and a pattern $P_1$, the following equivalence holds:
  \begin{align}
    (P_1\FILTER \neg \varphi) \equiv \;
    ((P_1 
    \EXCEPT\, (P_1 \FILTER \varphi))
    \EXCEPT\, (P_1 \FILTER \Error(\varphi))).
    \label{formula:neg-fixed}
  \end{align}
\end{lemma}

\begin{proof}
  The equivalence follows from the fact that the filter discards from the solutions of $P$ those solutions $\mu$ such that $\mu(\varphi)$ is $\false$ or $\error$.
\end{proof}

Now we are ready to present the effectiveness of rewriting that allows for the reduction of complex filter conditions.

\begin{definition}[Reduction of complex filter conditions]%
  \label{def:formula-reduction}
  Given a pattern $(P_1 \FILTER \varphi)$, the filter-reduced pattern of it is the pattern that results of applying recursively the equivalences (\ref{formula:conj}), (\ref{formula:disj-fixed}), and (\ref{formula:neg-fixed}) until in the resulting patterns only occur atomic formulas (i.e. no logical connectives).
\end{definition}

\begin{lemma}%
  \label{formula-reduction}
  Given a pattern $(P_1 \FILTER \varphi)$, the procedure to reduce complex filter conditions described in Definition~\ref{def:formula-reduction} produces a pattern equivalent to the original and with no logical connectives in filter conditions.
\end{lemma}

\begin{proof}
  This lemma is proved by induction on the structure of the filter condition in the pattern. The base case consists in a filter condition $\varphi$ without logical connectives.  The case where $\varphi$ is $\varphi_1 \land \varphi_2$ is straightforward.  The pattern $(P \FILTER \varphi)$ can be reduced to the pattern $((P \FILTER \varphi_1) \FILTER \varphi_2)$, and the inductive hypothesis can be applied on $\varphi_1$ and $\varphi_2$.  The cases where $\varphi$ is $\varphi_1 \lor \varphi_2$ or $\neg\varphi_1$ are more involved because the application of the respective equivalences eliminates a logical connective from $\varphi$ but adds new logical connectives to the resulting filter conditions.  The proof for the cases involving disjunction or negation follows from Claim~\ref{claim:reduction-complex-formula} in the appendix.
\end{proof}

\begin{table*}[hbt]
  \caption{Definition of function $\delta(P)$ which allows to translate a {\SPARQL} graph pattern $P$ into a set of {\NRMDatalogSN} rules.
  Given a graph pattern $P$, the function $\bar{P}$ returns the variables of $P$ in lexicographical order. Note that $p_i$ is a fresh predicate name used to codify the graph pattern $P_i$.}
  \centering
  \begin{adjustbox}{max width=\textwidth}
    \begin{tabular}{p{3.2cm}p{6cm}p{6cm}}
      \toprule

      \rs Graph pattern $P$
      & $\delta(P)$ 
      & where ...
      \\ \midrule \midrule

      \rs $(x_1,x_2,x_3)$
      & \begin{tabular}{l} 
        $p( \bar{P} ) \gets triple(x_1,x_2,x_3)$
        \end{tabular}
      & $\bar{P}$ contains the variables in the set $\{x_1, x_2, x_3\}$.
      \\ \midrule

      \rs $(P_1 \AAND P_2)$   
      & \begin{tabular}{l}
        $p(\bar{P})$ $\gets$ $\nu_1(p_1(\bar{P}_1))$,
        $\nu_2(p_2(\bar{P}_2))$,\\ 
        \hspace{0.9cm} 
        $\{\comp(\nu_1(X),\nu_2(X),X) \mid X \in \bar{P}_1 \cap \bar{P}_2\}$\\
        $\delta(P_1)$\\
        $\delta(P_2)$
        \end{tabular}
      & Assume that $\nu_1$ and $\nu_2$ are functions with the same domain $\bar{P}_1 \cap \bar{P}_2$ and with different range. Given a literal $L$, the function $\nu_i(L)$ returns a copy of $L$ where the variables have been renamed according to $v_i$. 
      \\ \midrule

      \rs $(P_1 \UNION P_2)$ 
      & \begin{tabular}{l}
        $p(\bar{P}) \gets p_1(\bar{P}_1)$\\
        $p(\bar{P}) \gets p_2(\bar{P}_2)$\\ 
        $\delta(P_1)$\\
        $\delta(P_2)$
        \end{tabular}
      & $\bar{P}$, $\bar{P}_1$ and $\bar{P}_2$ contain the same variables.
      \\ \midrule

      \rs $(P_1 \EXCEPT P_2)$ 
      & \begin{tabular}{l}
        $p(\bar{P}) \gets p_1(\bar{P}_1), \neg p_2(\bar{P}_2)$\\ 
        $\delta(P_1)$\\
        $\delta(P_2)$
        \end{tabular}
      & $\bar{P}$ and $\bar{P}_1$ contain the same variables.
      \\ \midrule

      \rs $(P_1 \FILTER x_1 = x_2)$ 
      & \begin{tabular}{l}
         $p( \bar{P}) \gets p_1(\bar{P}_1),\equal(x_1,x_2)$\\
         $\delta(P_1)$
        \end{tabular}
      & $\bar{P}$ and $\bar{P}_1$ contain the same variables.
      \\ \midrule

      \rs $(P_1 \FILTER \bound(?X))$ 
      & \begin{tabular}{l}
        $p( \bar{P}) \gets p_1(\bar{P}_1), \term(?X)$;\\
        $\delta(P_1)$; 
        \end{tabular}      
      & $\bar{P}$ and $\bar{P_1}$ contain the same variables.
      \\ \midrule

      \rs $(\SELECT~W~P_1)$ 
      & \begin{tabular}{l}
        $p(\bar{P}) \gets p_1(\bar{P}_1), \unull(x_1),\dots,\unull(x_n)$\\
        $\delta(P_1)$
        \end{tabular}
      & $\bar{P} = W$, and $x_1,\dots,x_n$ are the variables that are in $W$ but not in $\inScope(P_1)$.
      \\ \bottomrule
    \end{tabular}
  \end{adjustbox}
  \label{table:pattern2rules}
\end{table*}

\paragraph{Translation of SPARQL graph patterns into Datalog rules.}
The translation essentially follows the idea presented by Polleres~\cite{10150}, adapted to multisets by Angles and Gutierrez~\cite{91353}, and improved by Hernández~\cite{phd-hernandez}.
Specifically, we cover the following issues:
\begin{enumerate}
\item
  It considers the cases where a filter condition is evaluated as error. Some solutions are lost when these cases are not considered.
\item
  It considers that the equality $X=Y$ must be evaluated as true only if $X$ and $Y$ are bound.  The translation is fixed by using the literal $\equal(X,Y)$ instead of a built-in equality $X=Y$.  Since, atom $\equal(X,Y)$ is true only if $X$ and $Y$ are terms in the database, the translation of the filter-condition $X=Y$ is not evaluated as true when $X$ and $Y$ are unbound.
\end{enumerate}

Given a SPARQL graph pattern $P$, the function $\delta(P)$ is introduced to transform $P$ into a set of {\NRMDatalogSN} rules.
The definition of function $\delta$ is given by the translation rules presented in Table~\ref{table:pattern2rules}.
Note that this function implies a recursive translation of graph patterns into Datalog rules, as well as the generation of fresh predicates.
To do this, the application of function $\delta$ over a graph $P$ will be based on exploring the parse tree of $P$.

A \emph{parse tree} for a graph pattern $P$ is an ordered rooted tree whose internal nodes represent the graph pattern operators in $P$ (i.e., AND, UNION, EXCEPT, FILTER, SELECT), and whose leaves represent the triple patterns of $P$. 
Additionally, we will assume that the nodes will be enumerated using a post-order traversal of the tree, that is, for any node $n$, first visit the left subtree, then visit the right subtree, and finally visit the node $n$ itself.

Given a graph pattern $P$ and its parse tree $T$, the execution of the function $\delta$ consists of applying the rules shown in Table~\ref{table:pattern2rules} to each node in $T$, following their enumeration. Next we present an example of this procedure. 


\begin{example}%
  \label{ex:sparql-to-md-query}
  Let $Q$ be the following {\SPARQL} query asking for all people, the place where they live, and optionally the people they know:
  \[
    \begin{aligned}[t]
      ( & ((\Variable{person}, \pLivesIn, \Variable{somewhere}) \AAND\,
          (\Variable{person}, \pKnows, \Variable{somebody})) \\[-4pt]
        & \UNION \\[-4pt]
        & 
          \begin{aligned}[t]
            ( & (\Variable{person}, \pLivesIn, \Variable{somewhere}) \EXCEPT \\[-4pt]
              &
                \begin{aligned}[t]
                  ( & \SELECT \Variable{person}\; \Variable{somewhere} \\[-4pt]
                    & \WHERE\; ((\Variable{person}, \pLivesIn, \Variable{somewhere}) \AAND\,
                      (\Variable{person}, \pKnows, \Variable{somebody}))))).
                \end{aligned}
          \end{aligned}
    \end{aligned}
  \]
This {\SPARQL} query is not normalized because both sides of the $\UNION$ operator have different variables. 
To normalize this query, we replace the right-hand operand of the $\UNION$ pattern by a $\SELECT$ clause, and include the projection variables $\Variable{person}$, $\Variable{somewhere}$ and $\Variable{somebody}$.
The resulting normalized query is the following:
  \[
    \begin{aligned}[t]
      ( & ((\Variable{person}, \pLivesIn, \Variable{somewhere}) \AAND\,
          (\Variable{person}, \pKnows, \Variable{somebody})) \\[-4pt]
        & \UNION \\[-4pt]
        &
          \begin{aligned}[t]
            ( & \SELECT \Variable{person}\; \Variable{somewhere}\; \Variable{somebody} \\[-4pt]
              & \WHERE
                \begin{aligned}[t]
                ( & (\Variable{person}, \pLivesIn, \Variable{somewhere}) \EXCEPT \\[-4pt]
                  & \begin{aligned}[t]
                      ( & \SELECT \Variable{person}\; \Variable{somewhere} \\[-4pt]
                        & \WHERE\; ((\Variable{person}, \pLivesIn, \Variable{somewhere}) \AAND\,
                          (\Variable{person}, \pKnows, \Variable{somebody})))))).
                    \end{aligned}
              \end{aligned}
          \end{aligned}
    \end{aligned}
  \]

\begin{figure}[t!]
\begin{forest}
[$P_{11}(\UNION)$
    [$P_{3}(\AAND)$
        [$P_1$]
        [$P_2$]
    ]
    [$P_{10}(\SELECT)$
        [$P_9(\EXCEPT)$
            [$P_4$]
            [$P_8(\SELECT)$
                [$P_7(\AAND)$
                    [$P_{5}$]
                    [$P_{6}$]
                ]
            ]
        ]
    ]
]
\end{forest}
\caption{Parse tree of the normalized SPARQL graph pattern presented in Example \ref{ex:sparql-to-md-query}. Each node represents a graph pattern, and each edge represents a composition relationship between graph patterns. The nodes were enumerated using a post-order traversal of the tree. 
In this figure:
$P_1$ is the triple pattern $(\Variable{person}, \pLivesIn, \Variable{somewhere})$,
$P_2$ is $(\Variable{person}, \pKnows, \Variable{somebody})$,
$P_3$ is $(\Variable{person}, \pLivesIn, \Variable{somewhere})$,
$P_4$ is $(\Variable{person}, \pLivesIn, \Variable{somewhere})$, 
and
$P_5$ is $(\Variable{person}, \pKnows, \Variable{somebody})$.
}
\label{fig:query-tree}
\end{figure}
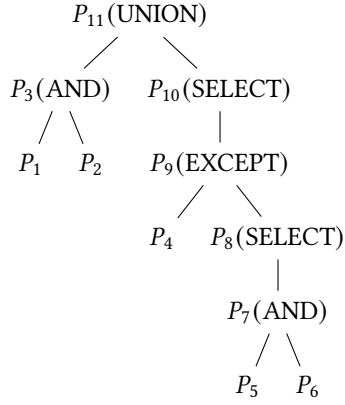

Assume that $P$ is the normalized SPARQL graph pattern presented above, and $T$ is the parse tree of $P$ shown in Figure \ref{fig:query-tree}.
Hence, the application of the function $\delta(P)$ results in the following Datalog rules:
\begin{verse}
$p_1(?person,?somewhere) \gets \triple(?person, \pLivesIn, ?somewhere)$\\
$p_2(?person,?somebody) \gets \triple(?person, \pKnows, ?somebody)$\\
$p_3(?person, ?somewhere, ?somebody) \gets p_1(?person_1,?somewhere),$\\
\hspace{6cm}$p_2(?person_2,?somebody),$\\
\hspace{6cm}$\comp(?person_1,?person_2,?person)$\\
$p_4(?person,?somewhere) \gets \triple(?person, \pLivesIn, ?somewhere)$\\
$p_5(?person,?somewhere) \gets \triple(?person, \pLivesIn, somewhere)$\\
$p_6(?person,?somebody) \gets \triple(?person, \pKnows, ?somebody)$\\
$p_7(?person, ?somewhere, ?somebody) \gets p_5(?person_1,?somewhere),$\\
\hspace{6cm}$p_6(?person_2,?somebody),$\\
\hspace{6cm}$\comp(?person_1,?person_2,?person)$\\
$p_8(?person, ?somewhere) \gets p_7(?person, ?somewhere, ?somebody)$\\
$p_9(?person, ?somewhere) \gets p_4(?person,?somewhere), \neg p_8(?person,?somewhere)$\\
$p_{10}(?person, ?somewhere, ?somebody) \gets p_9(?person, ?somewhere),$\\
\hspace{6cm}$\unull(?somebody)$\\
$p_{11}(?person, ?somewhere, ?somebody) \gets p_3(?person, ?somewhere, ?somebody)$\\
$p_{11}(?person, ?somewhere, ?somebody) \gets p_{10}(?person, ?somewhere, ?somebody)$\\
\end{verse}

Note that the predicate names are enumerated according to the parse tree shown in Figure \ref{fig:query-tree}.
In this sense:
the rule $p_1(\dots) \gets \dots$ encodes the triple pattern $P_1$;
the same applies for the predicated names $p_2$, $p_4$, $p_5$ and $p_6$;
the rule $p_3(\dots) \gets \dots$ encodes the pattern $(P_1 \AAND P_2)$;
the rule $p_7(\dots) \gets \dots$ encodes the pattern $(P_5 \AAND P_6)$;
the rule $p_8(\dots) \gets \dots$ encodes the pattern $(\SELECT \Variable{person},$ $\Variable{somewhere} \WHERE P_7)$;
the rule $p_9(\dots) \gets \dots$ encodes the pattern $(P_4 \EXCEPT P_8)$;
the rule $p_{10}(\dots) \gets \dots$ encodes the pattern $(\SELECT \Variable{person},$ $\Variable{somewhere},$ $\Variable{somebody} \WHERE P_9)$;
finally, the two rules with predicate name $p_{11}$ encode the graph pattern $(P_3 \UNION P_{10})$.

\end{example}

Based on the function $\delta$, we define a general method to transform {\SPARQL} queries into {\NRMDatalogSN} queries.

\begin{definition}[Function $f_{12}$]%
\label{def:patterns-to-datalog}
Given a {\SPARQL} query $Q$, the function $f_{12}(P)$ returns a {\NRMDatalogSN} query $(L,\Pi)$ where $L$ is the goal atom $p(\bar{P})$ and $\Pi$ is a Datalog program containing the rules produced by $\delta(Q)$ (note that $Q$ is a graph pattern).
\end{definition}

\subsubsection{Translating query answers from {\NRMDatalogSN} to {\SPARQL}}
Recall that a {\NRMDatalogSN} query answer is a pair $(V,M)$ where $V$ is a set of variables and $M$ is a multiset of substitutions. Additionally, a {\SPARQL} query answer is a multiset of solution mappings $\Omega$.

The main difference between a multiset of substitutions $M$ and a multiset of solution mappings $\Omega$ is the representation of null values.
Given a substitution $\theta$ and a variable $?X \in \dom(\theta)$, a null value is represented with the assignment $\theta(?X) = \bot$. On the other hand, for a solution mapping $\mu$, a null value for $?X$ is represented by not including $?X$ in the domain of $\mu$.
Given a Datalog substitution $\theta$, the function $\NotNull(\theta)$ returns a SPARQL solution mapping $\mu$ defined as follows: (i) $\dom(\mu) = \{ x \in \dom(\theta) \mid \theta(x) \neq \bot \}$, that is, the variables of $\mu$ are the variables in $\theta$ whose assignment is not $\bot$; and, (ii) for every variable $?X \in \dom(\theta)$, if $\theta(?X) \neq \bot$ then $\mu(?X) = \theta(?X)$.  
Hence, every value different from $\bot$ is translated into an unbound variable in the corresponding solution mapping.

\begin{definition}[Function $h_{12}$]%
\label{def:substitutions-to-mappings}
Given a multiset of Datalog substitutions $\Theta$, the function $h_{12}(\Theta)$ returns a multiset of SPARQL solution mappings $\Omega$ defined as follows: 
(i)
$\set(\Omega) = \{ \NotNull(\theta) \mid \theta \in \set(\Theta) \}$;  and, 
for each substitution $\theta \in \set(\Theta)$, it applies that  
$\card(\NotNull(\theta), \Omega) = \card(\theta, \Theta)$.
In other words, the cardinality of each solution mapping $\NotNull(\theta)$ in $\Omega$ is defined as the number of colored copies $\colCopy{\theta}{i}$ occurring in $\col(\Theta)$.
\end{definition}

\begin{lemma}%
  \label{lemma:sparql2datalog}
  {\SPARQL} can be simulated by {\NRMDatalogSN}.
\end{lemma}

\begin{proof}
We need to show that, using the functions defined above, $(f_{12},g_{12},h_{12})$ is a simulation of {\SPARQL} in {\NRMDatalogSN}.
The proof is in the Claim~\ref{claim:sparql2md} of the Appendix.
\end{proof}

\subsection{From {\NRMDatalogSN} to {\SPARQL}}
\label{sec:datalog2sparql}
This section shows that Non-Recursive Multiset Datalog with Safe Negation ({\NRMDatalogSN}) can be simulated by {\SPARQL}.
To support this, we describe the following translation functions:
\begin{itemize}
\item $f_{21}$ is the function that translates {\NRMDatalogSN} queries into {\SPARQL} queries; 
\item $g_{21}$ is the function that translates {\NRMDatalogSN} databases into {\SPARQL} databases; and
\item $h_{21}$ is the function that translates {\SPARQL} query answers into {\NRMDatalogSN} query answers.
\end{itemize}

\subsubsection{Translating databases from {\NRMDatalogSN} to {\SPARQL}}
\label{sec:md-to-sparql-database}

In general terms, a fact $p(c_1,\dots,c_n)$ can be translated into a set of triples of the form $(u,\alpha_i,c_i)$ where $u$ is a fresh IRI that identifies the fact, and $\alpha_i$ is a reserved IRI which allows to describe that constant $c_i$ is in the position $i$ of the fact\footnote{An option can be the use of properties \texttt{rdf:\_1}, \texttt{rdf:\_2}, \texttt{rdf:\_3}, $\dots$, defined in the RDF Schema 1.1 vocabulary.}. Also recall that the semantics of {\NRMDatalogSN} relies on the notion of colored set of a multiset (see Section~\ref{ss:datalogsemantics}), which is the set containing the colored copies of the element of the multiset. This idea is formalized next.

In what follows, we will assume that $A=\{\alpha_0,\alpha_1\dots\}$ is an enumerable set of special IRIs used to codify positions in Datalog atoms, 
$\NULL$ is a special IRI, and any Datalog constant $c$ and name of predicate $p$ has an equivalent  {RDF} IRI term (excluding the aforementioned special IRIs) that we will denote with the same symbol $c$ and $p$.

\begin{definition}[Function $g_{21}$]
\label{def:g-datalog-sparql}
Assume that
the {\NRMDatalogSN} database $D$ contains $n$ copies of a fact $F$ (namely $p(c_1,\dots,c_n)$), and the $\col(D)$ contains the colored copies $\colCopy{F}{1}$, \dots, $\colCopy{F}{n}$ of fact $F$. For each colored copy $\colCopy{F}{i}$ of $F$, we assume the existence of a fresh IRI $u_{\colCopy{F}{i}}$, which we use to identify the colored copy.
Similarly the subscript $_{\colCopy{F}{i}}$ is used to identify the colored copy of a constant or a predicate.

Then  the function $g_{21}$ applied to the multiset of {\NRMDatalogSN} facts $D$, returns a set of RDF triples (i.e. an RDF graph) defined as
\begin{multline*}
    g_{21}(D) = \{(\NULL,\NULL,\NULL)\} \bigcup_{\colCopy{F}{i} \in \col(D)} \\
    \{(u_{\colCopy{F}{i}},\alpha_0,p_{\colCopy{F}{i}}),(u_{\colCopy{F}{i}},\alpha_1,(c_1)_{\colCopy{F}{i}}),\dots,(u_{\colCopy{F}{i}},\alpha_n,(c_n)_{\colCopy{F}{i}})\}.
\end{multline*}

\end{definition}

\begin{example}
    Let $D$ be the following {\NRMDatalogSN} database:
    \[
      D =
      \lBrace p(a, b), p(a, b), p(a, c), q(b, d, a), q(b, e, a) \rBrace.
    \]
Note that the fact $p(a, b)$ occurs twice in the multiset $D$, so we need to generate the colored copies for this fact, namely $\colCopy{p(a, b)}{1}$ and $\colCopy{p(a, b)}{2}$. In contrast, the fact $p(a, b)$ occurs once, so we just generate a simple colored copy, that is, $\colCopy{p(a, c)}{1}$. 
    Then the data is translated for SPARQL as follows:
    \[
      g_{21}(D) =
      \begin{aligned}[t]
        \{ & (\NULL,\NULL,\NULL), \\[-4pt]
           & (u_{\colCopy{p(a, b)}{1}}, \alpha_0, p),
             (u_{\colCopy{p(a, b)}{1}}, \alpha_1, a),
             (u_{\colCopy{p(a, b)}{1}}, \alpha_2, b), \\[-4pt]
           & (u_{\colCopy{p(a, b)}{2}}, \alpha_0, p),
             (u_{\colCopy{p(a, b)}{2}}, \alpha_1, a),
             (u_{\colCopy{p(a, b)}{2}}, \alpha_2, b), \\[-4pt]
           & (u_{\colCopy{p(a, c)}{1}}, \alpha_0, p),
             (u_{\colCopy{p(a, c)}{1}}, \alpha_1, a),
             (u_{\colCopy{p(a, c)}{1}}, \alpha_2, c), \\[-4pt]
           & (u_{\colCopy{p(b, d, a)}{1}}, \alpha_0, q),
             (u_{\colCopy{p(b, d, a)}{1}}, \alpha_1, b),
             (u_{\colCopy{p(b, d, a)}{1}}, \alpha_2, d),
             (u_{\colCopy{p(b, d, a)}{1}}, \alpha_3, a), \\[-4pt]
           & (u_{\colCopy{p(b, e, a)}{1}}, \alpha_0, q),
             (u_{\colCopy{p(b, e, a)}{1}}, \alpha_1, b),
             (u_{\colCopy{p(b, e, a)}{1}}, \alpha_2, d),
             (u_{\colCopy{p(b, e, a)}{1}}, \alpha_3, a)
             \} .
      \end{aligned}
    \]
  \end{example}

Intuitively, the {\SPARQL} database corresponding to the {\NRMDatalogSN} database $D$ consists of a set of triples that describe each of the facts, and the inclusion of triple $(\NULL,\NULL, \NULL)$ allows to ensure that the {\SPARQL} database is not empty. The need of this additional triple is explained next when describing the translation from {\NRMDatalogSN} queries to {\SPARQL}.

\subsubsection{Translating queries from {\NRMDatalogSN} to {\SPARQL}}

A notable difference between {\NRMDatalogSN} and {\SPARQL} is the way both languages define the scope of variables. In {\NRMDatalogSN}, all variables in a rule are universally quantified, and they are not in the scope of the query. On the other hand, variables in a {\SPARQL} query are divided into in-scope and non-in-scope (see Subsection~\ref{sec:normalization-of-patterns}). To see this difference, consider the {\NRMDatalogSN} query $(q(X,Y), \Pi)$ where the program $\Pi$ consists of the single rule $R = q(Y,Z) \gets p(X,Z,Y)$. Notice that the variables in the goal of the query do not correspond to the variables in the head of the rule $R$. To simplify the translation, we rename variables in rules according to the goal of the query.

In this case, we rewrite $R$ as the rule $R' = q(X,Y) \gets p(Z,Y,X)$. 
Formally, given a literal $L=q(X_1,\dots,X_n)$ and a rule $R$ whose head is $q(Y_1,\dots,Y_n)$, the \emph{renamed rule of $R$ with respect to $L$}, denoted $\operatorname{vr}(R,L)$, is the rule $R'$ 
that results from $R$ by consistently renaming each variable $Y_i$ as $X_i$, for $1 \leq i \leq n$.

Let $L$ be a positive literal $p(X_1,\dots,X_n)$ and $\Pi$ be a normalized {\NRMDatalogSN} program. We define the function $\gpattern(L,\Pi)$ which translates $L$ into a {\SPARQL} graph pattern.  The function $\gpattern$ is defined recursively as follows:
\begin{enumerate}
\item 
  If predicate name $p$ does not occur in the head of any rule of $\Pi$ (i.e., $p$ is extensional), then $\gpattern(L,\Pi)$ returns
  \begin{align*}
    \SELECT \overline{X}\;
    ((?Y, \alpha_0 ,p) \AAND\, (?Y, \alpha_1 ,?X_1) \AAND \cdots \AAND\, (?Y, \alpha_n, ?X_n)),    
  \end{align*}
  where $\overline{X} = \var(L)$ and $?Y$ is a fresh variable.
\item 
  Otherwise, if $p$ occurs in the head of the rules  $\{R_1,\dots, R_n\}$ in $\Pi$  (i.e., $p$ is intensional), then $\gpattern(L,{\Pi})$ returns:
  \[ (T(\operatorname{vr}(R_1,L))
    \UNION \cdots \UNION
    T(\operatorname{vr}(R_n,L))),\] 
  where the operator $T(R)$ is defined as follows:
  \begin{itemize}
  \item 
    If $R$ is a projection rule $L \leftarrow L_1$ then $T(R)$ is $(\SELECT \overline{X}\; P_1 )$ where $\overline{X} = \var(L)$ and $P_1 = \gpattern(L_1,\Pi)$;
  \item 
    If $R$ is a join rule $L \leftarrow L_1, L_2$ then $T(R)$ is $( P_1 \AAND P_2)$ where $P_1 = \gpattern(L_1,\Pi)$ and $P_2 = \gpattern(L_2,\Pi)$;
  \item 
    If $R$ is a negation rule $L \leftarrow L_1, \neg L_2$ then $T(R)$ is $( P_1 \EXCEPT P_2 )$ where $P_1 = \gpattern(L_1,\Pi)$ and $P_2 = \gpattern(L_2,\Pi)$.
  \end{itemize}
  Note that, if there is just one rule $R_1$ then $\gpattern(L,\Pi)$ can be reduced to $T(R_1)$ (no need to rename variables).
\end{enumerate}

\begin{example}%
  \label{ex:empty-MD-border-case}
  Consider the {\NRMDatalogSN} query $(q(X), \Pi)$ where program $\Pi$ consists of the rule $q(X) \gets p(X,Y)$. Then, $\gpattern(q(X),\Pi)$ is the {\SPARQL} query
  \[ \SELECT { } ?X \;((?U, \alpha_0, p) \AAND\, (?U, \alpha_1, ?X) \AAND\, (?U, \alpha_2, ?Y)).\]
\end{example}

The function $\gpattern$ is not enough to translate {\NRMDatalogSN} queries to {\SPARQL} queries. Recall that a {\NRMDatalogSN} query answer is a pair $(V,M)$ where $V$ is a set of {\NRMDatalogSN} variables and $M$ is a set of {\NRMDatalogSN} substitutions, and a {\SPARQL} query answer is a multiset $\Omega$ of {\SPARQL} mappings. To conclude the translation, we need to define a function that, given a {\SPARQL} query answer $\Omega$, returns a {\NRMDatalogSN} query answer $(V,M)$. The issue is that we cannot compute the set $V$ when the multiset $\Omega$ is empty. For example, an empty {\NRMDatalogSN} database $D$ is translated as the {\SPARQL} database consisting of the set of triples $\{ (\NULL,\NULL,\NULL) \}$ (see Subsection~\ref{sec:md-to-sparql-database}). The evaluation of the query $\gpattern(q(X),\Pi)$ in Example~\ref{ex:empty-MD-border-case} returns an empty multiset of mappings, $\Omega$, where the query answer to the {\NRMDatalogSN} query $(q(X),\Pi)$ is a pair $(\{X\}, M)$ such that $M$ is an empty multiset of solutions. Hence, the {\SPARQL} query answer $\Omega$ does not contain the information needed to generate the set of variables $\{X\}$ in the answer of the {\NRMDatalogSN} query.

To solve the aforementioned issue of having an empty {\SPARQL} query answer, we can extend the function $\gpattern$ with a query that introduces the variables of the query. This is done using the additional triple $(\NULL,\NULL,\NULL)$ we introduced in the translation. Given a set of {\NRMDatalogSN} variables $V=\{X_1,\dots,X_n\}$ we write $\operatorname{VarQuery(V)}$ to denote the {\SPARQL} pattern
\[
(\NULL, \NULL, ?X_1)\; \AAND\, \cdots\, \AAND\; (\NULL, \NULL, ?X_n).
\]
The translation of a {\NRMDatalogSN} query is then the union of the graph patterns computed by the functions $\gpattern$ and $\operatorname{VarQuery}$.

\begin{definition}[Function $f_{21}$]
Given a {\NRMDatalogSN} query $Q = (L, \Pi)$, the function $f_{21}(Q)$ returns a {\SPARQL} graph pattern $(\gpattern(L,\Pi) \UNION \operatorname{VarQuery}(\var(L)))$.
\end{definition}

\begin{example}
  Consider the {\NRMDatalogSN} query $(q(X),\Pi)$ in Example~\ref{ex:empty-MD-border-case}. Then, $f_{21}((q(X),\Pi))$ is the following {\SPARQL} graph pattern:
  \[(\SELECT { } ?X \;((?U, \alpha_0, p) \AAND\, (?U, \alpha_1, ?X) \AAND\, (?U, \alpha_2, ?X)))
  \UNION\, (\NULL, \NULL, ?X).
  \]
  The result of evaluating the {\NRMDatalogSN} query on an empty set of facts $D$ is the pair $(\{X\}, M)$ where $M$ is an empty multiset of {\NRMDatalogSN} substitutions, whereas the result of evaluating the graph pattern $f_{21}((q(X),\Pi))$ on the {\SPARQL} database $g_{21}(D)=\{(\NULL,\NULL,\NULL)\}$ is the {\SPARQL} query answer $\Omega=\{\{?X \mapsto \NULL \}\}$. Intuitively, the mapping $\{?X \mapsto \NULL\}$ does not codify a {\NRMDatalogSN} substitution, but the variables in the domain of {\NRMDatalogSN} substitutions.
\end{example}

\subsubsection{Translating query answers from {\SPARQL} to {\NRMDatalogSN}}
Recall that a {\SPARQL} query answer is a multiset of solution mappings, and a {\NRMDatalogSN} query answer is a pair $(V,M)$ (where $V$ is a set of variables and $M$ is a multiset of substitutions).
Since a {\SPARQL} solution mapping can be seen as a {\NRMDatalogSN} substitution, the translation from {\SPARQL} mappings to {\NRMDatalogSN} substitutions does not require modifications, except for the mapping $\{?X_1 \mapsto \NULL, \dots, ?X_n \mapsto \NULL\}$ which is used to codify the solution variables.

\begin{definition}[Function $h_{21}$]
  Let $\Omega$ be a multiset of {\SPARQL} solution mappings that includes a mapping $\mu_{V\mapsto \NULL}$ with cardinality 1 where $\dom(\mu_{V \mapsto \NULL})=V$ and $\mu(?X) = \NULL$ for every variable $?X \in \dom(\mu_{V \mapsto \NULL})$, and for every mapping $\mu' \in \Omega$ it holds that $\dom(\mu') = V$.  The {\NRMDatalogSN} solution for $\Omega$, denoted $h_{21}(\Omega)$, is the pair $(V,M)$ where $M$ is the multiset of substitutions $\theta$ defined as follows:
  \begin{enumerate}
  \item
    Given an {\SPARQL} mapping $\mu = \{?X_1 \mapsto c_1,\dots, ?X_n \mapsto c_n\}$ the corresponding {\NRMDatalogSN} substitution for mapping $\mu$ is the substitution $\theta_{\mu} = 
    \{ X_1 \mapsto c_1,\dots, X_n \mapsto c_n\}$
    where the {\NRMDatalogSN} variable $X_i$ corresponds to the {\SPARQL} variable $?X_i$.
  \item
    $\set(M) = \{ \theta_\mu \mid \mu \in \Omega \setminus
    \{\mu_{V \mapsto \NULL}\}\}$.
  \item
    $\card(\theta_\mu, M) = \card(\mu, \Omega)$.
  \end{enumerate}
\end{definition}

\begin{lemma}%
  \label{lemma:datalog2sparql}
  {\NRMDatalogSN} can be simulated by {\SPARQL}.
\end{lemma}

\begin{proof}
This is a long but straightforward induction on Datalog queries using as hypothesis that $(f_{21},g_{21},h_{21})$ is a simulation of {\NRMDatalogSN} in {\SPARQL}. The details of this proof are in the appendix (Claim~\ref{claim:md2sparql}).
\end{proof}

\subsection{{\SPARQL} and {\NRMDatalogSN} have the same expressive power}
Putting together the simulations among {\SPARQL} and {\NRMDatalogSN} stated in this section, we get the following theorem:

\begin{theorem}
  {\SPARQL} and {\NRMDatalogSN} have the same expressive power.
\end{theorem}

\begin{proof}
The claim is based on the simulation of {\SPARQL} with {\NRMDatalogSN} (Lemma \ref{lemma:sparql2datalog}) and the simulation of {\NRMDatalogSN} with {\MRA} (Lemma  \ref{lemma:datalog2sparql}).
\end{proof}

\section{Equivalence between {\MRA} and {\NRMDatalogSN}}
\label{sec:mra-nrmd}
This section presents the simulations that prove that Multiset Relational Algebra ({\MRA}) and Non-Recursive Multiset Datalog with Safe Negation ({\NRMDatalogSN}) have the same expressive power. 
Specifically, we show that {\MRA} can be simulated by {\NRMDatalogSN} (Section \ref{sec:mra2datalog}), and {\NRMDatalogSN} can be simulated by {\MRA} (Section \ref{sec:datalog2mra}). 

\subsection{From {\MRA} to {\NRMDatalogSN}}
\label{sec:mra2datalog}
This section shows that Multiset Relational Algebra ({\MRA}) can be simulated by Non-Recursive Multiset Datalog with Safe Negation ({\NRMDatalogSN}).
To support this, we describe the following translation functions:
\begin{itemize}
\item $f_{32}$ is the function that translates {\MRA} queries into {\NRMDatalogSN} queries; 
\item $g_{32}$ is the function that translates {\MRA} databases into {\NRMDatalogSN} databases; and
\item $h_{32}$ is the function that translates {\NRMDatalogSN} query answers into {\MRA} query answers.
\end{itemize}

\subsubsection{Translating databases from {\MRA} to {\NRMDatalogSN}}
Recall that a {\MRA} database is a set of relations (where each relation is a multiset of tuples), and a {\NRMDatalogSN} database is a multiset of facts. First, we define a method to translate a {\MRA} relation into a multiset of facts. Then, we define a method to translate a set of {\MRA} relations into a multiset of {\NRMDatalogSN} facts.

Assume the existence of functions that map: {\MRA} relation names to {\NRMDatalogSN} predicate names, {\MRA} attributes to {\NRMDatalogSN} variables, and {\MRA} constants to {\NRMDatalogSN} constants.
Given a relation schema $R$, we write $\vec{R}$ to denote a tuple containing the attributes of $R$ in lexicographical order.

Given a multiset relation $r$, defined over a relation schema $R$, with $\vec{R}=(A_1,\dots,A_n)$, the function $\Sigma(r)$ returns a multiset of Datalog facts defined as follows: For each tuple $t$ in $r$, $\Sigma(r)$ contains a fact $f$ of the form $p(c_1,\dots,c_n)$ where $p$ is the image of $R$, every $c_i$ is $t(A_i)$, and the cardinality of $f$ in $\Sigma(r)$ is given by the cardinality of $t$ in $r$. 

\begin{definition}[Function $g_{32}$]%
  \label{def:mra2md-database}
  Given a {\MRA} database $D$, the function $g_{32}$ returns a multiset of {\NRMDatalogSN} facts $D'$ defined as follows:
  \begin{enumerate}
  \item For each {\MRA} relation $r$ in $D$, $D'$ contains the facts returned by $\Sigma(r)$;  
  \item For each constant $c$ in $D$, $D'$ contains a fact $\equal(c,c)$.
  \end{enumerate}
\end{definition}

\begin{example}
  Let $D$ be an {\MRA} dataset consisting in two relations $r$ and $s$ with respective relation schemas $R$ and $S$ with $\vec{R} = (A_1,A_2)$ and $\vec{S} = (A_1,A_3)$, and defined as follows:
  \begin{align*}
    r &= \lBrace
        \{A_1 \mapsto a_1, A_2 \mapsto a_2\},
        \{A_1 \mapsto a_1, A_2 \mapsto a_2\},
        \{A_1 \mapsto a_1, A_2 \mapsto a_3\} \rBrace,\\
    s &= \lBrace \{ A_1 \mapsto a_1, A_3 \mapsto a_4\} \rBrace.
  \end{align*}
  Then, the corresponding {\NRMDatalogSN} dataset is the following:
  \[
    g_{32}(D) =
    \lBrace
    p_R(a_1, a_2), p_R(a_1, a_2), p_R(a_1, a_3),
    p_S(a_1, a_4),
    \equal(a_1, a_1),
    \equal(a_2, a_2),
    \equal(a_3, a_3),
    \equal(a_4, a_4)
    \rBrace,
  \]
  where predicates $p_R$ and $p_S$ are the corresponding images for the relation schemas $R$ and $S$.
\end{example}

Note that the multiset of Datalog facts $D'$ is defined over the vocabulary that includes as predicate names all the relation names in $D$, and as arity of the predicate name $R$ the number of attributes of the relation name $R$.

\subsubsection{Translating queries from {\MRA} to {\NRMDatalogSN}}
Recall that a {\MRA} query is a relational algebra expression, and a {\NRMDatalogSN} query is a set of rules.

First, we need to provide a recursive method to reduce {\MRA} selection formulas into atomic formulas.
Such method is based on the following equivalences where $E$ is an {\MRA} expression, and $\psi$, $\psi_1$, and $\psi_2$ are selection formulas:
\begin{align}
  \sigma_{\psi_1 \land \psi_2}(E)
  &\equiv
    \sigma_{\psi_2}(\sigma_{\psi_1}(E)),
    \label{formula:conj-mra} \\
  \sigma_{\psi_1 \lor \psi_2}(E)
  &\equiv
    \sigma_{\psi_1 \land \neg\psi_2}(E)\cup
    \sigma_{\neg\psi_1 \land \psi_2}(E)\cup
    \sigma_{\psi_1 \land \psi_2}(E),
    \label{formula:disj-mra}
  \\
  \sigma_{\neg\psi}(E)
  &\equiv
    E \setminus \sigma_{\psi}(E).
  \label{formula:diff-mra}
\end{align}
The proof of the validity of the above equivalences follows directly from the semantics of the selection operator. In particular, Equivalence~\ref{formula:disj-mra} is rather involved because separates the disjunction in a union of three disjoint multiset relations in order to preserve the cardinality of each solution.  Using the above equivalence, we get the following lemma.

\begin{lemma}
  For every {\MRA} expression $E$, there exists an equivalent {\MRA} expression $E'$ satisfying that all selection formulas in $E'$ are atomic.
\end{lemma}

\begin{proof}
  The proof follows from induction in the number $k$ of Boolean connectives in selection formulas occurring in an {\MRA} expression $E$.  The base case is $k = 0$ and thus all selection formulas are atomic.  If $k > 0$, then the expression includes a selection expression whose formula has either the form $\psi_1 \land \psi_2$, $\neg \psi$, or $\psi_1 \lor \psi_2$.  In the first two cases, equivalences \ref{formula:conj-mra} and \ref{formula:diff-mra} reduce by one of the Boolean connectives of the expression.  In the third case, the consecutive application of equivalences \ref{formula:disj-mra}, \ref{formula:conj-mra}, and \ref{formula:diff-mra} (in that order) reduces by one the number of Boolean connectives.  Hence, we produce an equivalent query with $k - 1$ Boolean connectives.
\end{proof}

\begin{definition}[Function $f_{32}$]%
  \label{def:mra2datalog-query}
  Let $Q$ be a {\MRA} query (i.e. an {\MRA} expression), where selection formulas are atomic (i.e., have no Boolean connectives).  The function $f_{32}(Q)$ returns a {\NRMDatalogSN} query $(L,\Pi)$ where $L$ is a goal clause of the form $q(\vec{Q})$ where $q$ is a predicate name corresponding to $Q$, $\vec{Q}$ are the variables corresponding to the attributes in the schema $\widehat{Q}$ (sorted in lexicographical order), and $\Pi$ is a set of {\NRMDatalogSN} rules (i.e. a {\NRMDatalogSN} program) created by applying recursively the rules shown in Table~\ref{table:mra2rules}.
\end{definition}

Note that function $f_{32}$ assigns a fresh predicate name to each operation in query $Q$ by following a non-deterministic approach. Although it is not difficult to define deterministic ways (like in the translation from {\SPARQL} to {\NRMDatalogSN}), we omit in how intensional predicate names are assigned.

\begin{table*}[htb]
  \caption{Definition of function $\Gamma$ which translates an {\MRA} expression into a set of Datalog rules. Given a {\MRA} expression $E$, the recursive function $\Gamma(E)$ returns a set of {\NRMDatalogSN} rules where: $q_i(\bar{A})$ is a positive literal related to the {\MRA} expression $E_i$, $q_i$ is a fresh predicate name, $\bar{A_i}$ denotes a set of variables, $\vec{R}$ denotes the attributes in schema $\widehat{R}$, sorted in lexicographical order.}
  \centering
  \begin{adjustbox}{max width=\textwidth}
    \begin{tabular}{p{3cm}p{8cm}p{4cm}}
      \toprule

      \rs {\MRA} expression $E_0$
      & $\Gamma(E_0)$
      & where ...
      \\ \midrule

      \rs $R$
      & $q_0( \bar{A}_0 ) \gets R( \vec{R} )$
      & $\bar{A}_0 = \vec{R}$
      \\ \midrule

      \rs $(E_1 \Join E_2)$   
      & $q_0( \bar{A}_0) \gets q_1(\bar{A}_1),q_2(\bar{A}_2)$;
        $\Gamma(E_{1})$; $\Gamma(E_{2})$
      & $\bar{A}_0 = \bar{A}_1 \cup \bar{A}_2$  
      \\ \midrule

      \rs $(E_1 \cup E_2)$
      & $q_0(\bar{A}_0) \gets q_1(\bar{A}_1)$; 
        $q_0(\bar{A}_0) \gets q_2(\bar{A}_2)$;
        $\Gamma(E_{1})$; $\Gamma(E_{2})$
      & $\bar{A}_ 0 = \bar{A}_1 = \bar{A}_2$  
      \\ \midrule

      \rs $(E_1 \setminus E_2)$
      & $q_0(\bar{A}_0) \gets q_1(\bar{A}_1), \neg q_2(\bar{A}_2)$;
        $\Gamma(E_{1})$; $\Gamma(E_{2})$
      & $\bar{A}_0 = \bar{A}_1$
      \\ \midrule

      \rs $\pi_{S}(E_1)$
      & $q_0(\bar{A}_0) \gets q_1(\bar{A}_1)$;
        $\Gamma(E_{1})$
      & $\bar{A}_0 = S$
      \\ \midrule

      \rs $\rho_{A/B}(E_1)$
      & $q_0(\bar{A}_0) \gets q_1(\bar{A}_1), \equal(A,B)$; 
        $\Gamma(E_{1})$
      & $\bar{A}_0 = (\bar{A}_1 \setminus \{A\} ) \cup \{ B \}$
      \\ \midrule
      
      \rs $\sigma_{A = B}(E_1)$
      & $q_0(\bar{A}_0) \gets q_1(\bar{A}_1), \equal(A,B)$;
        $\Gamma(E_{1})$
      & $\bar{A_0} = \bar{A}_1$
      \\
      \bottomrule
    \end{tabular}
  \end{adjustbox}
  \label{table:mra2rules}
\end{table*}
  
\subsubsection{Translating query answers from {\NRMDatalogSN} to {\MRA}}
Recall that a {\NRMDatalogSN} query answer is a pair $(V,M)$ where $V$ is a set of variables, and $M$ is a multiset of {\NRMDatalogSN} substitutions. On the other hand, an {\MRA} query answer is a multiset relation. Next, we define a function $h_{32}$ which translates a {\NRMDatalogSN} query answer into a {\MRA} query answer.

\begin{definition}[Function $h_{32}$]%
  \label{def:mra2datalog-answer}
  Given a {\NRMDatalogSN} query answer $A = (V,M)$, the function $h_{32}(A)$ returns a multiset relation $r$ where: the schema of $r$ is given by the set of attributes $V$ (assume a simple transformation of variables to attribute names); for each substitution $\theta$ in $M$, there is a tuple $t$ in $r$ satisfying that $t(X) = \theta(X)$ for every attribute $X \in \widehat{R}$, and $\card(t, r) = \card(\theta, M)$.
\end{definition}

\begin{lemma}%
  \label{lemma:mra2datalog}
  {\MRA} can be simulated by {\NRMDatalogSN}.
\end{lemma}

\begin{proof}
  Let $f_{32}$, $g_{32}$, $h_{32}$ be the functions described in Definition~\ref{def:mra2md-database}, Definition~\ref{def:mra2datalog-answer} and Definition~\ref{def:mra2datalog-query} respectively. 
  The proof of this theorem follows from the claim that $(f_{32}$, $g_{32}$, $h_{32})$ simulates {\MRA} in {\NRMDatalogSN} by using induction in the structure of queries. The proof of this claim is in the appendix (Claim~\ref{claim:mra2md}).
\end{proof}

\subsection{From {\NRMDatalogSN} to {\MRA}}
\label{sec:datalog2mra}
This section shows that Non-Recursive Multiset Datalog with Safe Negation ({\NRMDatalogSN}) can be simulated by Multiset Relational Algebra ({\MRA}).
To support this, we describe the following translation functions:
\begin{itemize}
\item $f_{23}$ is the function that translates {\NRMDatalogSN} queries into {\MRA} queries; 
\item $g_{23}$ is the function that translates {\NRMDatalogSN} databases into {\MRA} databases; and
\item $h_{23}$ is the function that translates {\MRA} query answers into {\NRMDatalogSN} query answers.
\end{itemize}

\subsubsection{Translating databases from {\NRMDatalogSN} to {\MRA}}
Recall that a database in {\NRMDatalogSN} is a multiset of facts, and a database in {\MRA} is a set of relations (where each relation is a multiset of tuples).  First, we define a method to translate a multiset of facts with the same predicate name into a relation $r$.  Let $M$ be a multiset of {\NRMDatalogSN} facts having the same predicate name, i.e. every fact in $M$ has the form $p(t_1,\dots,t_n)$.  The function $\psi(M)$ returns a {\MRA} relation $r$ where: the relation schema $\widehat{r}$ of $r$ is given by the relation name $p$ and the set of attributes $\{ A_1, \dots, A_n \}$, where each attribute name has the form $att\_i$ with $1 \leq i \leq n$; for each fact $p(t_1,\dots,t_n)$ in $M$ there is a tuples $t$ in $r$ satisfying that $t(A_i) = t_i$.

Next, we define function $g_{23}$ which allows translating a multiset of facts into a set of relations. 

\begin{definition}[Function $g_{23}$]
  \label{def:md2mra-database}
  Let $M$ be a multiset of {\NRMDatalogSN} facts $M$ (i.e. an {\NRMDatalogSN} database), and $\{ p_1, \dots, p_n \}$ are the predicate names in $M$.  The function $g_{23}(M)$ returns a set of relations (i.e. a {\MRA} database) $\{ r_1, \dots, r_n \}$ where $r_i = \psi(M_i)$ such that $M_i$ is the subset of {\NRMDatalogSN} facts of $M$ having the predicate name $p_i$.
\end{definition}

  \begin{example}
    Let $M$ be the multiset of {\NRMDatalogSN} facts defined as follows:
    \[
      M =
      \lBrace
      p_1(c_1, c_2), p_1(c_1, c_2), p_1(c_1, c_3),
      p_2(c_1, c_4)
      \rBrace,
    \]
Then, the corresponding {\MRA} dataset $g_{23}(M)$ consists of the following relations $r_1$ and $r_2$, over relation schemas $R_1$ and $R_2$, and with attributes 
$\widehat{{R}_1} = \{ A_1, A_2 \}$ and $\widehat{{R}_2} = \{ B_1, B_2 \}$:
\begin{align*}
r_1 &= \lBrace
          \{A_1 \mapsto c_1, A_2 \mapsto c_2\},
          \{A_1 \mapsto c_1, A_2 \mapsto c_2\},
          \{A_1 \mapsto c_1, A_2 \mapsto c_3\} \rBrace,\\
      r_2 &= \lBrace \{ B_1 \mapsto c_1, B_2 \mapsto c_4\} \rBrace.
\end{align*}
\end{example}

\subsubsection{Translating queries from {\NRMDatalogSN} to {\MRA}}
Recall that a {\NRMDatalogSN} query is a set of rules, and a {\MRA} query is a relational algebra expression.

Let $\Pi$ be a normalized {\NRMDatalogSN} program. We define, by mutual recursion, functions $\delta_1(L,\Pi)$ and $\delta_2(r,\Pi)$ to translate (respectively) literals and rules into {\MRA} expressions.

Given a literal $L$ in $\Pi$ of the form $p(X_1,\dots,X_n)$, the function $\delta_1(L,\Pi)$ is defined as follows:
\begin{enumerate}
\item
  If predicate name $p$ does not occur in the head of any rule of $\Pi$, then $\delta_1(L,\Pi)$ returns the {\MRA} expression $\rho_{A_1/X_1}(\cdots\rho_{A_n/X_n}(R)\cdots)$ where $R$ is the relation name associated to $p$; 
\item
  Otherwise, if $p$ occurs in the head of the rules $\{ r_1,\dots,r_m \}$ in $\Pi$, then $\delta_1(L,\Pi)$ returns the {\MRA} expression $(E_1 \cup ( E_2 \cup ( \cdots E_m)\cdots)$ where each $E_i$ is a {\MRA} expression returned by $\delta_2(r_i,\Pi)$. 
\end{enumerate}

Given a rule $r$ in $\Pi$, the function $\delta_2(r,\Pi)$ is defined as follows:
\begin{itemize}
\item
  If $r$ is a projection rule $L_0 \leftarrow L_1$ then $\delta_2(r,\Pi)$ returns the {\MRA} expression $\pi_S(E)$ where $S$ is the set of variables $\var(L_0)$ and $E$ is the {\MRA} expression returned by $\delta_1(L_1,\Pi)$;
\item
  If $r$ is a join rule $L_0 \leftarrow L_1, L_2$ then $\delta_2(r,\Pi)$ returns the {\MRA} expression $(E_1 \Join E_2)$ where $E_1$ and $E_2$ are the {\MRA} expressions returned by $\delta_1(L_1,\Pi)$ and $\delta_1(L_2,\Pi)$ respectively;
\item
  If $r$ is a negation rule $L_0 \leftarrow L_1, \neg L_2$ then $\delta_2(r,\Pi)$ returns the {\MRA} expression $(E_1 \setminus E_2)$ where $E_1$ and $E_2$ are the {\MRA} expressions returned by $\delta_1(L_1,\Pi)$ and $\delta_1(L_2,\Pi)$ respectively.
\end{itemize}

\begin{definition}[Function $f_{23}$]
  Given a normalized {\NRMDatalogSN} query $Q=(L,\Pi)$ where $L$ is the goal clause, and $\Pi$ a {\NRMDatalogSN} program, the function $f_{23}(Q)$ returns a {\MRA} query defined by $\delta_1(L,\Pi)$.
\end{definition}

\subsubsection{Translating query answers from {\MRA} to {\NRMDatalogSN}}
Recall that a {\MRA} query answer is a multiset relation, and a {\NRMDatalogSN} query answer is a pair $(V,M)$ where $V$ is a set of variables, and $M$ is a multiset of substitutions.  
Since a {\MRA} tuple can be seen (interpreted) also as a Datalog substitution, the translation from {\MRA} tuples to Datalog substitutions requires essentially no modifications.
Next, we define a function $h_{23}$ which transforms a {\MRA} query answer into a {\NRMDatalogSN} query answer.

\begin{definition}[Function $h_{23}$]%
  Given a {\MRA} relation $R$ with schema $\widehat{R} = \{ A_1, \dots, A_n \}$, the function $h_{23}(R)$ returns a {\NRMDatalogSN} query answer $A = (V,M)$ where: $V$ is a set of variables $\{ X_1, \dots, X_n \}$ where variable $X_i$ corresponds to attribute $A_i$ (assume a simple transformation of attribute names to variable names); and, for each tuple $t$ in $R$, there is a substitution $\theta$ in $M$ satisfying that $\theta(X_i) = t(A_i)$ for every attribute $A_i \in \widehat{R}$, and $\card(\theta, M) = \card(t, R)$.
\end{definition}

\begin{lemma}
  \label{lemma:datalog2mra}
  {\NRMDatalogSN} can be simulated by {\MRA}.
\end{lemma}

\begin{proof}
  This is a long but straightforward induction on Datalog queries using as hypothesis that $(f_{23},g_{23},h_{23})$ is a simulation of {\NRMDatalogSN} in {\MRA}. The details of this proof are in the appendix (Claim~\ref{claim:md2mra}).
\end{proof}

\subsection{{\MRA} and {\NRMDatalogSN} have the same expressive power}
Putting together the simulations among {\MRA} and {\NRMDatalogSN} stated in this section, we get the following theorem:

\begin{theorem}
  {\MRA} and {\NRMDatalogSN} have the same expressive power.
\end{theorem}

\begin{proof}
The claim is based on the simulation of {\MRA} with {\NRMDatalogSN} (Lemma \ref{lemma:mra2datalog}) and the simulation of {\NRMDatalogSN} with {\MRA} (Lemma  \ref{lemma:datalog2mra}).
\end{proof}

\section{Equivalence between {\MRA} and {\SPARQL}}%
\label{sec:mra-sparql}
This section presents the simulations that prove that Multiset Relational Algebra ({\MRA}) and {\SPARQL} have the same expressive power. Specifically, we show that {\MRA} can be simulated by {\SPARQL} (Section \ref{sec:mra2sparql}), and {\SPARQL} can be simulated by {\MRA} (Section \ref{sec:sparql2mra}). 

\subsection{From {\MRA} to {\SPARQL}}
\label{sec:mra2sparql}
This section shows that Multiset Relational Algebra ({\MRA}) can be simulated by {\SPARQL}.
To support this, we describe the following translation functions:
\begin{itemize}
\item $f_{31}$ is the function that translates {\MRA} queries into {\SPARQL} queries;
\item $g_{31}$ is the function that translates {\MRA} databases into {\SPARQL} databases; and
\item $h_{31}$ is the function that translates {\SPARQL} query answers into {\MRA} query answers.
\end{itemize}

\subsubsection{Translating databases from {\MRA} to {\SPARQL}}
Recall that a {\MRA} database is a set of relations (where each relation is a multiset of tuples), and a {\SPARQL} database is a set of triples.

Assume the existence of functions that map relation names to IRIs, relation attributes to IRIs, and tuples to IRIs. 

Let $r$ be a multiset relation, $t$ be a tuple in $r$ and $\{ t^1, \dots, t^n \}$ be the set of colored copies of $t$ where $n = \card(t,r)$. The function $\beta(t,r)$ returns a set of RDF triples defined as follows: for each copy $t^i$ of $t$, $\beta(t,r)$ contains a triple $(\iriT{i}, \iriB, iri\_r)$ where $\iriT{i}$ is an IRI which identifies the tuple $t^i$, $iri\_r$ is an IRI which identifies the relation $r$, and $\iriB$ is an IRI which describes that $\iriT{i}$ is a tuple of $iri\_r$; and, for each attribute $A$ in $\widehat{r}$, $\beta(t,r)$ contains a triple of the form $(\iriT{i},\iriA,\litA)$ where $\iriA$ is an IRI which identifies the attribute $A$, and $\litA$ is a literal equivalent to the value $t(A)$.
Hence, for each copy of a tuple $t$ we create a set of RDF triples.


\begin{definition}[Function $g_{31}$]%
  \label{def:mra2sparql-data}
  Given a {\MRA} database $D$, the function $g_{31}(D)$ returns a set of RDF triples $D'$ defined as follows:
  \begin{itemize}
  \item
    For each multiset relation $r$ in $D$, and for each tuple $t$ in $r$, $D'$ contains the RDF triples returned by $\beta(t,r)$;
  \item
    $D'$ contains a triple $(\NULL,\NULL,\NULL)$ where $\NULL$ is a special IRI. Like in the simulation of {\NRMDatalogSN} with {\SPARQL}, the simulation of {\MRA} with {\SPARQL} uses this special triple to retrieve the variables that are attributes of the {\MRA} query answer.
  \end{itemize}
\end{definition}

\begin{example}
  Let $D$ be an {\MRA} dataset consisting in two relations $r$ and $s$ with respective relation schemas $R$ and $S$ with $\vec{R} = (A_1,A_2)$ and $\vec{S} = (A_1,A_3)$, and defined as follows:
  \begin{align*}
    r &= \lBrace
        \{A_1 \mapsto a_1, A_2 \mapsto a_2\},
        \{A_1 \mapsto a_1, A_2 \mapsto a_2\},
        \{A_1 \mapsto a_1, A_2 \mapsto a_3\} \rBrace,\\
    s &= \lBrace \{ A_1 \mapsto a_1, A_3 \mapsto a_4\} \rBrace.
  \end{align*}
  Then, the corresponding {\SPARQL} dataset is the following:
  \[
    g_{31} =
    \begin{aligned}[t]
      \{
      & (u_1^1, \iriB, \iriFor{r}),
        (u_1^1, \iriA_1, \litA_{a_1}),
        (u_1^1, \iriA_2, \litA_{a_2}),
      \\[-4pt]
      & (u_1^2, \iriB, \iriFor{r}),
        (u_1^2, \iriA_1, \litA_{a_1}),
        (u_1^2, \iriA_2, \litA_{a_2}),
      \\[-4pt]
      & (u_2^1, \iriB, \iriFor{r}),
        (u_2^1, \iriA_1, \litA_{a_1}),
        (u_2^1, \iriA_2, \litA_{a_3}),
      \\[-4pt]
      & (v_1^1, \iriB, \iriFor{s}),
        (v_1^1, \iriA_1, \litA_{a_1}),
        (v_1^1, \iriA_3, \litA_{a_4}),
      \\[-4pt]
      & (\NULL, \NULL, \NULL)
      \},
    \end{aligned}
  \]
  where $u_1^1$, $u_1^2$, and  $u_2^1$ correspond to the IRIs for the tuples in the multiset relation $r$ and $v_1^1$ correspond to the IRI of the tuple in the multiset relation $s$.
\end{example}

\subsubsection{Translating queries from {\MRA} to {\SPARQL}}
Recall that a {\MRA} query is a relational algebra expression and a {\SPARQL} query is a graph pattern.

First, consider the following issue. 
A query answer in {\MRA} is a multiset relation $r$ over a set of attributes $\widehat{r}$, whereas a query answer in {\SPARQL} does not specify a set of variables for which solutions are defined. 
For example, the evaluation of the triple pattern $(?X,?Y,?Z)$ over an empty RDF graph results in an empty multiset $\Omega$. The reference to the variables is not carried in the {\SPARQL} answer. Like with the simulation of {\NRMDatalogSN} with {\SPARQL}, we need to define a {\SPARQL} pattern to retrieve the answer variables.

Given an {\MRA} expression $E$, with attributes $\widehat{E}=\{X_1,\dots,X_n\}$, we write $\operatorname{AttrQuery}(E)$ to denote the {\SPARQL} pattern $ (\NULL, \NULL, ?X_1)\; \AAND\, \cdots\, \AAND\; (\NULL, \NULL, ?X_n), $ where, for $1 \leq i \leq n$, variable $?X_i$ is the corresponding {\SPARQL} variable for the {\MRA} attribute $X_i$.

\begin{example}
  Consider the {\MRA} expression $r \Join s$ where $\widehat{r}=\{X,Y\}$ and $\widehat{s}=\{Y,Z\}$. Then, $\operatorname{AttrQuery}(E)= (\NULL, \NULL, ?X)\; \AAND\, (\NULL, \NULL, ?Y)\; \AAND\, (\NULL, \NULL, ?Z)$, where $?X$, $?Y$, and $?Z$ are the corresponding {\SPARQL} variables for attributes $X$, $Y$, and $Z$.
\end{example}

Recall that a {\MRA} query is an {\MRA} expression, and a {\SPARQL} query is a {\SPARQL} graph pattern.  We will show that every type of {\MRA} expression can be translated to a specific type of {\SPARQL} graph pattern.  Table~\ref{table:mra2sparql} shows the translation rules which are the basis for the following definition.

\begin{definition}[Function $f_{31}$]%
  \label{def:mra2sparql-query}
  Given an {\MRA} expression $E$, the function $f_{31}$ returns a {\SPARQL} graph pattern defined by $(\Upsilon(E) \UNION \operatorname{AttrQuery(E)})$.
\end{definition}

\begin{table*}[htb]
  \caption{Definition of function $\Upsilon$ which translates an {\MRA} expression into a {\SPARQL} pattern.}
  \centering
  \begin{adjustbox}{max width=\textwidth}
    \begin{tabular}{p{2.3cm}p{9.3cm}p{6cm}}
      \toprule

      \rs {\MRA} expression $E$
      & {\SPARQL} pattern $\Upsilon(E)$
      & where ...
      \\ \midrule

      \rs $R$
      & ($\SELECT$ ${?X_1 \cdots ?X_n}$ 
      & $iri_r$ is the IRI that identifies $R$, $?Y$ is a variable used to
      \\

      & $((?Y,\mathit{iri\_b},\mathit{iri\_r}) \AAND$\newline $\quad ((?Y,\mathit{\iri\_A}_1,?X_1) \AAND (\cdots \AAND (?Y,\mathit{\iri\_A}_2,?X_n) \cdots)$
      & match every tuple of $R$, and $?X_i$ is a variable that corresponds to the attribute $A_i$ in schema $\widehat{R}$.
      \\ \midrule

      \rs $(E_1 \Join E_2)$
      & $( P_1 \AAND P_2 )$
      & $P_1 = \Upsilon(E_1)$ and $P_2 = \Upsilon(E_2)$.
      \\ \midrule

      \rs $(E_1 \cup E_2)$
      & $(P_1 \UNION P_2)$
      & $P_1 = \Upsilon(E_1)$ and $P_2 = \Upsilon(E_2)$.
      \\ \midrule

      \rs $(E_1 \setminus E_2)$
      & $(P_1 \EXCEPT P_2)$
      & $P_1 = \Upsilon(E_1)$ and $P_2 = \Upsilon(E_2)$.
      \\ \midrule

      \rs $\pi_S(E_1)$
      & $(\SELECT {W} {P_1} )$
      & $P_1 = \Upsilon(E_1)$ and $W$ is the set of variables corresponding to the attributes in $S$.  
      \\ \midrule
      
      \rs $\rho_{A/B}(E_1)$
      & $\subs_{?X/?Y}(P_1)$
      & $P_1 = \Upsilon(E_1)$, $\varX$ is the variable that corresponds to attribute $A$, $?Y$ is the variable that corresponds to attribute $B$, and $\subs_{?X/?Y}(P_1)$ denotes the renaming of variable $\varX$ with variable $\varY$ in the {\SPARQL} query $P_1$ (see Appendix~\ref{sec:sparql-variable-renaming}).
      \\ \midrule    

      \rs $\sigma_{\psi}(E_1)$
      & $(P_1 \FILTER \varphi)$
      & $P_1 = \Upsilon(E_1)$, and $\varphi$ is a filter condition equivalent to the selection condition $\psi$.     
      \\      
      \bottomrule
    \end{tabular}
  \end{adjustbox}
  \label{table:mra2sparql}
\end{table*}

\begin{example}
    Consider the {\MRA} expression $E = R \Join S$ where $\widehat{R} = \{A,B\}$ and $\widehat{S} = \{B,C\}$. Then the corresponding {\SPARQL} query $f_{31}(E)$ is the following:
    \[
    f_{31}(E) =
    \begin{aligned}[t]
        (&
        \begin{aligned}[t]
            (& (\SELECT\;\{\varX_A, \varX_B\}\WHERE\;
            (\varY_1,\mathit{\iri\_A}, \varX_A) \AAND
            (\varY_1,\mathit{\iri\_B}, \varX_B)) \; \AAND \\[-4pt]
            & (\SELECT\;\{\varX_B, \varX_C\}\WHERE\;
            (\varY_1,\mathit{\iri\_B}, \varX_B) \AAND
            (\varY_1,\mathit{\iri\_C}, \varX_C)))\; \UNION
        \end{aligned}\\[-4pt]
        & ((\NULL,\NULL, \varX_A) \AAND\, (\NULL,\NULL, \varX_A) \AAND\, (\NULL,\NULL, \varX_A))).
    \end{aligned}
    \]
    If $R = \lBrace \{A \mapsto a, B \mapsto b \},  \{A \mapsto a, B \mapsto b \} \rBrace$ and $S = \lBrace \{B \mapsto b, C \mapsto c \} \rBrace$, the answer to the {\SPARQL} query over the corresponding translation of the {\MRA} database $D$ to an RDF graph is the following multiset:
    \[
    \ev{f_{31}{E}}{g_{31}(D)} = \lBrace
    \begin{aligned}[t]
        & \{ \varX_A \mapsto a, \varX_B \mapsto b, \varX_C \mapsto c \},\\[-4pt]
        & \{ \varX_A \mapsto a, \varX_B \mapsto b, \varX_C \mapsto c \},\\[-4pt]
        & \{ \varX_A \mapsto \NULL, \varX_B \mapsto \NULL, \varX_C \mapsto \NULL \} \rBrace .
    \end{aligned}
    \]
    Otherwise, if $R$ is empty then:
    \[
    \ev{f_{31}{E}}{g_{31}(D)} = \lBrace
    \begin{aligned}[t]
        & \{ \varX_A \mapsto \NULL, \varX_B \mapsto \NULL, \varX_C \mapsto \NULL \} \rBrace .
    \end{aligned}
    \]
    Whereas the first two {\SPARQL} mappings $\{ \varX_A \mapsto a, \varX_B \mapsto b, \varX_C \mapsto c \}$ correspond are duplicates of the same {\MRA} answer, $\{ A \mapsto a, B \mapsto b, C \mapsto c \}$, the last mapping does not correspond to an answer, but encodes the attributes of the {\MRA} query. By encoding the attributes of the {\MRA} query, we can reconstruct the result {\MRA} relation even in the case it is empty.
\end{example}

\subsubsection{Translating query answers from {\SPARQL} to {\MRA}}
Recall that a query answer in {\SPARQL} is a multiset of mappings, and a query answer in {\MRA} is a multiset relation (i.e. a multiset of tuples). Intuitively, a multiset of mappings $\Omega$ can be transformed into a {\MRA} relation $r$ where the attributes in $\widehat{r}$ are the variables in the domain of $\Omega$. This notion is defined next.

\begin{definition}[Function $h_{31}$]%
  \label{def:h31}
  Let $\Omega$ be a multiset of mappings with $\dom(\mu)=V$ for every mapping $\mu\in\Omega$, and that includes the mapping $\mu_{V \mapsto \NULL}$ with $\dom(\mu_{V \mapsto \NULL})=V$, $\mu_{V \mapsto \NULL}(\varX)=\NULL$ for every variable $?X \in V$, . The function $h_{31}(\Omega)$ returns a multiset relation $r$ where:
  \begin{itemize}
  \item
    For each variable $\varX \in V$, the schema $\widehat{r}$ includes the {\MRA} attribute $A$ corresponding to variable $\varX$.
  \item
    The tuple $t_\mu$ corresponding to a mapping $\mu$ with $\dom(\mu)=V$ is the tuple with attributes $\widehat{r}$ such that $t(A)=\mu(\varX)$, for each {\MRA} attribute $A\in\widehat{r}$ corresponding to a variable $\varX \in V$.
  \item
    $\set(r) = \{ t_\mu \mid \mu \in \set(\Omega) \setminus \{ \mu_{V \mapsto \NULL}\}\}$.
  \item
    $\card(t_\mu, r) = \card(\mu, \Omega)$
  \end{itemize}
\end{definition}

\begin{lemma}%
  \label{lemma:mra2sparql}
  {\MRA} can be simulated in {\SPARQL}.
\end{lemma}

\begin{proof}
  Let $f_{31},g_{31},h_{31}$ denote respectively the functions stated in definitions~and \ref{def:mra2sparql-query}, \ref{def:mra2sparql-data}, and \ref{def:h31}. The proof of this theorem follows from the claim that $(f_{31},g_{31},h_{31})$ simulates {\MRA} in {\SPARQL} by using induction in the structure of queries. The proof of this claim is in the appendix (Claim~\ref{claim:mra2sparql}).
\end{proof}

\subsection{From {\SPARQL} to {\MRA}}
\label{sec:sparql2mra}
This section shows that {\SPARQL} can be simulated by Multiset Relational Algebra ({\MRA}).
To support this, we describe the following translation functions:
\begin{itemize}
\item $f_{13}$ is the function that translates {\SPARQL} queries into {\MRA} queries;
\item $g_{13}$ is the function that translates {\SPARQL} databases into {\MRA} databases; and
\item $h_{13}$ is the function that translates {\MRA} query answers into {\SPARQL} query answers.
\end{itemize}

The translation presented here is inspired by the one presented by Cyganiak~\cite{10140}. However, unlike Cyganiak, we do not use null values with the SQL semantics. Instead, we use a special constant, denoted $\bot$, used to codify unbound values.

\subsubsection{Translating databases from {\SPARQL} to {\MRA}}
Recall that a {\SPARQL} database is a set of RDF triples, and a {\MRA} database is a set of multiset relations. The translation of a set of RDF triples $G$ will produce three multiset relations (without duplicates): $\TripRel$, which codifies the RDF triples in $G$; $\NullRel$, introduced to manage the unbound values of {\SPARQL}; and $\CompRel$, introduced to simulate the notion of compatibility between mappings.

 \begin{definition}[Function $g_{13}$]%
\label{def:sparql2mra-data}
Let $\bot$ be a special constant.
Given a set of RDF triples $G$, the function $g_{13}(G)$ returns a multiset relational database $D'$ containing the multiset relations $\TripRel$, $\NullRel$, and $\CompRel$ defined as follows:
\begin{enumerate}
\item
  $\widehat{\TripRel}=\{S,P,O\}$, $\set(\TripRel) = \{\{S\mapsto s, P\mapsto p, O\mapsto o\} \mid (s,p,o) \in G\}$, and $\card(t, \TripRel)=1$ for every tuple $t \in \set(\TripRel)$.
\item
  $\widehat{\NullRel}=\{N\}$, $\set(\NullRel) =\{\{N \mapsto \bot\}\}$, and $\card(\{N \mapsto \bot\}, \NullRel)=1$.
\item
  $\widehat{\CompRel}=\{A_1,A_2,A\}$, $\set(\CompRel)$ includes the tuple $\{A_1 \mapsto \bot, A_2 \mapsto \bot, A \mapsto \bot\}$ and all tuples of the form $\{A_1 \mapsto a, A_2 \mapsto a, A \mapsto a\}$, $\{A_1 \mapsto \bot, A_2 \mapsto a, A \mapsto a\}$, and $\{A_1 \mapsto a, A_2 \mapsto \bot, A \mapsto a\}$ where $a$ is an RDF term in $G$, and $\card(t, \CompRel)=1$ for every tuple $t \in \set(\CompRel)$.
\end{enumerate}
\end{definition} 

\begin{example}
  Let $G$ be the RDF graph defined as follows
  \[
    G = \{
    \begin{aligned}[t]
      &(\iAlice, \pLivesIn, \iSantiago), (\iAlice, \pKnows, \iBob),\\[-4pt]
      &(\iBob, \pLivesIn, \iSantiago), (\iBob, \pKnows, \iCarol),\\[-4pt]
      &(\iCarol, \pLivesIn, \iLima) \}.
    \end{aligned}
  \]
  Then the data is translated for {\MRA} as the database $g_{13}(G)$ with the multiset relations $\TripRel$, $\NullRel$, and $\CompRel$ defined as follows:
  \begin{align*}
    \TripRel &= 
               \begin{aligned}[t]
                 \lBrace
                 & \{S \mapsto \iAlice, P \mapsto \pLivesIn, O \mapsto \iSantiago\},
                   \{S \mapsto \iAlice, P \mapsto \pKnows,   O \mapsto \iBob\}\\[-4pt]
                 & \{S \mapsto \iBob,   P \mapsto \pLivesIn, O \mapsto \iSantiago\},
                   \{S \mapsto \iBob,   P \mapsto \pKnows,   O \mapsto \iCarol\},\\[-4pt]
                 & \{S \mapsto \iCarol, P \mapsto \pLivesIn, O \mapsto \iLima\}
                   \rBrace \\
               \end{aligned} \\
    \NullRel &= \lBrace \{ N \mapsto \bot \} \rBrace \\
    \CompRel &=
               \begin{aligned}[t]
                 \lBrace
                 & \{ A_1 \mapsto \bot, A_2 \mapsto \bot, A_3 \mapsto \bot\},\\[-4pt]
                 & \{ A_1 \mapsto \iAlice, A_2 \mapsto \iAlice, A_3 \mapsto \iAlice\},\\[-4pt]
                 & \{ A_1 \mapsto \iAlice, A_2 \mapsto \bot, A_3 \mapsto \iAlice\},\\[-4pt]
                 & \{ A_1 \mapsto \bot, A_2 \mapsto \iAlice, A_3 \mapsto \iAlice\},\\[-4pt]
                 & \{ A_1 \mapsto \pLivesIn, A_2 \mapsto \pLivesIn, A_3 \mapsto \pLivesIn\},\\[-4pt]
                 & \{ A_1 \mapsto \pLivesIn, A_2 \mapsto \bot, A_3 \mapsto \pLivesIn\},\\[-4pt]
                 & \{ A_1 \mapsto \bot, A_2 \mapsto \pLivesIn, A_3 \mapsto \pLivesIn\},\\[-4pt]
                 & \quad \vdots \\[-4pt]
                 & \{ A_1 \mapsto \iLima, A_2 \mapsto \iLima, A_3 \mapsto \iLima\},\\[-4pt]
                 & \{ A_1 \mapsto \iLima, A_2 \mapsto \bot, A_3 \mapsto \iLima\},\\[-4pt]
                 & \{ A_1 \mapsto \bot, A_2 \mapsto \iLima, A_3 \mapsto \iLima\}
                   \rBrace.
               \end{aligned}
  \end{align*}
\end{example}

\subsubsection{Translating queries from {\SPARQL} to {\MRA}}
Recall that a {\SPARQL} query is a graph pattern, and a {\MRA} query is a relational algebra expression.
First, we define the function $\Lambda$ which allows translating an RDF triple pattern into a {\MRA} expression.  

Assume that $a$, $b$, $c$ are RDF terms, and $?X$, $?Y$, $?Z$ are variables. 
Recall that $\TripRel$ is a multiset relation that is obtained from a set of RDF triples, where $\widehat{\TripRel} = \{S, P, O\}$ is the schema of $\TripRel$.
Given a triple pattern $T$, the function $\Lambda(T)$ returns a {\MRA} expression defined as follows\footnote{These rules are based on Cyganiak's translation~\cite{10140}.}:
\begin{itemize}
\item 
if $T$ is $(\varX,b,c)$  then $\Lambda(T)$ returns
$\pi_{\varX}(\rho_{S/\varX}(\sigma_{P=b \land O=c}(\TripRel)))$;
\item 
if $T$ is $(a,\varY,c)$  then $\Lambda(T)$ returns
$\pi_{\varY}(\rho_{P/\varY}(\sigma_{S=a \land O=c}(\TripRel)))$;
\item 
if $T$ is $(a,b,\varZ)$  then $\Lambda(T)$ returns
$\pi_{\varZ}(\rho_{O/\varZ}(\sigma_{S=a \land P=b}(\TripRel)))$;
\item 
if $T$ is $(\varX,\varY,c)$  then $\Lambda(T)$ returns
$\pi_{\varX,\varY}(\rho_{P/\varY}(\rho_{S/\varX}(\sigma_{O=c}(\TripRel))))$;
\item 
if $T$ is $(\varX,b,\varZ)$ then $\Lambda(T)$ returns
$\pi_{\varX,\varZ}(\rho_{O/\varZ}(\rho_{S/\varX}(\sigma_{P=b}(\TripRel))))$;
\item 
if $T$ is $(a,\varY,\varZ)$  then $\Lambda(T)$ returns
$\pi_{\varY,\varZ}(\rho_{O/\varZ}(\rho_{P/\varY}(\sigma_{S=a}(\TripRel))))$;
\item 
if $T$ is $(\varX,\varY,\varZ)$ then $\Lambda(T)$ returns
$\pi_{\varX,\varY,\varZ}(\rho_{O/\varZ}(\rho_{P/\varY}(\rho_{S/\varX}(\TripRel))))$;
\item 
if $T$ is $(\varX,\varX,c)$ then $\Lambda(T)$ returns
$\pi_{\varX}(\rho_{S/\varX}(\sigma_{S = P \land O=c}(\TripRel)))$;
\item 
if $T$ is $(\varX,b,\varX)$ then $\Lambda(T)$ returns 
$\pi_{\varX}(\rho_{S/\varX}(\sigma_{P=b\land S=O}(\TripRel)))$.
\item 
if $T$ is $(a,\varX,\varX)$  then $\Lambda(T)$ returns
$\pi_{\varX}(\rho_{P/\varX}(\sigma_{S = a \land P = O }(\TripRel)))$;
\item 
if $T$ is $(\varX,\varX,\varX)$ then $\Lambda(T)$ returns
$\pi_{\varX}(\rho_{S/\varX}(\sigma_{S = P \land P = O }(\TripRel)))$;
\end{itemize}

Second, we define a function $\gamma$ that allows translating a {\SPARQL} filter condition into a {\MRA} selection condition.  Like in the translation from {\SPARQL} to {\NRMDatalogSN}, it is not necessary to translate complex filter conditions ({\SPARQL}) to complex selection formulas ({\MRA}) because {\SPARQL} queries can be normalized to avoid logical connectives.

Given an atomic filter condition $\varphi$, the function $\gamma(\varphi)$ is defined recursively as follows:
\begin{itemize}
\item
  If $\varphi$ is $\varX = c$ then $\gamma(\varphi)$ is $(\neg(X = \bot) \land X = c)$ where $X$ is the attribute name corresponding to variable $\varX$;
\item
  If $\varphi$ is $\varX = \varY$ then $\gamma(\varphi)$ is $((\neg(X = \bot) \land \neg(Y = \bot)) \land X = Y)$ where $X$ and $Y$ are the attribute names corresponding to variables $\varX$ and $\varY$, respectively;
\item
  If $\varphi$ is $\bound(X)$ then $\gamma(\varphi)$ is $\neg(X = \bot)$ where $X$ is the attribute name corresponding to variable $\varX$.
\end{itemize}

In Definition \ref{def:sparql2mra-data}, we introduced the relation named $\CompRel$ to simulate the compatibility between mappings. For example, to simulate the {\SPARQL} query $Q = (P_1 \AAND P_2)$ we need to ensure that check if two pairs of mappings $\mu_1 \in \ev{P_1}{G}$ and $\mu_2 \in \ev{P_2}{G}$ are compatible, and if they are compatible, return the mapping $\mu=\mu_1 \cup \mu_2$ resulting from joining them. To explain how this operation is simulated with {\MRA}, let $\inScope(P_1)\cap\inScope(P_2)=\{\varX\}$ and tuples $t_1$ and $t_2$ correspond to mappings $\mu_1$ and $\mu_2$. To be compatible, either both mappings map variable $\varX$ to the same value, or at least for one of the mappings, variable $\varX$ is unbound. For tuples, an unbound variable $\varX$ is represented with an attribute value $\bot$ (e.g., $t(X)=\bot$). Then, to check if tuples $t_1$ and $t_2$ are \emph{compatible}, we need to rename the attribute name $X$ corresponding to variable $\varX$ as two attributes, namely $X_1$ and $X_2$ and check if there exists a tuple $t_3$ in the result of query $\rho_{A_1/X_1}(\rho_{A_2/X_2}(\CompRel))$ that agrees with tuples $t_1$ and $t_2$ (i.e., $t_3(X_1)=t_1(X)$ and $t_3(X_2)=t_2(X)$) or agrees with either $t_1$ or $t_2$ whereas for the other tuple the value is $\bot$ (e.g., $t_3(X_1)=t_1(X)$ and $t_2(X_2)=\bot$). We recover the renamed attribute $X$ for the attribute $A$ in the relation named $\CompRel$. That is, for the compatibility we use the {\MRA} expression $\rho_{A/X}(\rho_{A_1/X_1}(\rho_{A_2/X_2}(\CompRel)))$ which is generalized as follows for multiple common variables in the scope of patterns $P_1$ and $P_3$.

Let $\mathcal{X}$ be a finite set of attribute names, and $\nu_1$ and $\nu_2$ be two bijective functions that map each attribute $X \in \mathcal{X}$ to two different sets of attributes (i.e., the ranges of $\nu_1$ and $\nu_2$ are disjoint). Then, we write $\CompRel(\nu_1,\nu_2,\mathcal{X})$ to denote the join of {\MRA} expressions of the form
\[
\rho_{A/X}(\rho_{A_1/\nu_1(X)}(\rho_{A_2/\nu_2(X)}(\CompRel))),
\]
for every attribute name $X \in \mathcal{X}$.

Let $E$ be a {\MRA} expression, and $\mathcal{X} = \{X_1,\dots,X_n\}$ be a subset of the attribute names in $\widehat{E}$, and $\nu$ a bijective function that maps each attribute name in $\mathcal{X}$ to a fresh attribute name (i.e., $\nu(X) \notin \widehat{E}$). We call $\nu(E)$ to the {\MRA} expression that renames each attribute name $X \in \mathcal{X}$ with $\nu(X)$. That is, $\nu(E) = \rho_{X_1/\nu(X_1)}(\cdots \rho_{X_n/\nu(X_n)}(E)\cdots)$.

Given two {\MRA} expressions $E_1$ and $E_2$, assume two bijective functions $\nu_1$ and $\nu_2$ that map each attribute $X \in \widehat{E}_1\cap \widehat{E}_2$ to two fresh attributes (i.e., $\nu_1(X),\nu_2(X)\notin\widehat{E}_1\cup\widehat{E}_2$), and satisfy $\range(\nu_1)\cap \range(\nu_2)=\emptyset$. Then, we define the {\MRA} operation $E_1 * E_2$ in terms of existing {\MRA} operators as follows:
\[ E_1 * E_2 = \pi_{\widehat{E}_1 \cup \widehat{E}_2}(
  \CompRel(\nu_1,\nu_2,\widehat{E}_1 \cap \widehat{E}_2)
  \Join \nu_1(E_1) \Join \nu_2(E_2)).\]
Notice that the attribute names in the ranges of functions $\mu_1$ and $\mu_2$ in the definition of expression $E_1 * E_2$ do not matter because are not in the schema of the multiset that results from expression $E_1 * E_2$.

To translate {\SPARQL} queries $Q$ of the form $(\SELECT\;\mathcal{X}\;P)$ where the set of variables $\mathcal{X}$ include a variable that is not in the scope of $P$, we need to generate values $\bot$ to fill the tuples returned by the translated query. For example, if $\inScope(P)=\{\varX\}$ and $\mathcal{X}=\{\varX, \varY\}$, then the {\MRA} expression $E$ that corresponds to the {\SPARQL} pattern $P$ can be extended with an attribute name $Y$ by joining $E$ with the {\MRA} relation $\rho_{N/Y}(\NullRel)$. Given a set $\mathcal{Y} = \{Y_1,\dots,Y_n\}$ of attribute names, we define the {\MRA} expression $\Delta(\mathcal{Y})$ as $\rho_{N/Y_1}(\NullRel) \Join \cdots \Join \rho_{N/Y_n}(\NullRel)$.

Next, we present the translation of {\SPARQL} queries to {\MRA} queries.

\begin{definition}[Function $f_{13}$]
\label{def:sparql2mra-queries}
The translation rules in Table~\ref{table:sparql2mra} define the function $f_{13}$ from normalized graph patterns whose filter conditions have no Boolean connectives to {\MRA} queries.
\end{definition}

\begin{table*}[htb]
  \caption{Definition of the function $f_{13}$, which takes a normalized {\SPARQL} pattern $P$ as input (without logical connectives in filter conditions) and returns an {\MRA} query.}
  \centering
  \begin{adjustbox}{max width=\textwidth}
    \begin{tabular}{p{4cm}p{3cm}p{9cm}}
      \toprule

      \rs {\SPARQL} pattern $P$
      & {\MRA} query $f_{13}(P)$
      & where...
      \\ \midrule

      
      \rs $(s,p,o)$
      & $\Lambda(s,p,o)$
      \\ \midrule
      
      \rs $(P_1 \AAND P_2)$
      & $(f_{13}(P_1) * f_{13}(P_2))$
      \\ \midrule

      \rs $(P_1 \UNION P_2)$
      & $(f_{13}(P_1) \cup f_{13}(P_2))$
      \\ \midrule

      \rs $(P_1 \EXCEPT P_2)$
      & $(f_{13}(P_1) \setminus f_{13}(P_2))$
      \\ \midrule

      \rs $(\SELECT \inScope(P)\; P_1)$ 
      & $\pi_{\mathcal{A}}(f_{13}(P_1) \Join \Delta_{\mathcal{B}})$
      & $\mathcal{A}$ is the set of attribute names corresponding to the variables in set $\inScope(P)$ and $\mathcal{B}$ is the set of attribute names that correspond to variables that are in set $\inScope(P) \setminus \inScope(P_1)$.
      \\ \midrule

      \rs $(P_1 \FILTER~\varphi)$
      & $\sigma_{\gamma(\varphi)}(f_{13}(P_1))$
      \\
      \bottomrule
    \end{tabular}
  \end{adjustbox}
  \label{table:sparql2mra}
\end{table*}

\begin{example}
  Let $Q$ be the following {\SPARQL} query asking for all people, the place where they live, and optionally the people their know (notice that this query is already normalized as we discussed in Example~\ref{ex:sparql-to-md-query}).
  \[
    \begin{aligned}[t]
      ( & ((\Variable{person}, \pLivesIn, \Variable{somewhere}) \AAND\,
          (\Variable{person}, \pKnows, \Variable{somebody})) \\[-4pt]
        & \UNION \\[-4pt]
        &
          \begin{aligned}[t]
            ( & \SELECT \Variable{person}\; \Variable{somewhere}\; \Variable{somebody} \\[-4pt]
              & \WHERE
                \begin{aligned}[t]
                ( & (\Variable{person}, \pLivesIn, \Variable{somewhere}) \EXCEPT \\[-4pt]
                  & \begin{aligned}[t]
                      ( & \SELECT \Variable{person}\; \Variable{somewhere} \\[-4pt]
                        & \WHERE\; ((\Variable{person}, \pLivesIn, \Variable{somewhere}) \AAND\,
                          (\Variable{person}, \pKnows, \Variable{somebody})))))).
                    \end{aligned}
              \end{aligned}
          \end{aligned}
    \end{aligned}
  \]
  Then, the corresponding query $f_{13}(Q)$ is the query $(q(X), \Pi)$ where $\Pi$ is defined as follows:
  \[
    \begin{aligned}[t]
      & (\Lambda(\Variable{person}, \pLivesIn, \Variable{somewhere}) *
        \Lambda(\Variable{person}, \pKnows, \Variable{somebody})) \; \cup \\[-4pt]
      &
        \begin{aligned}[t]
          ( &
              \begin{aligned}[t]
                ( & \pi_{\mathit{Person}}(\Lambda(\Variable{person}, \pLivesIn, \Variable{somewhere})) \; \setminus \\[-4pt]
                  & \pi_{\mathit{Person}}(\Lambda(\Variable{person}, \pLivesIn, \Variable{somewhere}) *
                    \Lambda(\Variable{person}, \pKnows, \Variable{somebody}))) \; \Join 
              \end{aligned}
          \\[-4pt]
            & \rho_{N/\mathit{Somebody}}(\NullRel)) \,,
        \end{aligned}
    \end{aligned}
  \]
  where the {\MRA} attributes $\mathit{Person}$ and $\mathit{Somebody}$ correspond to the {\SPARQL} variables $\Variable{person}$ and $\Variable{somebody}$.
\end{example}

\subsubsection{Translating query answers from {\MRA}  to {\SPARQL}}
Recall that a {\MRA} query answer is a multiset of tuples, and a {\SPARQL} query answer is a multiset of solution mappings.  
Next, we define the function $h_{31}$ that transforms {\MRA} query answers into {\NRMDatalogSN} query answers.

Intuitively, the translation of a {\MRA} tuple $t$ as a {\SPARQL} solution mapping $\mu$ consists of removing from tuple $t$ every attribute whose value is $\bot$, and viewing the result tuple as a {\SPARQL} mapping $\mu$. For example, the result of translating a tuple $t$ with $\hat{t} = \{X,Y\}$, $t(X)=a$, and $t(Y)=\bot$, is the {\SPARQL} mapping $\mu = \{\varX \mapsto a\}$. Recall that we write $\varX$ to denote the corresponding {\SPARQL} variable for a {\MRA} attribute $X$.

\begin{definition}[Function $h_{31}$]
\label{def:sparql2mra-answer}
Given a {\MRA} tuple $t$, we write $f_{31}(t)$ to denote the {\SPARQL} mapping $\mu$ such that: (1) $\mu(\varX) = t(X)$ if $X \in \hat{t}$ and $t(X)\neq \bot$, and (2) variable $\varY$ is not in $\dom(\mu)$ if $Y \notin \hat{t}$ or $t(Y)=\bot$. Abusing notation, $f_{31}(r)$ is also the function that receives a {\MRA} relation $r$ and returns the multiset $\Omega$ of {\SPARQL} mappings where $\set(\Omega) = \{ \mu \mid \text{ there exist } t \in r \text{ such that } f_{31}(t) = \mu \}$ and the cardinality of mapping $f_{31}(t)$ in $\Omega$ is the cardinality of tuple $t$ in $r$.
\end{definition}

\begin{lemma}
  \label{lemma:sparql2mra}
  {\SPARQL} can be simulated by {\MRA}.
\end{lemma}

\begin{proof}
This is a long but straightforward induction on the structure of {\SPARQL} queries using as hypothesis that $(f_{13},g_{13},h_{13})$ is a simulation of {\SPARQL} by {\MRA}. The details of this proof are in the appendix (Claim~\ref{claim:sparql2mra}).
\end{proof}

\subsection{{\MRA} and {\SPARQL} have the same expressive power}
Putting together the simulations among {\MRA} and {\SPARQL} stated in this section, we get the following theorem:

\begin{theorem}
  {\MRA} and {\SPARQL} have the same expressive power.
\end{theorem}

\begin{proof}
The claim is based on the simulation of {\MRA} with {\SPARQL} (Lemma \ref{lemma:mra2sparql}) and the simulation of {\SPARQL} with {\MRA} (Lemma  \ref{lemma:sparql2mra}).
\end{proof}

\section{Conclusions}
\label{sec:conclusions}
We studied the algebraic and logic structure of the multiset semantics of the core {\SPARQL} patterns, and compared it to the classical and well-studied formalisms of multiset relational algebra and multiset Datalog. Our motivation was to shed light on the underlying theoretical structure of the multiset features of {\SPARQL} that could
help improve future designs and implementations.
In this regard, the main discoveries of this research are: (1) the core fragment of {\SPARQL} patterns matches precisely the multiset semantics of Datalog as defined by Mumick et al. \cite{90820}; and (2) this logical structure corresponds to a simple multiset algebra, namely the Multiset Relational Algebra ({\MRA}). 
These correspondences, besides showing a nice parallel to the one exhibited by classical set relational algebra and relational calculus, and thus transferring theoretical guarantees from these well-studied formalisms, could
help to give new insights on possible optimizations and future extensions of {\SPARQL}.

We think there are a couple of lessons learned in the investigation of the multiset features of {\SPARQL}.
First, contrary to the rather chaotic variety of multiset operators in SQL, it is interesting to observe that the {\SPARQL} design comprises a more coherent body of multiset operators. 
We suggest that this asset should be considered and curated by designers in order to try to keep this clean design in future extensions of {\SPARQL}.
Second, there is a challenging goal for query language designers that work with multisets: existing a diversity of multiset extensions for each of the classical set operators, it is not evident at all from a theoretical perspective how to develop a logically coherent formalism that could integrate all or most of them. 

Our study shows that there are fragments that behave coherently, but that operators that do not fit in this schema, when available (not always), have to be accessed in a very ad-hoc manner. 
Last but not least, this study shows (and adds evidence of) the complexities and challenges that the introduction of multisets brings to query languages, exemplified here in the case of {\SPARQL}. Much more use cases are needed in order to match the theoretical restrictions and recommendations (e.g. as studied in this paper), and real-life use cases that to the best of our knowledge do not have yet a good systematization.

\begin{acks}
R. Angles was supported by ANID FONDECYT Chile through grant 1221727.
D. Hernández was partially supported by the German Research Foundation, \emph{Deutsche Forschungsgemeinschaft (DFG)}, grant SFB-1574-471687386.
This work was partly funded by ANID - Millennium Science Initiative Program - Code ICN17\_002.
\end{acks}




\bibliographystyle{ACM-Reference-Format}
\bibliography{bibliography}        

\appendix

\section{Variable renaming in {\SPARQL}}%
\label{sec:sparql-variable-renaming}
This appendix section defines function $\subs(\cdot,\cdot)$, which renames {\SPARQL} variables. This function is used to simulate the {\MRA} operator renaming $\rho_{A/B}$ (see Table~\ref{table:mra2sparql}). Note that function $\subs(\cdot,\cdot)$ is not an additional algebraic operation but an operation over expressions (i.e., a query rewriting). Intuitively, given a {\MRA} query $Q$, a {\SPARQL} pattern $P$ that simulates $Q$, a renaming of {\MRA} attributes $A/B$ and a renaming of variables $\varX/\varY$ where $\varX$ and $\varY$ are the corresponding variables for attributes $A$ and $B$, the query rewriting $\subs_{\varX/\varY}(P)$ simulates the {\MRA} query $\rho_{A,B}(Q)$. To this end, {\SPARQL} variables are renamed in the pattern, instead of renaming query result attributes as {\MRA} does.

\begin{definition}[{\SPARQL} Variable Renaming]
  Given two {\SPARQL} variables $\varX$ and $\varY$, we define the function $\nu_{\varX/\varY}: \sI \cup \sL \cup \sV \to \sI \cup \sL \cup \sV$ as the function such that $\nu_{\varX/\varY}(\varX)=\varY$ and $\nu_{\varX/\varY}(s)=s$, for every $s \in (\sI \cup \sL \cup \sV) \setminus \{\varX\}$. Given a {\SPARQL} pattern $P$ and two {\SPARQL} variables $\varX \in \inScope(P)$ and $\varY \notin \inScope(P)$, we write $\subs_{\varX/\varY}(P)$ to denote the pattern defined recursively as follows:
  \begin{enumerate}
  \item If $P$ is a triple pattern $(s,p,o)$ then $\subs_{\varX/\varY}(P)=(\nu_{\varX/\varY}(s), \nu_{\varX/\varY}(p), \nu_{\varX/\varY}(o))$.
  \item {\sloppy If $P$ has the form $(P_1 \AAND P_2)$ then $\subs_{\varX/\varY}(P)=\subs_{\varX/\varY}(P_1) \AAND \subs_{\varX/\varY}(P_2)$.\par}
  \item If $P$ has the form $(P_1 \UNION P_2)$ then $\subs_{\varX/\varY}(P)=\subs_{\varX/\varY}(P_1) \UNION \subs_{\varX/\varY}(P_2)$.
  \item {\sloppy If $P$ has the form $(P_1 \EXCEPT P_2)$ then $\subs_{\varX/\varY}(P)=\subs_{\varX/\varY}(P_1) \EXCEPT \subs_{\varX/\varY}(P_2)$.\par}
  \item If $P$ has the form $(P_1 \FILTER \varphi)$ then $\subs_{\varX/\varY}(P)=(\subs_{\varX/\varY}(P_1) \FILTER \nu_{\varX/\varY}(\varphi))$ where, abusing of notation, $\nu_{\varX/\varY}(\varphi)$ is the selection formula defined recursively as follows:
    \begin{enumerate}
    \item If $\varphi$ has the form $a = b$, where $a,b \in \sV \cup \sI \cup \sI$, then $\nu_{\varX/\varY}(\varphi)=\nu_{\varX/\varY}(a)=\nu_{\varX/\varY}(b)$.
    \item If $\varphi$ has the form $\bound(?x)$ then $\nu_{\varX/\varY}(\varphi)=\bound(\nu_{\varX/\varY}(?x))$.
    \item If $\varphi$ has the form $\psi_1 \land \psi_2$ then $\nu_{\varX/\varY}(\varphi)=\nu_{\varX/\varY}(\psi_1) \land \nu_{\varX/\varY}(\psi_2)$.
    \item If $\varphi$ has the form $\psi_1 \lor \psi_2$ then $\nu_{\varX/\varY}(\varphi)=\nu_{\varX/\varY}(\psi_1) \lor \nu_{\varX/\varY}(\psi_2)$.
    \item If $\varphi$ has the form $\neg\psi$ then $\nu_{\varX/\varY}(\varphi)=\neg\nu_{\varX/\varY}(\psi)$. 
    \end{enumerate}
  \item If $P$ has the form $(\SELECT W \WHERE P_1)$ then:
    \begin{enumerate}
    \item If $\varY \notin \inScope(P_1)$, then
      $\subs_{\varX/\varY}(P)=(\SELECT~(W\setminus\{\varX\}\cup\{\varY\})
      \WHERE~P_1)$.
    \item Otherwise,
      $\subs_{\varX/\varY}(P)=(\SELECT~(W\setminus\{\varX\}\cup\{\varY\})
      \WHERE~\subs_{\varY/Z}(P_1))$, where 
      $\varZ$ is a fresh variable. We rename variable $\varY$ as $\varZ$ when is not in-scope of $P$ to avoid a variable name clash.
    \end{enumerate}
  \end{enumerate}
\end{definition}

\section{Proof of claims}

\subsection{Error filter condition}

\newcommand{\hlcell}{\cellcolor[gray]{0.8}}

\begin{table*}
  \caption{Truth values for the error formula of a conjunction.  According to Definition~\ref{def:error-formula}, given a formula $\varphi$ of the form $\varphi_1 \land \varphi_2$, the formula $\Error(\varphi)$ is the formula $\psi_1 \lor \psi_2 \lor \psi_3$ where $\psi_1$ is the formula $(\varphi_1 \land \Error(\varphi_2))$, $\psi_2$ is the formula $(\Error(\varphi_1) \land \varphi_2)$, and $\psi_3$ is the formula $(\Error(\varphi_1) \land \Error(\varphi_2))$.  Given an arbitrary mapping $\mu$, this table shows the possible truth values for formulas $\varphi$, $\Error(\varphi)$, and its components.}
  \label{table:truth-values-error-conjunction}
  \centering
  \begin{adjustbox}{max width=\textwidth}
    \begin{tabular}{ccc@{\hskip 3em}cc@{\hskip 3em}ccc@{\hskip 3em}c}
      \toprule
      $\mu(\varphi_1)$
      & $\mu(\varphi_2)$
      & $\mu(\varphi)$
      & $\mu(\Error(\varphi_1))$
      & $\mu(\Error(\varphi_2))$
      & $\mu(\psi_1)$
      & $\mu(\psi_2)$
      & $\mu(\psi_3)$
      & $\mu(\Error(\varphi))$ \\
      \midrule
      $\true$  & $\true$  & $\true$         & $\false$ or $\error$ & $\false$ or $\error$ & $\false$ or $\error$ & $\false$ or $\error$ & $\false$ or $\error$ & $\false$ or $\error$ \\
      $\true$  & $\false$ & $\false$        & $\false$ or $\error$ & $\false$ or $\error$ & $\false$ or $\error$ & $\false$             & $\false$ or $\error$ & $\false$ or $\error$ \\
      $\true$  & $\error$ & \hlcell$\error$ & $\false$ or $\error$ & $\true$              & \hlcell$\true$       & $\false$ or $\error$ & $\false$ or $\error$ & \hlcell$\true$ \\
      $\false$ & $\true$  & $\false$        & $\false$ or $\error$ & $\false$ or $\error$ & $\false$             & $\false$ or $\error$ & $\false$ or $\error$ & $\false$ or $\error$ \\
      $\false$ & $\false$ & $\false$        & $\false$ or $\error$ & $\false$ or $\error$ & $\false$             & $\false$             & $\false$ or $\error$ & $\false$ or $\error$ \\
      $\false$ & $\error$ & $\false$        & $\false$ or $\error$ & $\true$              & $\false$             & $\false$ or $\error$ & $\false$ or $\error$ & $\false$ or $\error$ \\
      $\error$ & $\true$  & \hlcell$\error$ & $\true$              & $\false$ or $\error$ & $\false$ or $\error$ & \hlcell$\true$       & $\false$ or $\error$ & \hlcell$\true$ \\
      $\error$ & $\false$ & $\false$        & $\true$              & $\false$ or $\error$ & $\false$ or $\error$ & $\false$             & $\false$ or $\error$ & $\false$ or $\error$ \\
      $\error$ & $\error$ & \hlcell$\error$ & $\true$              & $\true$              & $\error$             & $\error$             & \hlcell$\true$       & \hlcell$\true$ \\
      \bottomrule
    \end{tabular}
  \end{adjustbox}
\end{table*}

\begin{table*}
  \caption{Truth values for the error formula of a disjunction.  According to Definition~\ref{def:error-formula}, given a formula $\varphi$ of the form $\varphi_1 \lor \varphi_2$, the formula $\Error(\varphi)$ is the formula $\psi_1 \lor \psi_2 \lor \psi_3$ where $\psi_1$ is the formula $(\neg\varphi_1 \land \Error(\varphi_2))$, $\psi_2$ is the formula $(\Error(\varphi_1) \land \neg\varphi_2)$, and $\psi_3$ is the formula $(\Error(\varphi_1) \land \Error(\varphi_2))$.  Given an arbitrary mapping $\mu$, this table shows the possible truth values for formulas $\varphi$, $\Error(\varphi)$, and its components.}
  \label{table:truth-values-error-disjunction}
  \centering
  \begin{adjustbox}{max width=\textwidth}
    \begin{tabular}{ccc@{\hskip 3em}cc@{\hskip 3em}ccc@{\hskip 3em}c}
      \toprule
      $\mu(\varphi_1)$
      & $\mu(\varphi_2)$
      & $\mu(\varphi)$
      & $\mu(\Error(\varphi_1))$
      & $\mu(\Error(\varphi_2))$
      & $\mu(\psi_1)$
      & $\mu(\psi_2)$
      & $\mu(\psi_3)$
      & $\mu(\Error(\varphi))$ \\
      \midrule
      $\true$  & $\true$  & $\true$         & $\false$ or $\error$ & $\false$ or $\error$ & $\false$             & $\false$             & $\false$ or $\error$ & $\false$ or $\error$ \\
      $\true$  & $\false$ & $\true$         & $\false$ or $\error$ & $\false$ or $\error$ & $\false$             & $\false$ or $\error$ & $\false$ or $\error$ & $\false$ or $\error$ \\
      $\true$  & $\error$ & $\true$         & $\false$ or $\error$ & $\true$              & $\false$             & $\false$ or $\error$ & $\false$ or $\error$ & $\false$ or $\error$ \\
      $\false$ & $\true$  & $\true$         & $\false$ or $\error$ & $\false$ or $\error$ & $\false$ or $\error$ & $\false$             & $\false$ or $\error$ & $\false$ or $\error$ \\
      $\false$ & $\false$ & $\false$        & $\false$ or $\error$ & $\false$ or $\error$ & $\false$ or $\error$ & $\false$ or $\error$ & $\false$ or $\error$ & $\false$ or $\error$ \\
      $\false$ & $\error$ & \hlcell$\error$ & $\false$ or $\error$ & $\true$              & \hlcell$\true$       & $\false$ or $\error$ & $\false$ or $\error$ & \hlcell$\true$       \\
      $\error$ & $\true$  & $\true$         & $\true$              & $\false$ or $\error$ & $\false$ or $\error$ & $\false$             & $\false$ or $\error$ & $\false$ or $\error$ \\
      $\error$ & $\false$ & \hlcell$\error$ & $\true$              & $\false$ or $\error$ & $\false$ or $\error$ & \hlcell$\true$       & $\false$ or $\error$ & \hlcell$\true$       \\
      $\error$ & $\error$ & \hlcell$\error$ & $\true$              & $\true$              & $\error$             & $\error$             & \hlcell$\true$       & \hlcell$\true$       \\
      \bottomrule
    \end{tabular}
  \end{adjustbox}
\end{table*}

\begin{claim}
  \label{claim:error-formula}
  For every {\SPARQL} formula $\varphi$, the formula $\Error(\varphi)$ can be expressed as a formula of the form $\bigvee_{\psi \in C} \psi$ where $C$ is a non-empty set of conjunctions of formulas belonging to one of the following types:
  \begin{enumerate}
  \item
    positive or negative literals (i.e., formulas of the form $\false$, $\varX = a$, $\neg(\varX = a)$, $\neg(\varX = {\varY})$, $\bound(\varX)$, or $\neg\bound(\varX)$),
  \item
    formulas $\varphi'$, $\neg\varphi'$, or $\Error(\varphi')$ such that $\varphi'$ occurs in $\varphi$ and $\varphi'$ is strictly smaller than $\varphi$;
  \end{enumerate}
  and for every mapping $\mu$, $\mu(\varphi)=\error$ if and only if there exists a unique formula $\psi \in C$ for which $\mu(\psi)=\true$.
\end{claim}

\begin{proof}
  We next show this result by induction on the structure of the query.
  \begin{enumerate}
  \item
    If $\varphi$ has the form $\bound(\varX)$ then $\Error(\varphi)$ is the formula $\false$.  Formula $\varphi$ satisfies the claim.  Indeed, $C = \{\false\}$ and $\mu(\Error(\varphi))=\false$ for every mapping $\mu$ because formula $\varphi$ does not produce error.
    
  \item
    If $\varphi$ has the form $\varX = a$ then $\Error(\varphi)$ is the formula $\neg\bound(\varX)$.  Formula $\varphi$ satisfies the claim.  Indeed, $C = \{\neg\bound(\varX)\}$ and $\mu(\Error(\varphi))=\true$ if and only if variable $\varX$ is unbound in $\mu$, that is the unique case when formula $\varphi$ produces error.
    
  \item
    If $\varphi$ has the form $\varX = {\varY}$ then $\Error(\varphi)$ is the formula $\psi_1 \lor \psi_2 \lor \psi_3$ where $\psi_1$ is the formula $(\neg\bound(\varX) \land \bound(\varY))$, $\psi_2$ is the formula $(\bound(\varX) \land \neg\bound(\varY))$, and $\psi_3$ is the formula $(\neg\bound(\varX) \land \neg\bound(\varY))$.  Formula $\varphi$ satisfies the claim.  Indeed, $C = \{\psi_1, \psi_2, \psi_3\}$, and by construction, only one formula in $C$ can be true, and $\mu(\Error(\varphi))=\true$ if and only if $\mu(\varphi)=\error$.
    
  \item
    If $\varphi$ has the form $\neg\varphi_1$ then $\Error(\varphi)$ is the formula $\Error(\varphi_1)$.  In this case $C = \{\Error(\varphi_1)\}$.  By the induction hypothesis, $\mu(\Error(\varphi_1))=\true$ if and only if $\mu(\varphi_1)=\error$.  Because $\neg\error$ is $\error$, we conclude that $\mu(\Error(\varphi))=\true$ if and only if $\mu(\varphi)=\error$.  Hence, formula $\varphi$ satisfies the claim.
    
  \item
    If $\varphi$ has the form $\varphi_1 \land \varphi_2$ then $\Error(\varphi)$ is the formula $\psi_1 \lor \psi_2 \lor \psi_3$ where $\psi_1$ is the formula $(\varphi_1 \land \Error(\varphi_2))$, $\psi_2$ is the formula $(\Error(\varphi_1) \land \varphi_2)$, and $\psi_3$ is the formula $(\Error(\varphi_1) \land \Error(\varphi_2))$.  The validity of the claim for formula $\varphi$ is shown in Table~\ref{table:truth-values-error-conjunction}.  There are three cases where $\mu(\varphi)=\error$:
    \begin{enumerate}
    \item[Case ET:] $\mu(\varphi_1)=\error$ and $\mu(\varphi_2)=\true$, 
    \item[Case TE:] $\mu(\varphi_1)=\true$ and $\mu(\varphi_2)=\error$,
    \item[Case EE:] $\mu(\varphi_1)=\error$ and $\mu(\varphi_2)=\error$.
    \end{enumerate}
    
    Note that if $\mu(\varphi_1)=\error$ and $\mu(\varphi_2)=\false$ then $\mu(\varphi)=\false$ (because $\error \land \false$ is $\false$).

    The values in the remaining columns can be computed using the inductive hypothesis.  We next present case ET as an example.  The other cases follow the same reasoning.  In case ET, $\mu(\varphi_1)=\error$ and $\mu(\varphi_2)=\true$.  By the induction hypothesis, $\mu(\Error(\varphi_1))=\true$ and $\mu(\Error(\varphi_2))$ is either $\false$ or $\error$.
    \begin{enumerate}
    \item If $\mu(\Error(\varphi_2))=\false$ then:
      \begin{align*}
        \mu(\psi_1) &= \mu(\varphi_1) \land \mu(Error(\varphi_2)) \\[-5pt]
                    &= \error \land \false \\[-5pt]
                    &= \false. \\
        \mu(\psi_2) &= \mu(Error(\varphi_1)) \land \mu(\varphi_2) \\[-5pt]
                    &= \true \land \true \\[-5pt]
                    &= \true. \\
        \mu(\psi_3) &= \mu(Error(\varphi_1)) \land \mu(\Error(\varphi_2)) \\[-5pt]
                    &= \true \land \false \\[-5pt]
                    &= \false.
      \end{align*}
      Hence,
      $\mu(\Error(\varphi)) = \mu(\psi_1\lor\psi_2\lor\psi_3) = \false
      \lor \true \lor \false =\true$.

    \item If $\mu(\Error(\varphi_2))=\error$ then:
      \begin{align*}
        \mu(\psi_1) &= \mu(\varphi_1) \land \mu(Error(\varphi_2)) \\[-5pt]
                    &= \error \land \error \\[-5pt]
                    &= \error. \\
        \mu(\psi_2) &= \mu(Error(\varphi_1)) \land \mu(\varphi_2) \\[-5pt]
                    &= \true \land \true \\[-5pt]
                    &= \true. \\
        \mu(\psi_3) &= \mu(Error(\varphi_1)) \land \mu(\Error(\varphi_2)) \\[-5pt]
                    &= \true \land \error \\[-5pt]
                    &= \error.
      \end{align*}
      Hence,
      $\mu(\Error(\varphi)) = \mu(\psi_1\lor\psi_2\lor\psi_3) = \error
      \lor \true \lor \error =\true$.
    \end{enumerate}
    
  \item
    If $\varphi$ has the form $\varphi_1 \lor \varphi_2$ then the validity of the claim for formula $\varphi$ is shown in Table~\ref{table:truth-values-error-disjunction}, following the same reasoning as for the previous case where $\varphi$ is a conjunction $\varphi_1 \land \varphi_2$.
  \end{enumerate}  
\end{proof}

\newcommand{\fcomp}{\operatorname{comp}}

\subsection{Reduction of complex filter conditions}

\noindent
To prove the following claims, we introduce the notion of \emph{reduction} and \emph{reducible filter condition}.  Section~\ref{sec:sparql-to-md} presents three equivalences to transform a pattern with complex filter conditions into a pattern where all filter conditions are atomic. In this appendix, we show that these equivalences can be used to this end. For each equivalence $(P \FILTER \varphi) \equiv P'$, we define a function that maps the filter condition $\varphi$ to the set $\Sigma_\varphi$ of filter conditions in pattern~$P'$.

Consider the following equivalences:
\begin{align*}
  &(P \FILTER \psi_1 \land \psi_2) \equiv
    ((P \FILTER \psi_1) \FILTER \psi_2),\\
  &\begin{aligned}[t]
     (P \FILTER \psi_1 \lor \psi_2) \equiv\;
     &(P \FILTER \psi_1 \land \psi_2) \UNION \\[-5pt]
     &(P \FILTER \psi_1 \land \neg \psi_2) \UNION \\[-5pt]
     &(P \FILTER \neg \psi_1 \land \psi_2) \UNION \\[-5pt]
     &(P \FILTER \psi_1 \land \Error(\psi_2)) \UNION \\[-5pt]
     &(P \FILTER \Error(\psi_1) \land \psi_2) , \\
   \end{aligned}\\
  &(P \FILTER \neg\psi)
    \equiv ((P \EXCEPT\, (P \FILTER \psi)) \EXCEPT\, (P \FILTER
    \Error(\psi))).
\end{align*}
These three equivalences define the following functions, called \emph{reduction rules}:
\begin{align*}
  f_{\land}(\varphi)
  &= \left\{
    \begin{array}{ll}
      \{\psi_1,\,\psi_2\}
      &\quad\text{if }\varphi\text{ has the form }\psi_1 \land \psi_2,\\
      \{\varphi\}
      &\quad\text{otherwise};
    \end{array}
    \right.\\
  f_{\lor}(\varphi)
  &= \left\{
    \begin{array}{ll}
      \begin{aligned}[t]
        \{ &\psi_1\land\psi_2,\,
        \psi_1\land\neg\psi_2,\,
        \psi_1\land\Error(\psi_2),\\
        &\neg\psi_1\land\psi_2,\,
        \Error(\psi_1)\land\psi_2 \}
      \end{aligned}
      &\quad\text{if }\varphi\text{ has the form }\psi_1 \lor \psi_2,\\
      \{\varphi\}
      &\quad\text{otherwise};
    \end{array}
    \right.\\
  f_{\neg}(\varphi)
  &= \left\{
    \begin{array}{ll}
      \{\psi,\,\Error(\psi)\}
      &\quad\text{if }\varphi\text{ has the form }\neg \psi,\\
      \{\varphi\}
      &\quad\text{otherwise};
    \end{array}
    \right.
\end{align*}
Note that if the filter condition $\varphi$ does not have the form of filter condition on the left side of the identity, we return the set $\{\varphi\}$. This captures the fact that the equivalence cannot be applied to reduce filter condition $\varphi$.

For convenience, we also define the reduction function that eliminates atomic formulas $f_{\circ}$ and a reduction that composes $f_{\lor}$ with $f_{\land}$, called $f_{\lor\land}$.
\begin{align*}
  f_{\circ}(\varphi)
  &= \left\{
    \begin{array}{cl}
      \{\varphi\} & \quad\text{if }\varphi\text{ is a complex filter condition},\\
      \emptyset & \quad\text{if }\varphi\text{ is an atomic filter condition};
    \end{array}
  \right.\\
  f_{\lor\land}(\varphi)
  &= \left\{
    \begin{array}{ll}
      \{
      \psi_1,\psi_2,\,
      \neg\psi_1,\,\neg\psi_2,\,
      \Error(\psi_1),\,\Error(\psi_2)
      \}
      &\quad\text{if }\varphi\text{ has the form }\psi_1 \lor \psi_2,\\
      \{\varphi\}
      &\quad\text{otherwise}.
    \end{array}
    \right.
\end{align*}
For $r \in \{\land,\lor,\neg,\circ,\lor\land\}$, let $F_r$ be the function that receives a set of filter conditions $\Sigma$ and returns the set of filter conditions $F_r(\Sigma)=\bigcup_{\varphi \in \Sigma}f_r(\varphi)$. Given two sets of filter conditions $\Sigma_1$ and $\Sigma_2$ we write $\Sigma_1 \xrightarrow{r} \Sigma_2$ if $F_r(\Sigma_1)=\Sigma_2$. In this case, we say that $\Sigma_1 \xrightarrow{r} \Sigma_2$ is a \emph{reduction}. We said that a filter condition $\varphi$ is reducible if there is a finite sequence of reductions $\{\varphi\}\xrightarrow{r_1} \Sigma_1\xrightarrow{r_2} \cdots \xrightarrow{r_n} \emptyset$. Intuitively, reductions are applied until all complex filter conditions are eliminated.  It is not difficult to see that we can apply the aforementioned equivalences to transform every pattern $P_1$ to a pattern $P_2$ with no complex formulas if and only if every filter condition $\varphi$ is reducible.

We next prove that every filter condition is reducible by induction on the structure of the filter condition. For this induction, we define the components of a filter condition $\varphi$, denoted $\fcomp(\varphi)$, to be the set of filter conditions defined as follows: If $\varphi$ is atomic, then $\fcomp(\varphi)=\emptyset$; if $\varphi=\psi_1 \lor \psi_2$ or $\varphi=\psi_1 \land \psi_2$, then $\fcomp(\varphi)=\{\psi_1,\psi_2\}\cup\fcomp(\psi_1)\cup\fcomp(\psi_2)$; and if $\varphi=\neg\psi$ then $\fcomp(\varphi)=\{\psi\}\cup\fcomp(\psi)$.

\begin{claim}
  \label{claim:reduction-complex-formula}
  Every filter condition $\varphi$ is reducible.
\end{claim}

\begin{proof}
  We prove this by induction using the following hypothesis: if $\varphi$ is a filter condition where for each filter condition $\psi \in \fcomp(\varphi)$, $\psi$ and $\Error(\psi)$ are reducible, then the filter conditions $\varphi$ and $\Error(\varphi)$ are reducible.  
  \begin{enumerate}
  \item If $\varphi$ is $\bound(\varX)$ then
    \begin{align*}
      &\{\varphi\} \xrightarrow{\circ} \emptyset,\\
      &\{\Error(\varphi)\}
      = \{\false\}
      \xrightarrow{\circ}\emptyset.
    \end{align*}
  \item If $\varphi$ is $\varX=a$ then
    \begin{align*}
      &\{\varphi\} \xrightarrow{\circ} \emptyset,\\
      &\{\Error(\varphi)\}
      = \{\neg\bound(\varX)\}
      \xrightarrow{\neg} \{\bound(\varX),\;\Error(\neg\bound(\varX))\}
      = \{\bound(\varX),\;\false\}
      \xrightarrow{\circ} \emptyset.
    \end{align*}
  \item If $\varphi$ is $\varX=\varY$ then
    \begin{align*}
      &\{\varphi\} \xrightarrow{\circ} \emptyset,\\
      &\{\Error(\varphi)\}
        \begin{aligned}[t]
          &= \{\neg\bound(\varX) \lor \neg\bound(\varY)\}\\
          &\xrightarrow{\lor\land}
            \begin{aligned}[t]
              \{&\bound(\varX),\;\bound(\varY),\;\neg\bound(\varX),\;\neg\bound(\varY),\\
                &\Error(\neg\bound(\varX)),\;\Error(\neg\bound(\varY))\}\\
            \end{aligned}\\
          &=
            \{\bound(\varX),\;\bound(\varY),\;\neg\bound(\varX),\;\neg\bound(\varY),\;false\}\\
          &\xrightarrow{\circ}
            \{\neg\bound(\varX),\;\neg\bound(\varY)\}\\
          &\xrightarrow{\neg}
            \{\bound(\varX),\;\Error(\bound(\varX)),\;\bound(\varY),\;\Error(\bound(\varY))\}\\
          &=
            \{\bound(\varX),\;\false,\;\bound(\varY),\;\false\}\\
          &\xrightarrow{\circ} \emptyset.
        \end{aligned}
    \end{align*}
  \item If $\varphi$ is $\neg\psi_1$ then
    \begin{align*}
      &\{\varphi\} \xrightarrow{\neg} \{\psi,\;\Error(\psi)\},\\
      &\{\Error(\varphi)\} = \{\Error(\psi)\}.
    \end{align*}
    Since $\psi \in \fcomp(\varphi)$ and by inductive hypothesis, the filter conditions $\psi$ and $\Error(\psi)$ are reducible. Hence, the filter conditions $\varphi$ and $\Error(\varphi)$ are reducible.
  \item If $\varphi$ is $\psi_1\land\psi_2$ then
    \begin{align*}
      &\{\varphi\} \xrightarrow{\land} \{\psi_1,\;\psi_2\},\\
      &\{\Error(\varphi)\}
        \begin{aligned}[t]
          &= \{\Error(\psi) \lor \Error(\psi_2)\}\\
          &\xrightarrow{\lor\land} \{
            \begin{aligned}[t]
              &\Error(\psi_1),\; \Error(\psi_2),\\
              &\neg\Error(\psi_1),\;\neg\Error(\psi_2),\\
              &\Error(\Error(\psi_1)),\; \Error(\Error(\psi_2))\}
            \end{aligned}\\
          &= \{
            \Error(\psi_1) ,\; \Error(\psi_2),\; \neg\Error(\psi_1),\;\neg\Error(\psi_2),\;\false\}.
        \end{aligned}
    \end{align*}
    Since $\psi_1,\psi_2 \in \fcomp(\varphi)$ and by the induction hypothesis, the filter conditions $\Error(psi_1)$ and $\Error(\psi_2)$ are reducible. To show that $\varphi$ is reducible, we have to show that $\neg\Error(\psi_1)$ and $\neg\Error(\psi_2)$ are reducible.
    \begin{align*}
      &\{\neg\Error(\psi_1)\}
          \xrightarrow{\neg} \{\Error(\psi),\; \Error(\Error(\psi_2))\}
          = \{\Error(\psi),\; false\}.
    \end{align*}
    By the induction hypothesis, $\Error(\psi)$ is reducible. Hence, $\neg\Error(\psi_1)$ is reducible. Similarly, $\neg\Error(\psi_1)$ is reducible. Then, $\Error(\varphi)$ is reducible.

  \item
    Let $\varphi$ be $\psi_1 \lor \psi_2$. First, we show that $\varphi$ is reducible.
    \begin{align*}
      &\{\varphi\} \xrightarrow{\lor\land}
        \{\psi_1,\psi_2,\neg\psi_1,\neg\psi_2,\Error(\psi_1),\Error(\psi_2)\}
    \end{align*}
    By the induction hypothesis on $\psi_1$ and $\psi_2$, $\psi_1$, $\psi_2$, $\Error(\psi_1)$, and $\Error(\psi_2)$ are reducible. To prove that $\varphi$ is reducible, suffices to prove that $\neg\psi_1$ and $\neg\psi_2$ are reducible.
    \begin{align*}
      \{\neg\psi_1\} \xrightarrow{\neg} \{\psi_1, \Error(\psi_1)\}.
    \end{align*}
    By the induction hypothesis in $\psi_1$, $\psi_1$ and $\Error(\psi_1)$ are reducible. Hence, $\neg\psi_1$ is reducible. Similarly, $\neg\psi_2$ is reducible. Hence, $\varphi$ is reducible.

    Second, we show that $\Error(\varphi)$ is reducible.
    \begin{align*}
      \{\Error(\varphi)\} = \{\Error(\psi_1) \land \Error(\psi_2)\}
      \xrightarrow{\land} \{\Error(\psi_1),\; \Error(\psi_2)\}.
    \end{align*}
    By the induction hypothesis in $\psi_1$ and $\psi_2$, $\Error(\psi_1)$ and $\Error(\psi_2)$ are reducible. Hence, $\Error(\varphi)$ is reducible.
  \end{enumerate}
  Hence, for every filter condition $\varphi$, the filter conditions $\varphi$ and $\Error(\varphi)$ are reducible.
\end{proof}

\subsection{Normalization of {\NRMDatalogSN} queries}

\tikzset{Rule/.style={fill=white, font=\scriptsize}}

\begin{claim}[Normalized {\NRMDatalogSN}]
  \label{claim:normalized-md}
  Let $(p(\bar{X}),\Pi)$ be a {\NRMDatalogSN} query, and $R$ be a rule in $\Pi$ with form \[p(\bar{X}) \gets A_1,\dots,A_m, \neg B_1,\dots,\neg B_n,\] where $A_1,\dots,A_m$ are positive literals, and $\neg B_{1},\dots,\neg B_n$ are negative literals. For $1 \leq i \leq m$, let $\bar{Y}_i$ be the set of variables that consists of the variables atoms $A_1,\dots,A_i$. Consider the minimal set of rules $\Pi_R$ that includes the following rules:
  \begin{enumerate}
  \item
    Rules $R_i^A$, for $2 \leq i \leq m$, defined recursively as follows:
    \begin{enumerate}
    \item
      $R_2^A = q_2^A(\bar{Y}_2) \gets A_1,A_2$.
    \item
      $R_i^A = q_i^A(\bar{Y}_{i}) \gets q_{i-1}^A(\bar{Y}_{i-1}),A_i$.
    \end{enumerate}
  \item
    Rules $R_j^B$  for $1 \leq j \leq n$,
    defined recursively as follows:
    \begin{enumerate}
    \item
      $R_0^B=  r_0^B(\bar{Y}_m) \gets q_m^A(\bar{Y}_m)$,
    \item $R_j^B = r_j^B(\bar{Y}_m) \gets r_{j-1}^B(\bar{Y}_m), \neg B'_j(\bar{Y}_m)$,
    \item $R_j^{B'} =  B'_j(\bar{Y}_m) \gets r^B_{j-1}(\bar{Y}_m), B_j$ . 
    \end{enumerate}
  \item
    A rule $R' = p(\bar{X}) \gets r_n^B(\bar{Y}_m).$
  \end{enumerate}
  The {\NRMDatalogSN} query $(p(\bar{X}),\Pi')$ that results from replacing rule $R$ in query $(p(\bar{X}),\Pi)$ with the rules in $\Pi_R$ is equivalent to query $(p(\bar{X}),\Pi)$.
\end{claim}

\begin{proof}
  We next prove this claim by induction on the numbers $m$, of positive literals, and $n$, of negative literals, in a rule $R$. The hypothesis of induction states that the query $(q(\bar{X}), \{R\})$ and its normalized query $(q(\bar{X}), \Pi_R)$ are equivalent. Since we assumed that every literal in the body of a rule must have at least one variable (see Section~\ref{sec:datalog}), to guarantee safeness, the body of the rule cannot include a negative literal without having at least a positive literal.

  \begin{enumerate}
  \item
    If $m = 1$ and $n = 0$, rule $R$ is already normalized because it is the projection rule $p(\bar{X}) \gets A_1$.
  \item
    If $m > 1$ and $n = 0$ then the normalization of rule $R$ consists of a set $\Pi_R$ of rules $R_i^A$, for $2 \leq i \leq m$, defined recursively as follows:
    \begin{align*}
      R_2^A &= q_2^A(\bar{Y}_2) \gets A_1,A_2, \\
      R_i^A &= q_i^A(\bar{Y}_{i}) \gets q_{i-1}^A(\bar{Y}_{i-1}),A_i &\text{for } 2 \leq i \leq m,\\
      R_0^B &= r_0(\bar{Y}_m) \gets q_n^A(\bar{Y}_m),\\
      R' &= p(\bar{X}) \gets r_0^B(\bar{Y}_m).
    \end{align*}
    By the induction hypothesis, the query $(p(\bar{X}), \{R_3^A,\dots,R_m^A, R_0^B, R'\})$ is equivalent to the query $(p(\bar{X}), \{R''\})$ where $R''$ is the rule $p(\bar{X}) \gets q_{2}^A(\bar{Y}_{2}), A_3, \dots, A_m$. Hence, the query $(p(\bar{X}), \Pi_R)$ is equivalent to the query $(p(\bar{X}), \{R_2^A, R''\})$.
    To show that these queries are equivalent to query $(p(\bar{X}), \{R\})$, we need to show that they have the same answers, and each answer has the same cardinality.

    Assume that a substitution $\theta$ is an answer to query $(p(\bar{X}), \{R\})$. Then, program $\{R\}$ has a derivation tree whose root is labeled with the ground literal $\theta(p(\bar{X}))$, has $m$ children labeled with the ground literals $\theta(A_i)$, for $1 \leq i \leq m$, and the edges from the root to the children are labeled with rule $R$, as is shown in Figure~\ref{fig:normalization-proof-positive} (on the left). Then, for $1 \leq i \leq m$, there is a derivation three whose root is labeled with the ground literal $\theta(A_i)$. The existence of the ground literals $\theta(A_i)$ as labels of derivation tree roots proves that the ground literal $\theta(p(\bar{X}))$ is inferred using the rules $R_2^A$ and $R''$, as is shown in the Figure~\ref{fig:normalization-proof-positive} (on the right). Then, if $\theta$ is an answer to query $(p(\bar{X}), \{R\})$ then $\theta$ is an answer to query $(p(\bar{X}), \{R_2^A, R''\})$. The same argument can be used in the contrary direction to prove that if $\theta$ is an answer to query $(p(\bar{X}), \{R_2^A, R''\})$ then $\theta$ is an answer to query $(p(\bar{X}), \{R\})$. Finally, the cardinality of $\theta(p(\bar{X}))$ is, for both queries, the product of the cardinalities of $\theta(A_i)$, for $1 \leq i \leq m$. Hence, both queries are equivalent.

  \begin{figure}[t]
    \centering
    \begin{tikzpicture}[sibling distance=1.3cm, level distance=1.5cm,
      edge from parent path={(\tikzparentnode) -- (\tikzchildnode.north)}]
      \node {$\theta(p(\bar{X}))$}
      child {
        node {$\theta(A_1)$}
        edge from parent node[Rule] {$R$}
      }
      child {
        node {$\theta(A_2)$}
        edge from parent node[Rule] {$R$}
      }
      child {
        node {$\theta(A_3)$}
        edge from parent node[Rule] {$R$}
      }
      child [draw=white] {
        node {$\dots$}
      }
      child {
        node {$\theta(A_m)$}
        edge from parent node[Rule] {$R$}
      }
      ;
    \end{tikzpicture}
    \hspace{2em}
    \begin{tikzpicture}[sibling distance=1.3cm, level distance=1.5cm,
      edge from parent path={(\tikzparentnode) -- (\tikzchildnode.north)}]
      \node {$\theta(p(\bar{X}))$}
      child {
        node {$\theta(q_2^A(\bar{Y}_2))$}
        child {
          node {$\theta(A_1)$}
          edge from parent node[Rule] {$R_2^A$}
        }
        child {
          node {$\theta(A_2)$}
          edge from parent node[Rule] {$R_2^A$}
        }
        edge from parent node[Rule] {$R''$}
      }
      child {
        node {$\theta(A_3)$}
        edge from parent node[Rule] {$R''$}
      }
      child [draw=white] {
        node {$\dots$}
      }
      child {
        node {$\theta(A_m)$}
        edge from parent node[Rule] {$R''$}
      }
      ;
    \end{tikzpicture}
    \caption{Derivation trees for the ground literal $\theta(p(\bar{X}))$ regarding query $(p(\bar{X}), \{R\})$ (on the left), and query $(p(\bar{X}), \{R_2^A, R''\})$ (on the right). The children of the nodes labeled with the positive ground literals $\theta(A_i)$ are omitted.}
    \label{fig:normalization-proof-positive}
  \end{figure}
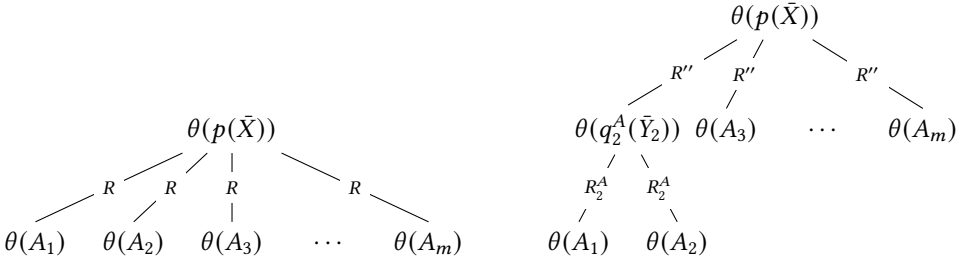
    
  \item
    If $m > 1$ and $n > 0$ then the normalization of rule $R$ consists of a set $\Pi_R$ of rules $R_i^A$, for $2 \leq i \leq m$, defined recursively as follows:
    \begin{align*}
      R_2^A &= q_2^A(\bar{Y}_2) \gets A_1,A_2, \\
      R_i^A &= q_i^A(\bar{Y}_{i}) \gets q_{i-1}^A(\bar{Y}_{i-1}),A_i
            &\text{for } 2 \leq i \leq m,\\
      R_0^B &=  r_0^B(\bar{Y}_m) \gets q_m^A(\bar{Y}_m), \\
      R_j^B &= r_j^B(\bar{Y}_m) \gets r_{j-1}^B(\bar{Y}_m), \neg {B'}_j(\bar{Y}_m)
            &\text{for } 1 \leq j \leq n, \\
 R_j^{B'} &=  B'_j(\bar{Y}_m) \gets r^B_{j-1}(\bar{Y}_m), B_j,  \\    
      R' &= p(\bar{X}) \gets r_n^B(\bar{Y}_m).
    \end{align*}
    The rules above are equivalent to the following rules:
    \begin{align*}
      R_2^A &= q_2^A(\bar{Y}_2) \gets A_1,A_2, \\
      R_i^A &= q_i^A(\bar{Y}_{i}) \gets q_{i-1}^A(\bar{Y}_{i-1}),A_i
            &\text{for } 2 \leq i \leq m,\\
      R_0^B &=  r_0^B(\bar{Y}_m) \gets q_m^A(\bar{Y}_m), \\
      R_j^B &= r_j^B(\bar{Y}_m) \gets r_{j-1}^B(\bar{Y}_m), \neg B'_j(\bar{Y}_m)
            &\text{for } 1 \leq j \leq n - 1, \\
            R_j^{B'} &=  B'_j(\bar{Y}_m) \gets r^B_{j-1}(\bar{Y}_m), B_j,  \\  
      R_{\alpha} &= t(\bar{Y}_m) \gets r_{n-1}^B(\bar{Y}_m),\\
      R_{\beta} &= r_n^B(\bar{Y}_m) \gets t(\bar{Y}_m), \neg B_n,\\
      R'' &= p(\bar{X}) \gets r_n^B(\bar{Y}_m).
    \end{align*}
    By the induction hypothesis, the query $(t(\bar{Y}_m), (\Pi_R \cup \{R_{\alpha}\}) \setminus \{ R_n^B, R'\})$ is equivalent to the query $(t(\bar{Y}_m), \{R_{\gamma}\})$ where $R_\gamma$ is the rule $t(\bar{Y}_m) \gets A_1,\dots,A_m,B_1,\dots,B_{n-1}$. Hence, the query $(p(\bar{X}),\Pi_R)$ is equivalent to the query $(p(\bar{X}), \{R_{\gamma}, R_{\beta}, R''\})$ To show that these queries are equivalent to query $(p(\bar{X}), \{R\})$, we need to show that they have the same answers, and each answer has the same cardinalities.

    Assume that a substitution $\theta$ is an answer to query $(p(\bar{X}), \{R\})$. Then, program $\{R\}$ has a derivation tree whose root is labeled with the ground literal $\theta(p(\bar{X}))$, has $m$ children labeled with the ground literals $\theta(A_i)$, and $n$ children labeled with literals $\neg\theta(B_j)$, for $1 \leq i \leq m$ and $1 \leq j \leq n$, and the edges from the root to the children are labeled with rule $R$, as is shown in Figure~\ref{fig:normalization-proof-negative} (on the left). Then, for $1 \leq i \leq m$, there is a derivation three whose root is labeled with the ground literal $\theta(A_i)$, and for $1 \leq j \leq n$, there is no derivation three whose root is labeled with the ground literal $\theta(B_j)$.  The existence of the ground literals $\theta(A_i)$ and the non-existence of the ground literals $\theta(B_j)$ as labels of derivation tree roots prove that the ground literal $\theta(p(\bar{X}))$ is inferred using the rules $R_\gamma$, $R_\beta$, and $R''$, as is shown in the Figure~\ref{fig:normalization-proof-negative} (on the right). Then, if $\theta$ is an answer to query $(p(\bar{X}), \{R\})$ then $\theta$ is an answer to query $(p(\bar{X}), \{R_\gamma, R_\beta, R''\})$. The same argument can be used in the contrary direction to prove that if $\theta$ is an answer to query $(p(\bar{X}), \{R_\gamma, R_\beta, R''\})$ then $\theta$ is an answer to query $(p(\bar{X}), \{R\})$. Finally, the cardinality of $\theta(p(\bar{X}))$ is, for both queries, the product of the cardinalities of $\theta(A_i)$, for $1 \leq i \leq m$. Hence, both queries are equivalent.
  \end{enumerate}
  Hence, we have proved that the normalization method produces an equivalent {\NRMDatalogSN} query.
\end{proof}

  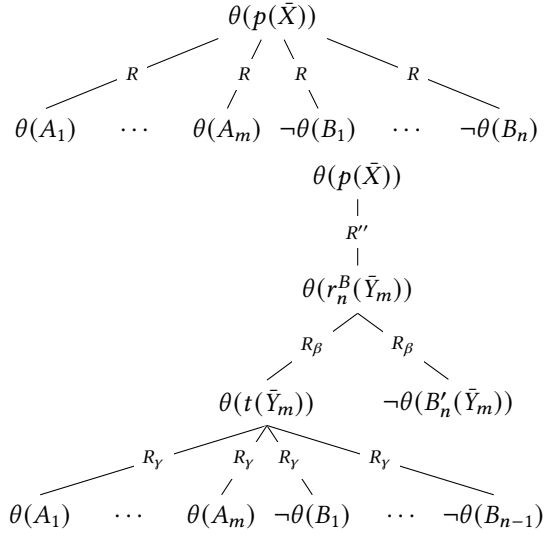
\begin{figure}[t]
    \centering
    \begin{tikzpicture}[sibling distance=1.2cm, level distance=1.5cm,
      edge from parent path={(\tikzparentnode) -- (\tikzchildnode.north)}]
      \node {$\theta(p(\bar{X}))$}
      child {
        node {$\theta(A_1)$}
        edge from parent node[Rule] {$R$}
      }
      child [draw=white] {
        node {$\dots$}
      }
      child {
        node {$\theta(A_m)$}
        edge from parent node[Rule] {$R$}
      }
      child {
        node {$\neg\theta(B_1)$}
        edge from parent node[Rule] {$R$}
      }
      child [draw=white] {
        node {$\dots$}
      }
      child {
        node {$\neg\theta(B_n)$}
        edge from parent node[Rule] {$R$}
      }
      ;
    \end{tikzpicture}
    \hfill
    \begin{tikzpicture}[sibling distance=1.2cm, level distance=1.5cm]
      \node {$\theta(p(\bar{X}))$}
      child {
        node {$\theta(r_n^B(\bar{Y}_m))$}
        child {
          node {$\theta(t(\bar{Y}_m))$}
          child {
            node {$\theta(A_1)$}
            edge from parent node[Rule] {$R_\gamma$}
          }
          child [draw=white] {
            node {$\dots$}
          }
          child {
            node {$\theta(A_m)$}
            edge from parent node[Rule] {$R_\gamma$}
          }
          child {
            node {$\neg\theta(B_1)$}
            edge from parent node[Rule] {$R_\gamma$}
          }
          child [draw=white] {
            node {$\dots$}
          }
          child {
            node {$\neg\theta(B_{n-1})$}
            edge from parent node[Rule] {$R_\gamma$}
          }
          edge from parent node[Rule] {$R_\beta$}
        }
        child[missing]
        child {
          node {$\neg\theta(B'_n(\bar{Y}_m))$}
          edge from parent node[Rule] {$R_\beta$}
        }
        edge from parent node[Rule] {$R''$}
      }
      ;
    \end{tikzpicture}
    \caption{Derivation trees for the ground atom $\theta(p(\bar{X}))$ regarding query $(p(\bar{X}), \{R\})$ (on the left), and query $(p(\bar{X}), \{R_2^A, R''\})$ (on the right). The children of the nodes labeled with the positive ground literals $\theta(A_i)$ are omitted. Nodes label with the negative ground literals $\neg\theta(B_j)$ have no children and do no derivation tree include the positive literal $\theta(B_j)$ as the root label.}
    \label{fig:normalization-proof-negative}
  \end{figure}

\subsection{Simulations between query languages}

\begin{claim}[{\SPARQL} to {\NRMDatalogSN}]
  \label{claim:sparql2md}
  The triple $(f_{12},g_{12},h_{21})$ is a simulation of {\SPARQL} in {\NRMDatalogSN}.
\end{claim}

\begin{proof}
  To prove this claim, we show that, for every {\SPARQL} query $Q$ and RDF graph $G$, it holds that $\ev{Q}{G} = h_{21}(\ev{f_{12}(Q)}{g_{12}})$ by induction on the structure of a normalized {\SPARQL} query $Q$. In this proof, we assume that $\theta$ is a {\NRMDatalogSN} substitution for the variables of the {\NRMDatalogSN} query $f_{12}(Q)$, and $\mu$ is the {\SPARQL} mapping $h_{21}(\theta)$. To show then that $\ev{Q}{G} = h_{21}(\ev{f_{12}(Q)}{g_{12}})$, we have to prove that $\mu \in \ev{Q}{G}$ if and only if $\theta \in h_{21}(\ev{f_{12}(Q)}{g_{12}})$, and $\card(\mu, \ev{Q}{G}) = \card(\theta, h_{21}(\ev{f_{12}(Q)}{g_{12}}))$.
  \begin{enumerate}
  \item
    Let $Q$ be a triple pattern $(\varX, p, \varY)$. In this case, there is a corresponding version of the triple pattern as a {\NRMDatalogSN} literal $\triple(X, p, Y)$, where $X$ and $Y$ are the corresponding variables for $\varX$ and $\varY$. The {\NRMDatalogSN} query $f_{12}(Q)$ is then $(q(X,Y),\Pi)$ where $\Pi$ is the program with a rule $q(X,Y) \gets \triple(X, p, Y)$. Let $\theta$ be the {\NRMDatalogSN} substitution $\theta=(X/s,\,Y/o)$ and $\mu$ be the {\SPARQL} mapping $h_{21}(\theta) = \{\varX \mapsto s,\, \varY \mapsto o\}$.
    \begin{enumerate}
    \item
      According to the {\NRMDatalogSN} semantics, $\theta \in \ev{f_{12}(Q)}{g_{12}(G)}$ if and only if $\triple(s,p,o) \in g_{12}(G)$. By the definition function $g_{12}$, $\triple(s,p,o) \in g_{12}(G)$ if and only if $(s,p,o) \in G$. By the {\SPARQL} semantics, the {\SPARQL} mapping $\mu$ is in $\ev{Q}{G}$ if and only if $(s,p,o) \in G$. Hence, $\theta \in \ev{f_{12}(Q)}{g_{12}(G)}$ if and only if $\mu \in \ev{Q}{G}$.
    \item
      By construction, $\card(\theta, \ev{f_{12}(Q)}{g_{12}(G)})=1$ and $\card(\mu, \ev{Q}{G})=1$. Hence, $\card(\theta, \ev{f_{12}(Q)}{g_{12}(G)}) = \card(\mu, \ev{Q}{G})$.
    \end{enumerate}
    We have shown that we can simulate triple patterns of the form $(X,p,Y)$ with {\NRMDatalogSN} queries. However, it is not difficult to apply the same argument for the other forms of triple patterns (e.g., $(X,p,o)$ or $(s,X,Y)$). Hence, {\SPARQL} triple patterns are simulable with {\NRMDatalogSN}.
  \item
    Let $Q$ be a query $(P_1 \AAND P_2)$. Assume that $\inScope(P_1)=\{\varX,\varY\}$ and $\inScope(P_2)=\{\varX,\varZ\}$. The {\NRMDatalogSN} query $f_{12}(Q)$ is then $(q(X,Y,Z),\Pi)$ where $\Pi$ is the program that consists of the rules in the programs of queries $(p_1(X,Y),\Pi_1)=f_{12}(P_1)$ and $(p_2(X,Z),\Pi_2)=f_{12}(P_2)$, the rule
    \begin{align*}
      &q(X,Y,Z) \gets p_1(X_1,Y), p_2(X_2,Z), \comp(X_1,X_2,X)
    \end{align*}
    and the rules that define the compatibility between values (which may also included in $\Pi_1$ and $\Pi_2$)
    \begin{align*}
      &\comp(X,X,X) \gets \term(X)\\
      &\comp(X,Y,X) \gets \term(X),\unull(Y)\\
      &\comp(Y,X,X) \gets \term(X),\unull(Y)\\
      &\comp(Y,Y,Y) \gets \unull(Y).
    \end{align*}
    \begin{enumerate}
    \item
      If $\theta \in \ev{f_{12}(Q)}{g_{12}(G)}$ then, by the semantics of {\NRMDatalogSN}, there exists the {\NRMDatalogSN} solutions $\theta_1=\{X_1/a_1,Y/b\}$, $\theta_2=\{X_2/a_2,Z/c\}$, and
      $\theta_3=\{X_1/a_1,X_2/a_2,X/a\}$ such that
      \begin{align*}
        &\{X_1/a_1,Y/b\} \in \ev{(p_1(X_1,Y),\Pi_1)}{g_{12}(G)},\\
        &\{X_2/a_2,Z/c\} \in \ev{(p_2(X_2,Z),\Pi_2)}{g_{12}(G)},\\
        &\{X_1/a_1,X_2/a_2,X/a\} \in \ev{(\comp(X_1,X_2,X),\Pi)}{g_{12}(G)}.
      \end{align*}
      By the induction hypothesis in $P_1$ and $P_2$, $\theta_1 \in \{X_1/a_1,Y/b\} \in \ev{(p_1(X_1,Y),\Pi_1)}{g_{12}(G)}$ and $\theta_2 \in \{X_1/a_1,Y/b\} \in \ev{(p_2(X_1,Y),\Pi_2)}{g_{12}(G)}$ if and only if mappings $\mu_1 = h_{21}(\theta_1)$ and $\mu_2 = h_{21}(\theta_2)$ hold
      $\mu_1 \in \ev{P_1}{G}$ and $\mu_2 \in \ev{P_2}{G}$. By the
      rules defining $\comp$, it holds that $\mu_1\sim\mu_2$ and
      $\mu_1 \cup \mu_2 = \mu$. By the semantics of the {\SPARQL}
      operator $\AAND$, this it holds that $\mu \in \ev{Q}{G}$. Hence, $\theta \in \ev{f_{12}(Q)}{g_{12}(G)}$ if and only if $\mu \in \ev{Q}{G}$.
    \item
      By definition,
      \[\card(\mu, \ev{Q}{G})=
        \sum_{
          \substack{
            \mu_1 \in \ev{P_1}{G}\\
            \mu_2 \in \ev{P_2}{G}\\
            \mu_1 \sim \mu_2\\
            \mu=\mu_1 \cup \mu_2
          }
        }
        \card(\mu_1,\ev{P_1}{G}) \times
        \card(\mu_2,\ev{P_2}{G}).
      \]
      By the induction hypothesis,
      \[\card(\mu, \ev{Q}{G})=
        \hspace{-4em}
        \sum_{
          \substack{
            \{X_1/a_1,Y/b\} \in \ev{(p_1(X_1,Y),\Pi_1)}{g_{12}(G)}\\
            \{X_2/a_2,Z/c\} \in \ev{(p_2(X_2,Y),\Pi_2)}{g_{12}(G)}\\
            \{X_1/a_1,X_2/a_2,X/a\} \in \ev{(\comp(X_1,X_2,X),\Pi)}{g_{12}(G)}
          }
        } \hspace{-7.5em}
        \card(\{X_1/a_1,Y/b\},\ev{f_{12}(P_1)}{g_{12}(G)}) \times
        \card(\{X_2/a_2,Z/c\},\ev{f_{12}(P_1)}{g_{12}(G)}).
      \]
      By the semantics of {\NRMDatalogSN}, we conclude that
      $\card(\mu, \ev{Q}{G})=\card(\theta, \ev{f_{12}(Q)}{g_{12}(G)})$.
    \end{enumerate}
    We have shown that we can simulate queries of the form $(P_1 \AAND P_2)$, where $\inScope(P_1)=\{\varX,\varY\}$ and $\inScope(P_2)=\{\varX,\varZ\}$, with {\NRMDatalogSN} queries. However, it is not difficult to apply the same argument for queries where $P_1$ and $P_2$ have different sets of in-scope variables. Hence, {\SPARQL} queries of the form $(P_1 \AAND P_2)$ are simulable with {\NRMDatalogSN}.
  \item
    Let $Q$ be a query $(P_1 \EXCEPT P_2)$, and $?\bar{X}$ be the list of {\SPARQL} variables in set $\inScope(Q)$. The {\NRMDatalogSN} query $f_{12}(Q)$ is then $(q(\bar{X}), \Pi)$ where $\Pi$ is the program that consists of the rules in programs of queries $(p_1(\bar{X}),\Pi_1)=f_{12}(P_1)$ and $(p_2(\bar{X}),\Pi_2)=f_{12}(P_2)$, and the rule $q(\bar{X}) \gets p_1(\bar{X}), \neg p_2(\bar{X})$.
    \begin{enumerate}
    \item
      By the semantics of {\NRMDatalogSN}, $\theta \in \ev{f_{12}(Q)}{g_{12}(G)}$ if and only if $\theta \in \ev{f_{12}(P_1)}{g_{12}(G)}$ and $\theta \notin \ev{f_{12}(P_2)}{g_{12}(G)}$. By the induction hypothesis, the last condition is equivalent to $\mu \in \ev{P_1}{G}$ and $\mu \notin \ev{P_2}{G}$. By the {\SPARQL} semantics, this is equivalent to $\mu \in \ev{Q}{G}$.
    \item
      By definition, $\card(\mu, \ev{Q}{G})=\card(\mu, \ev{P_1}{G})$ and
      \[\card(\theta, \ev{f_{12}(Q)}{g_{12}(G)})=\card(\theta, \ev{f_{12}(P_1)}{g_{12}(G)}).\]
      By the induction hypothesis, $\card(\mu, \ev{P_1}{G})=\card(\theta, \ev{f_{12}(P_1)}{g_{12}(G)})$. Hence,
      \[\card(\mu, \ev{Q}{G})=\card(\theta, \ev{f_{12}(Q)}{g_{12}(G)}).\]
    \end{enumerate}
    Hence, {\SPARQL} queries of the form $(P_1 \EXCEPT P_2)$ are simulable with {\NRMDatalogSN}.
  \item
    Let $Q$ be a {\SPARQL} query $(P_1 \UNION P_2)$. The {\NRMDatalogSN} query $f_{12}(Q)$ is then $(q(\bar{X}),\Pi)$ where $\bar{X}$ is the list with the variables in set $\inScope(Q)$, and $\Pi$ is the program that consists of the rules in program of queries $(p_1(\bar{X}),\Pi_1)=f_{12}(P_1)$ and $(p_2(\bar{X}),\Pi_2)=f_{12}(P_2)$, and the rules that correspond the operation $\UNION$, namely $q(\bar{X}) \gets p_1(\bar{X})$ and $q(\bar{X}) \gets p_2(\bar{X})$.
    \begin{enumerate}
    \item
      By the {\NRMDatalogSN} semantics, $\theta$ is a solution of $(q(\bar{X}), \Pi)$ if and only if $\theta \in \ev{(p_1(\bar{X}), \Pi_1)}{g_{12}(G)}$ or $\theta \in \ev{(p_2(\bar{X}), \Pi_2)}{g_{12}(G)}$. By the induction hypothesis, this is equivalent to that $\mu \in \ev{P_1}{G}$ or $\mu \in \ev{P_1}{G}$. By the {\SPARQL} semantics, this is equivalent to $\mu \in \ev{Q}{G}$.
    \item
      By the {\NRMDatalogSN} semantics, $\card(\theta, \ev{f_{12}(Q)}{g_{12}(G)}) = \card(\theta, \ev{f_{12}(P_1)}{g_{12}(G)}) + \card(\theta, \ev{f_{12}(P_2)}{g_{12}(G)})$ and $\card(\theta, \ev{f_{12}(Q)}{g_{12}(G)}) = \card(\mu, \ev{P_1}{G}) + \card(\mu, \ev{P_2}{G})$.
      
      By the induction hypothesis, $\card(\theta, \ev{f_{12}(P_1)}{g_{12}(G)}) = \card(\mu, \ev{P_1}{G})$ and $\card(\theta, \ev{f_{12}(P_2)}{g_{12}(G)}) = \card(\mu, \ev{P_2}{G})$. Hence $\card(\theta, \ev{f_{12}(Q)}{g_{12}(G)}) = \card(\mu, \ev{Q}{G})$.
    \end{enumerate}
    Hence, {\SPARQL} queries of the form $(P_1 \UNION P_2)$ are simulable with {\NRMDatalogSN}.
  \item
    Let $Q$ be the {\SPARQL} query $(P \FILTER \varphi)$ where is an atomic filter condition (i.e., a filter condition of the form $\varX=c$, $\varX=\varY$, or $\bound(\varX)$), and $L_\varphi$ be a set of {\NRMDatalogSN} literals defined as follows:
    \[
      L_\varphi = \left\{
        \begin{array}{ll}
          X = c, \bound(X) & \quad\text{if }\varphi \text{ is } \varX=c, \\
          X = Y, \bound(X), \bound(Y) & \quad\text{if }\varphi \text{ is } \varX=\varY, \\
          \bound(X) & \quad\text{if } \varphi \text{ is } \bound(\varX).
        \end{array}
      \right.
    \]
    The {\NRMDatalogSN} query $f_{12}(Q)$ is then $(q(\bar{X}),\Pi)$ where $\bar{X}$ is the list with the variables in set $\inScope(Q)$, and $\Pi$ is the program that consists of the rules in program of query $(p(\bar{X}),\Pi')=f_{12}(P)$, and the rule that corresponds the operation $\FILTER$, namely rule $q(\bar{X}) \gets p(\bar{X}), L_\varphi$.
    \begin{enumerate}
    \item
      By the {\NRMDatalogSN} semantics, $\theta$ is a solution of $(q(\bar{X}), \Pi)$ if and only if $\theta \in \ev{(p(\bar{X}), \Pi')}{g_{12}(G)}$, and $\theta(L_\varphi) \subseteq g_{12}(G)$. By the induction hypothesis, $\theta \in \ev{(p(\bar{X}), \Pi')}{g_{12}(G)}$ is equivalent to $\mu \in \ev{P}{G}$. By construction, $\theta(L_\varphi) \subseteq \ans(\Pi', g_{12}(G))$ if and only if $\mu(\varphi)=\true$. By the {\SPARQL} semantics, this is equivalent to $\mu \in \ev{Q}{G}$.
    \item
      By construction, every fact in $\theta(L_\varphi)$ occurs once in $g_{12}(G)$. For each fact in $F \in \theta(L_\varphi)$ there is then only one proof that $F \in \ans(\Pi, g_{12}(G))$. Hence, $\card(\theta, \ev{f_{12}(Q)}{g_{12}(G)}) = \card(\theta, \ev{f_{12}(P)}{g_{12}(G)})$. By the induction hypothesis, $\card(\theta, \ev{f_{12}(P)}{g_{12}(G)}) = \card(\mu, \ev{P}{G})$. According to the {\SPARQL} semantics, $\card(\mu, \ev{P}{G}) = \card(\mu, \ev{Q}{G})$. Hence, $\card(\theta, \ev{f_{12}(Q)}{g_{12}(G)}) = \card(\mu, \ev{Q}{G})$.
    \end{enumerate}
    Hence, {\SPARQL} queries of the form $(P \FILTER \varphi)$ are simulable with {\NRMDatalogSN}.
  \item
    Let $Q$ be the {\SPARQL} query $(\SELECT \bar{X} \; P)$. The {\NRMDatalogSN} query $f_{12}(Q)$ is then $(q(\bar{X}), \Pi)$, where $\Pi$ is the program that consists of the rules in the program of query $(p(\bar{Y}, \Pi')=f_{12}(P)$, and the rule that corresponds to the operation projects, namely rule $q(\bar{X}) \gets p(\bar{Y}),\unull(x_1),\dots,\unull(x_n)$, where $x_1,\dots,x_n$ are the variables that are in $W$ but not in $\inScope(P_1)$.
    \begin{enumerate}
    \item
      By the {\NRMDatalogSN} semantics, $\theta$ is a solution of $(q(\bar{X}), \Pi)$ if and only if there exists a solution $\theta' \in \ev{(p(\bar{Y}, \Pi')}{g_{12}(G)}$ such that $\theta(x)=\theta'(x)$ if $x \in \inScope(Q) \cap \inScope(P)$. Let $\mu = h_{21}(\theta)$ and $\mu' = h_{21}(\theta')$. By construction $\mu = \mu'|_{\inScope(Q)}$. By the induction hypothesis, $\mu' \in \ev{P}{G}$. Hence, $\mu \in \ev{Q}{G}$.
    \item
      By construction, \[\card(\theta, \ev{f_{12}(Q)}{g_{12}(G)})=\sum_{\substack{\theta'|_{\inScope(Q)}=\theta \\ \theta' \in \ev{f_{12}(P)}{g_{12}(G)}}}\card(\theta', \ev{f_{12}(P)}{g_{12}(G)}).\]
      By the induction hypothesis,
      \[\card(\theta, \ev{f_{e1,2}(Q)}{g_{12}(G)})=\sum_{\substack{
            \theta'|_{\inScope(Q)}=\theta \\
            \mu' = h_{21}(\theta') \\
            \mu' \in \ev{P}{G}}}
        \card(\mu', \ev{P}{G}).\]
      By construction,
      \[\card(\theta, \ev{f_{e1,2}(Q)}{g_{12}(G)})=\sum_{\substack{
            \mu'|_{\inScope(Q)}=h_{21}(\theta) \\
            \mu' \in \ev{P}{G}}}
        \card(\mu', \ev{P}{G}).\]
      Hence, $\card(\theta, \ev{f_{12}(Q)}{g_{12}(G)}) = \card(h_{21}(\theta), \ev{Q}{G})$.
    \end{enumerate}
    Hence, {\SPARQL} queries of the form $(\SELECT~\bar{X}~P)$ are simulable with {\NRMDatalogSN}.
  \end{enumerate}
  Hence, the triple $(f_{12} , g_{12} , h_{21})$ is a simulation of {\SPARQL} in {\NRMDatalogSN}.
\end{proof}

\begin{claim}[{\NRMDatalogSN} to {\SPARQL}]
  \label{claim:md2sparql}
  The triple $(f_{21},g_{21},h_{12})$ is a simulation of {\NRMDatalogSN} in {\SPARQL}.
\end{claim}

\begin{proof}
  To prove this claim, we consider only normalized {\NRMDatalogSN} queries $Q$, that is, queries where rules consist of projection rules, join rules and negation rules (see Section~\ref{sec:datalog2sparql}). This proof follows from induction on the structure of query $Q=(q(\bar{X}),\Pi)$ with inductive hypothesis $\theta \in \ev{Q}{D}$ if and only if $h_{12}(\theta) \in \ev{f_{21}(Q)}{g_{21}(D)}$ and $\card(\theta, \ev{Q}{D}) = \card(h_{12}(\theta), \ev{f_{21}(Q)}{g_{21}(D)})$.
  \begin{enumerate}
  \item
    If $q$ is an extensional predicate, then $f_{21}(Q)$ is the {\SPARQL} query
    \begin{align*}
      (\SELECT\; \varX_1 \dots \varX_n\;
      ((\varY, \alpha_0 ,p) \AAND\, (\varY, \alpha_1 ,\varX_1) \AAND \cdots \AAND\, (\varY, \alpha_n, \varX_n)),    
    \end{align*}
    where the {\SPARQL} variables $\varX_1 \dots \varX_n$ correspond to the $n$ {\NRMDatalogSN} variables in $\bar{X}$.
    \begin{enumerate}
    \item
      By construction, $\theta \in \ev{Q}{D}$ if and only if $h_{12}(\theta) \in \ev{f_{21}(Q)}{g_{21}(D)}$. Indeed, each colored fact $\langle q(a_1,\dots,a_n), i \rangle$ in $\col(D)$ corresponds to a subgraph $\{(u_i, \alpha_0 ,p), (u_i, \alpha_1 ,a_1), \cdots, (u_i, \alpha_n, a_n)\}$ where $u_i$ is a fresh IRI to identify the colored fact, and $a_i=\theta(x_i)$, for the $i$-th variable $x_i \in \bar{X}$.
    \item
      $\card(\theta, \ev{Q}{D})$ is the cardinality of $q(a_1,\dots,a_n)$ in multiset $D$. By construction, this is the number of subsgraphs of the form $\{(u_i, \alpha_0 ,p), (u_i, \alpha_1 ,a_1), \cdots, (u_i, \alpha_n, a_n)\}$ of $g_{21}(D)$. Hence, $\card(\theta, \ev{Q}{D}) = \card(h_{12}(\theta), \ev{f_{21}(Q)}{g_{21}(D)})$.
    \end{enumerate}
  \item
    If $q$ is an intensional predicate, then there are several rules in $\Pi$ with head $q(\bar{X})$, each one matching one of the following forms:
    \begin{itemize}
    \item $q(\bar{X}) \gets p(\bar{Y})$,
    \item $q(\bar{X}) \gets p_1(\bar{Y}_1), p_2(\bar{Y}_2)$,
    \item $q(\bar{X}) \gets p_3(\bar{Y}_3), \neg p_4(\bar{Y}_4)$.
    \end{itemize}
    The function $f_{21}(Q)$ maps each of these rules to one of the following {\SPARQL} queries:
    \begin{itemize}
    \item $(\SELECT\;\bar{X}\;f_{21}((p(\bar{Y}),\Pi)))$,
    \item $(f_{21}((p_1(\bar{X}),\Pi)) \AAND f_{21}((p_2(\bar{X}),\Pi)))$,
    \item $(f_{21}((p_3(\bar{X}),\Pi)) \EXCEPT f_{21}((p_4(\bar{X}),\Pi))$.
    \end{itemize}
    If $\{R_1,\dots,R_n\}$ is the set rules in $\Pi$ with predicate $q$ in the head, then the {\SPARQL} query $f_{21}(Q)$ has the form $(P_1 \UNION \cdots \UNION P_n)$, where $P_i$ is the corresponding {\SPARQL} query for the rule $R_i$, for $1 \leq i \leq n$.
    \begin{enumerate}
    \item
      First we will prove that the {\SPARQL} query and the {\NRMDatalogSN} query have the same answers. A substitution $\theta$ is an answer of query $Q$ if and only if at least one of the following conditions holds:
      \begin{itemize}
      \item
        For a rule $R_i$ of the form $q(\bar{X}) \gets p(\bar{Y})$, there exists a solution $\theta'$ of query $(p(\bar{Y}),\Pi)$ such that $\theta(x)=\theta'(x)$ for every variable $x \in \bar{X}$. Then, by the inductive hypothesis, there exists a solution $\mu' \in \ev{f_{21}((p(\bar{Y}),\Pi))}{g_{21}(D)}$ such that $h_{12}(\mu') = \theta'$. Let $\mu$ be the solution mapping $\mu'|_{\bar{X}}$. By construction, $\mu \in \ev{f_{21}(Q)}{g_{21}(D)}$ and $h_{12}(\mu)=\theta$.
      \item
        For a rule $R_i$ of the form $q(\bar{X}) \gets p_1(\bar{Y}_1), p_2(\bar{Y}_2)$, substitutions $\theta_1=\theta|_{\bar{Y}_1}$ and $\theta_2 =\theta|_{\bar{Y}_2}$ are solutions of queries $(p_1(\bar{Y}_1),\Pi)$ and $(p_2(\bar{Y}_2),\Pi)$. By the inductive hypothesis, there exist two solutions $\mu_1 \in \ev{f_{21}((p_1(\bar{Y}_1), \Pi))}{g_{21}(D)}$ and $\mu_2 \in \ev{f_{21}((p_2(\bar{Y}_2), \Pi))}{g_{21}(D)}$ such that $h_{12}(\mu_1)=\theta_1$ and $h_{12}(\mu_2)=\theta_2$. Let $\mu$ be the solution mapping $\mu_1 \cup \mu_2$. By construction, $\mu \in \ev{f_{21}(Q)}{g_{21}(D)}$ and $h_{12}(\mu)=\theta$.
      \item
        For a rule $R_i$ of the form $q(\bar{X}) \gets p_3(\bar{Y}_3), \neg p_4(\bar{Y}_4)$, substitution $\theta$ is a solution of query $(p_3(\bar{Y}_3),\Pi)$ and $\theta$ is not a solution of query $(p_4(\bar{Y}_3),\Pi)$. By the inductive hypothesis, there exists a solution $\mu \in \ev{f_{21}((p_3(\bar{Y}_3), \Pi))}{g_{21}(D)}$ such that $\mu \notin \ev{f_{21}((p_4(\bar{Y}_4), \Pi))}{g_{21}(D)}$, and $h_{12}(\mu)=\theta$. By construction, $\mu \in \ev{f_{21}(Q)}{g_{21}(D)}$.
      \end{itemize}
      Hence, $\theta \in \ev{Q}{D}$ if and only if there exists $\mu$ such that $f_{12}(\mu)=\theta$ and $\mu \in \ev{f_{21}(Q)}{g_{21}(D)}$.
    \item
      We next prove that the answers have the same cardinality in {\SPARQL} and {\NRMDatalogSN}. By definition,
      \[\card(\theta, \ev{Q}{D})=
        \begin{aligned}[t]
          &\sum_{\theta'|_{\bar{X}}=\theta}
            \card(\theta', \ev{(p(\bar{Y}),\Pi)}{D})\; +\\
          &\card(\theta_1, \ev{(p_1(\bar{Y}_1),\Pi)}{D})\times
            \card(\theta_2, \ev{(p_2(\bar{Y}_2),\Pi)}{D})\; +\\
          &\card(\theta, \ev{(p_3(\bar{Y}_3),\Pi)}{D}).
        \end{aligned}
      \]
      By the inductive hypothesis,
      \begin{align*}
        \card(\theta, \ev{Q}{D})
        &=
          \begin{aligned}[t]
            &\sum_{\substack{\theta'|_{\bar{X}}=\theta \\ h_{12}(\mu')=\theta'}}
            \card(\mu', \ev{f_{21}((p(\bar{Y}),\Pi))}{g_{21}(D)})\; +\\
            &\card(\mu_1, \ev{f_{21}((p_1(\bar{Y}_1),\Pi))}{g_{21}(D)})\times
              \card(\mu_2, \ev{f_{21}((p_2(\bar{Y}_2),\Pi))}{g_{21}(D)})\; +\\
            &\card(\mu, \ev{f_{21}((p_3(\bar{Y}_3),\Pi))}{g_{21}(D)})
          \end{aligned}\\
        &= \card(\mu, \ev{f_{21}(Q)}{g_{21}(D)}).
      \end{align*}
    \end{enumerate}
  \end{enumerate}
  Hence, the triple $(f_{21}, g_{21}, h_{12})$ is a simulation of {\NRMDatalogSN} in {\SPARQL}.
\end{proof}

\begin{claim}[{\MRA} to {\NRMDatalogSN}]
  \label{claim:mra2md}
  The triple $(f_{32},g_{32},h_{2,3})$ is a simulation of {\MRA} in {\NRMDatalogSN}.
\end{claim}

\begin{proof}
  We prove this claim for normalized {\MRA} expressions where the condition of a selection formula is always an equality atom (e.g., $\sigma_{A=B}(R)$). This proof follows by induction on the structure of a {\MRA} expression $E$, assuming that given a {\MRA} database $D$, for every subquery $E'$ of $E$ it holds that $t' \in \ev{E'}{D}$ if and only if there exists a {\NRMDatalogSN} solution $\theta' \in \ev{f_{32}(E)}{g_{32}(D)}$ such that $h_{2,3}(\theta')=t'$.
  \begin{enumerate}
  \item
    If $E$ is a relation name $R$ then $f_{32}(E)$ is the {\NRMDatalogSN} query $(r(\widehat{E}), \emptyset)$ where $r$ is an extensional predicate.
    \begin{enumerate}
    \item
      By definition, $t \in \ev{E}{D}$ if and only if $t$ belongs to the multiset relation corresponding to the relation name $R$ in the database $D$. By construction, $t \in \ev{E}{D}$ is thus equivalent to $\theta \in \ev{f_{32}(E)}{g_{32}(D)}$, where $h_{2,3}(\theta)=t$. Indeed, $t \in R^I$ if and only if $r(a_1,\dots,a_n) \in g_{32}(D)$ and $t=(a_1,\dots,a_n)$.
    \item
      The fact that $\card(t, \ev{E}{D})=\card(\theta, \ev{f_{32}(E)}{g_{32}(D)})$ follows by construction; the cardinality of $t$ in the multiset relation corresponding to the relation name $R$ is the same as the cardinality of fact $r(a_1,\dots,a_n)$ in multiset $g_{32}(D)$.
    \end{enumerate}
  \item
    If $E$ is a query $E_1 \cup E_2$, then $\widehat{E}_1=\widehat{E}$ and $\widehat{E}_2=\widehat{E}$, and $f_{32}(E)$ is a {\NRMDatalogSN} query $(q(\widehat{E}),\Pi)$ such that $f_{32}(E_1)=(q_1(\widehat{E}),\Pi)$ and $f_{32}(E_2)=(q_2(\widehat{E}),\Pi)$, and program $\Pi$ includes the rules $q(\widehat{E}) \gets q_1(\widehat{E})$ and $q(\widehat{E}) \gets q_2(\widehat{E})$.
    \begin{enumerate}
    \item
      By definition, $t \in \ev{E}{D}$ if and only if $t \in \ev{E_1}{D}$ or $t \in \ev{E_2}{D}$. By the induction hypothesis, $t \in \ev{E_2}{D}$ is equivalent to say that there exists $\theta$ such that $h_{2,3}(\theta)=t$ and $\theta \in \ev{f_{32}(E_1)}{g_{32}(D)}$ or $\theta \in \ev{f_{32}(E_2)}{g_{32}(D)}$. That is, $\theta \in \ev{f_{32}(E)}{g_{32}(D)}$.
    \item
      Assume the respective answers $t$ and $\theta$ described in (a).
      By definition,
      \begin{align*}
        &\card(t, \ev{E}{D}) = \card(t, \ev{E_1}{D}) + \card(t, \ev{E_2}{D}),\\
        &\card(\theta, \ev{f_{32}(E)}{g_{32}(D)}) =
          \card(\theta, \ev{f_{32}(E_1)}{g_{32}(D)}) +
          \card(\theta, \ev{f_{32}(E_2)}{g_{32}(D)}).
      \end{align*}
      By the inductive hypothesis, these two cardinalities are equal.
    \end{enumerate}
  \item
    If $E$ is a query $E_1 \Join E_2$  then $\widehat{E}_1 \cup \widehat{E}_2 = \widehat{E}$, $f_{32}(E)=(q(\widehat{E}),\Pi)$, $f_{32}(E_1)=(q_1(\widehat{E}_1),\Pi)$, $f_{32}(E_2)=(q_2(\widehat{E}_2),\Pi)$, and program $\Pi$ includes the rule $q(\widehat{E}) \gets q_1(\widehat{E}_1),q_2(\widehat{E}_2)$.
    \begin{enumerate}
    \item
      By definition, $t \in \ev{E}{D}$ if and only if there exists two tuples $t_1$ and $t_2$ such that $t_1 \sim t_2$, $t=t_1 \cup t_2$, $t_1 \in \ev{E_1}{D}$ and $t_2 \in \ev{E_2}{D}$. By the induction hypothesis, $t \in \ev{E}{D}$ if and only if there exists two {\NRMDatalogSN} solutions $\theta_1 \in \ev{f_{32}(E_1)}{g_{32}(D)}$ and $\theta_2 \in \ev{f_{32}(E_2)}{g_{32}(D)}$ where $h_{2,3}(\theta_1)=t_1$ and $h_{2,3}(\theta_2)=t_2$. Let $\theta$ be $\theta_1 \cup \theta_2$. By construction, $\theta \in \ev{f_{32}(E)}{g_{32}(D)}$ and $h_{2,3}(\theta)=t$.
    \item
      Assume the respective answers $t$, $t_1$, $t_2$, $\theta$, $\theta_1$, and $\theta_2$ described in (a).
      By definition,
      \begin{align*}
        &\card(t, \ev{E}{D}) = \card(t_1, \ev{E_1}{D}) \times \card(t_2, \ev{E_1}{D}),\\
        &\card(\theta, \ev{f_{32}(E)}{g_{32}(D)}) =
          \card(\theta_1, \ev{f_{32}(E_1)}{g_{32}(D)}) \times
          \card(\theta_2, \ev{f_{32}(E_1)}{g_{32}(D)}).
      \end{align*}
      By the inductive hypothesis, these two cardinalities are equal.
    \end{enumerate}
  \item
    If $E$ is a query $E_1 \setminus E_2$ then $\widehat{E}_1 = \widehat{E}$, $\widehat{E}_1 = \widehat{E}$, $f_{32}(E)=(q(\widehat{E}),\Pi)$, $f_{32}(E_1)=(q_1(\widehat{E}),\Pi)$, $f_{32}(E_2)=(q_2(\widehat{E}),\Pi)$, and program $\Pi$ includes the rule $q(\widehat{E}) \gets q_1(\widehat{E}),\neg q_2(\widehat{E})$.
    \begin{enumerate}
    \item
      By definition, $t \in \ev{E}{D}$ if and only if $t \in \ev{E_1}{D}$ and $t \notin \ev{E_2}{D}$. By the induction hypothesis, $t \in \ev{E}{D}$ if and only if there exists a {\NRMDatalogSN} solution $\theta$ such that $\theta \in \ev{f_{32}(E_1)}{g_{32}(D)}$, $\theta \notin \ev{f_{32}(E_2)}{g_{32}(D)}$, and $h_{2,3}(\theta)=t$. By construction, $\theta \in \ev{f_{32}(E)}{g_{32}(D)}$.
    \item
      Assume the respective answers $t$ and $\theta$ described in (a). By definition, $\card(t, \ev{E}{D}) = \card(t, \ev{E_1}{D})$ and $\card(\theta, \ev{f_{32}(E)}{g_{32}(D)}) = \card(\theta, \ev{f_{32}(E_1)}{g_{32}(D)})$. By the induction hypothesis, these two cardinalities are equal.
    \end{enumerate}
  \item
    If $E$ is a query $\pi_S(E_1)$ then $\widehat{E}=S$ and $S \subseteq \widehat{E_1}$, and $f_{32}(E)$ is a {\NRMDatalogSN} query $(q(\bar{X}),\Pi)$ such that $f_{32}(E_1)=(q_1(\widehat{E}),\Pi)$, and program $\Pi$ includes the rule $q(\widehat{E}) \gets q_1(\widehat{E}_1)$.
    \begin{enumerate}
    \item
      By definition, $t \in \ev{E}{D}$ if and only if there exists a tuple $t_1 \in \ev{E_1}{D}$ such that $t_1|_{\widehat{E}}=t$. By the induction hypothesis, $t_1 \in \ev{E}{D}$ if and only if there exists $\theta_1 \in \ev{f_{32}(E_1)}{g_{32}(D)}$ such that $h_{2,3}(\theta_1)=t_1$. Let $\theta$ be $\theta_1|_{\widehat{E}}$. By construction, $t \in \ev{E}{D}$ if and only if $\theta \in \ev{f_{32}(E)}{g_{32}(D)}$ and $h_{2,3}(\theta)=t$.
    \item
      Assume the respective answers $t$ and $\theta$ described in (a). By definition,
      \begin{align*}
        &\card(t, \ev{E}{D}) = \sum_{\substack{
          t_1 \in \ev{E_1}{D} \\ t_1|_{\widehat{E}}=t
        }} \card(t_1, \ev{E_1}{D}),\\
        &\card(\theta, \ev{f_{32}(E)}{g_{32}(D)}) = \sum_{\substack{
          \theta_1 \in \ev{f_{32}(E_1)}{g_{32}(D)} \\ \theta_1|_{\widehat{E}}=\theta
        }} \card(\theta_1, \ev{f_{32}(E_1)}{g_{32}(D)}).
      \end{align*}
      By the induction hypothesis, these two cardinalities are equal.
    \end{enumerate}
  \item
    If $E$ is a query $\rho_{A/B}(E_1)$ then $\widehat{E}=(\widehat{E}_1 \setminus \{A\}) \cup \{B\}$, $f_{32}(E)=(q(\widehat{E}),\Pi)$, and $f_{32}(E_1)=(q(\widehat{E}_1),\Pi)$.
    \begin{enumerate}
    \item
      By definition, $t \in \ev{E}{D}$ if and only if there exists a tuple $t_1 \in \ev{E_1}{D}$ where $t(C)=t_1(C)$ for every attribute $C \in \widehat{E} \setminus \{A\}$, and $t(A)=t_1(B)$. By the induction hypothesis, $t_1 \in \ev{E_1}{D}$ if and only if there exists a solution $\theta_1 \in \ev{f_{32}(E_1)}{g_{32}(D)}$ such that $h_{2,3}(\theta_1) = t_1$. Let $\theta$ be the tuple with domain $\widehat{E}$ such that $\theta(C)=\theta_1(C)$ for every attribute $C \in \widehat{E} \setminus \{A\}$, and $\theta(A)=\theta_1(B)$. By construction, $\theta \in \ev{f_{32}(E)}{g_{32}(D)}$ if and only if $\theta_1 \in \ev{f_{32}(E_1)}{g_{32}(D)}$ and $h_{2,3}(\theta)=t$. 
    \item
      Assume the respective query answers $t$, $t_1$, $\theta$, and $\theta_1$ described in (a). By definition,
      \begin{align*}
        &\card(t, \ev{E}{D}) = \card(t_1, \ev{E_1}{D}),\\
        &\card(\theta, \ev{f_{32}(E)}{g_{32}(D)}) = \card(\theta_1, \ev{f_{32}(E_1)}{g_{32}(D)}).
      \end{align*}
      By the induction hypothesis, these two cardinalities are equal.
    \end{enumerate}
  \item
    If $E$ is a query $\sigma_{A=B}(E_1)$ then $\widehat{E}=\widehat{E}_1$, $f_{32}(E)=(q(\widehat{E}),\Pi)$, and $f_{32}(E_1)=(q_1(\widehat{E}_1),\Pi)$ and program $\Pi$ includes the rule $q(\widehat{E}) \gets q_1(\widehat{E}_1), A = B$.
    \begin{enumerate}
    \item
      By definition, $t \in \ev{E}{D}$ if and only if $t \in \ev{E_1}{D}$ and $t(A)=t(B)$. By the induction hypothesis, there is an answer $\theta \in \ev{f_{32}(E_1)}{g_{32}(D)}$ such that $h_{2,3}(\theta)=t$. By construction, $\theta(A)=\theta(B)$. Then, $\theta \in \ev{f_{32}(E)}{g_{32}(D)}$.
    \item
      Assume the respective query answers $t$ and $\theta$ described in (a). By definition,
      \begin{align*}
        &\card(t, \ev{E}{D}) = \card(t_1, \ev{E_1}{D}),\\
        &\card(\theta, \ev{f_{32}(E)}{g_{32}(D)}) = \card(\theta_1, \ev{f_{32}(E_1)}{g_{32}(D)}).
      \end{align*}
      By the induction hypothesis, these two cardinalities are equal.
    \end{enumerate}
  \end{enumerate}
  Hence, the triple $(f_{32}, g_{32}, h_{2,3})$ is a simulation of {\MRA} in {\NRMDatalogSN}.
\end{proof}

\begin{claim}[{\NRMDatalogSN} to {\MRA}]
  \label{claim:md2mra}
  The triple $(f_{2,3},g_{2,3},h_{32})$ is a simulation of {\NRMDatalogSN} in {\SPARQL}.
\end{claim}

\begin{proof}
To prove this claim we consider only normalized {\NRMDatalogSN} queries $Q$, that is, queries where rules consist of projection rules, join rules and negation rules (see Section~\ref{sec:datalog2sparql}). This proof follows from induction on the structure of query $Q=(q(\bar{X}),\Pi)$ with inductive hypothesis $\theta \in \ev{Q}{D}$ if and only if $h_{32}(\theta) \in \ev{f_{2,3}(Q)}{g_{2,3}(D)}$ and $\card(\theta, \ev{Q}{D}) = \card(h_{32}(\theta), \ev{f_{2,3}(Q)}{g_{2,3}(D)})$.
  \begin{enumerate}
  \item
    If $q$ is a extensional predicate, then $f_{21}(Q)$ is the {\MRA} query $\rho_{A_1/X_1}(\cdots\rho_{A_n/X_n}(R))$, where the {\MRA} attributes $X_1,\dots,X_n$ correspond to the $n$ {\NRMDatalogSN} variables in $\bar{X}$, and $R$ is the relation name corresponding to predicate $q$.
    \begin{enumerate}
    \item
      Let $r$ be the {\MRA} relation associated to relation name $R$ in the {\MRA} database $g_{2,3}(D)$. Let $\theta$ be a {\NRMDatalogSN} answer with domain $\{X_1,\dots,X_1\}$, and $t$ be the {\MRA} tuple where $t(A_i)=\theta(X_i)$ for $1 \leq i \leq n$. By definition, $\theta \in \ev{Q}{D}$ if and only if $p(\theta(X_1),\dots,\theta(X_n)) \in D$. Because, by definition, each fact $q(a_1,\dots,a_n)$ in $D$ corresponds to a tuple $t \in r$ where $t(A_i)=a_i$ for $1 \leq i \leq n$, then $\theta \in \ev{Q}{D}$ if and only if $t \in r$. Let $s$ be a {\MRA} tuple with $\hat{s} = \{X_1,\dots,X_n\}$ where $s(X_i)=t(A_i)$, for $1 \leq i \leq n$. By definition, $t \in r$ if and only if $s \in \ev{\rho_{A_1/X_1}(\cdots\rho_{A_n/X_n}(R))}{g_{2,3}(D)}$. By construction, $s = h_{32}(\theta)$. Hence, $\theta \in \ev{Q}{D}$ if and only if $h_{32}(\theta) \in \ev{f_{2,3}(Q)}{g_{2,3}(D)}$.
    \item
      The identity $\card(\theta, \ev{Q}{D}) = \card(h_{32}(\theta), \ev{f_{2,3}Q}{f_{2,3}(D)})$ follows from the next identities:
      \begin{align*}
        \card(\theta, \ev{Q}{D})
        \begin{aligned}[t]
          &= \card(q(\theta(X_1),\dots,\theta(X_n)), D) \\
          &= \card(t, r) \\
          &= \card(s, \ev{f_{2,3}(Q)}{g_{2,3}(D)}) \\
          &= \card(h_{32}(\theta), \ev{f_{2,3}Q}{g_{2,3}(D)}).
        \end{aligned}
      \end{align*}
    \end{enumerate}
  \item
    If $q$ is an intensional predicate then there are several rules in $\Pi$ with head $q(\bar{X})$, each one has matches of the following forms:
    \begin{itemize}
    \item $q(\bar{X}) \gets p(\bar{Y})$,
    \item $q(\bar{X}) \gets p_1(\bar{Y}_1), p_2(\bar{Y}_2)$,
    \item $q(\bar{X}) \gets p_3(\bar{Y}_3), \neg p_4(\bar{Y}_4)$.
    \end{itemize}
    where $\bar{X} \subseteq \bar{Y}$, $\bar{Y}_1 \cup \bar{Y}_2 = \bar{X}$, $\bar{Y}_3 = \bar{X}$, and $\bar{Y}_4 = \bar{X}$.
    The function $f_{2,3}(Q)$ maps each of these rules to one of the following {\MRA} queries:
    \begin{itemize}
    \item $\pi_{\bar{X}}(f_{2,3}((p(\bar{Y}),\Pi)))$,
    \item $(f_{2,3}((p_1(\bar{X}),\Pi)) \Join f_{2,3}((p_2(\bar{X}),\Pi)))$,
    \item $(f_{2,3}((p_3(\bar{X}),\Pi)) \setminus f_{2,3}((p_4(\bar{X}),\Pi)))$.
    \end{itemize}
    If $\{R_1,\dots,R_n\}$ is the set rules in $\Pi$ with predicate $q$ in the head, then the {\MRA} query $f_{2,3}(Q)$ has the form $(E_1 \cup \cdots \cup E_n)$, where $E_i$ is the corresponding {\MRA} expression for the rule $R_i$, for $1 \leq i \leq n$.
    \begin{enumerate}
    \item
      First, we will prove that the {\MRA} expression and the {\NRMDatalogSN} query have the same answers. A substitution $\theta$ is an answer of query $Q$ if and only if one of the following conditions holds:
      \begin{itemize}
      \item
        There exists a solution $\theta'$ of query $(p(\bar{Y}),\Pi)$ such that $\theta(x)=\theta'(x)$ for every variable $x \in \bar{X}$. By the induction hypothesis, there exists a solution $t' \in \ev{f_{2,3}((p(\bar{Y}),\Pi))}{g_{2,3}(D)}$ such that $h_{32}(t') = \theta'$. Let $t$ be the solution mapping $t'|_{\bar{X}}$. By construction, $t \in \ev{f_{2,3}(Q)}{g_{2,3}(D)}$ and $h_{32}(t)=\theta$.
      \item
        {\sloppy Substitutions $\theta_1=\theta|_{\bar{Y}_1}$ and $\theta_2 =\theta|_{\bar{Y}_2}$ are solutions of queries $(p_1(\bar{Y}_1),\Pi)$ and $(p_2(\bar{Y}_2),\Pi)$. By the induction hypothesis, there exists two solutions $t_1 \in \ev{f_{2,3}((p_1(\bar{Y}_1), \Pi))}{g_{2,3}(D)}$ and $t_2 \in \ev{f_{2,3}((p_2(\bar{Y}_2), \Pi))}{g_{2,3}(D)}$ such that $h_{32}(t_1)=\theta_1$ and $h_{32}(t_2)=\theta_2$. Let $t$ be the {\MRA} solution $t_1 \cup t_2$. By construction, $t \in \ev{f_{2,3}(Q)}{g_{2,3}(D)}$ and $h_{32}(t)=\theta$.\par}
      \item
        {\sloppy $\theta$ is a solution of query $(p_3(\bar{Y}_3),\Pi)$ and $\theta$ is not a solution of query $(p_4(\bar{Y}_3),\Pi)$. By the induction hypothesis, there exists a solution $t \in \ev{f_{2,3}((p_3(\bar{Y}_3), \Pi))}{g_{2,3}(D)}$ such that $t \notin \ev{f_{2,3}((p_4(\bar{Y}_4), \Pi))}{g_{2,3}(D)}$, and $h_{32}(t)=\theta$. By construction, $t \in \ev{f_{2,3}(Q)}{g_{2,3}(D)}$.\par}
      \end{itemize}
      Hence, $\theta \in \ev{Q}{D}$ if and only if there exists $\mu$ such that $f_{12}(\mu)=\theta$ and $\mu \in \ev{f_{21}(Q)}{g_{21}(D)}$.
    \item
      We next prove that the answers have the same cardinality in {\MRA} and {\NRMDatalogSN}.
      By definition,
      \[\card(\theta, \ev{Q}{D})=
        \begin{aligned}[t]
          &\sum_{\theta'|_{\bar{X}}=\theta}
            \card(\theta', \ev{(p(\bar{Y}),\Pi)}{D})\; +\\
          &\card(\theta_1, \ev{(p_1(\bar{Y}_1),\Pi)}{D})\times
            \card(\theta_2, \ev{(p_2(\bar{Y}_2),\Pi)}{D})\; +\\
          &\card(\theta, \ev{(p_3(\bar{Y}_3),\Pi)}{D}).
        \end{aligned}
      \]
      By the induction hypothesis,
      \begin{align*}
        \card(\theta, \ev{Q}{D})
        &=
          \begin{aligned}[t]
            &\sum_{\substack{\theta'|_{\bar{X}}=\theta \\ h_{12}(t')=\theta'}}
            \card(t', \ev{f_{2,3}((p(\bar{Y}),\Pi))}{g_{2,3}(D)})\; +\\
            &\card(t_1, \ev{f_{2,3}((p_1(\bar{Y}_1),\Pi))}{g_{2,3}(D)})\times
              \card(t_2, \ev{f_{2,3}((p_2(\bar{Y}_2),\Pi))}{g_{2,3}(D)})\; +\\
            &\card(t, \ev{f_{2,3}((p_3(\bar{Y}_3),\Pi))}{g_{2,3}(D)})
          \end{aligned}\\
        &= \card(t, \ev{f_{2,3}(Q)}{g_{2,3}(D)}).
      \end{align*}
    \end{enumerate}
  \end{enumerate}
  Hence, the triple $(f_{2,3}, g_{2,3}, h_{32})$ is a simulation of {\NRMDatalogSN} in {\MRA}.
\end{proof}

\begin{claim}[{\MRA} to {\SPARQL}]
  \label{claim:mra2sparql}
  The triple $(f_{31},g_{31},h_{13})$ is a simulation of {\NRMDatalogSN} in {\SPARQL}.
\end{claim}

\begin{proof}
  We proof this claim for normalized {\MRA} expressions where the condition of a selection formula is always an equality atom (e.g., $\sigma_{A=B}(R)$). We proof this claim by induction on the structure of a {\MRA} expression $E$, assuming that given an {\MRA} database $D$, for every subquery $E'$ of $E$ it holds that $t' \in \ev{E'}{D}$ if and only if there exists a {\SPARQL} solution $\mu' \in \ev{f_{31}(E)}{g_{31}(D)}$ such that $h_{13}(\mu')=t'$.
  \begin{enumerate}
  \item
    If $E$ is a relation name $R$ then $f_{31}(E)$ is the {\SPARQL} query $(\SELECT\; {?A_1}\cdots {?A_n}\;P)$ where $P$ is the basic graph pattern
    $((\varX,u_b,u_r) \AAND (\varX,u_1,?A_1)\AAND \cdots \AAND (\varX,u_n,?A_n)))$, and $?A_1,\dots,?A_n$ are the variables corresponding to the attributes associated to relation name $R$.
    \begin{enumerate}
    \item
      Let $t$ be an {\MRA} tuple with $\hat{t}=\widehat{R}$, and $\mu$ be an {\SPARQL} mapping with $h_{13}(\mu)=t$. By definition, $t \in \ev{E}{D}$ if and only if tuple $t$ belongs to multiset relation $R^D$. By construction, $t \in \ev{E}{D}$ is thus equivalent to the existence of an an IRI $u$ such that the triples $(u,u_b,u_r)$, $(u,u_1,t(A_1))$, \dots, $(u,u_n,t(A_n))$ belong to the RDF graph $g_{31}(D)$. By definition, there exists such an IRI $u$ if and only if there exists a {\SPARQL} mapping $\mu' \in \ev{P}{g_{31}(D)}$ where $\mu'(\varX)=u$ and $\mu(?A_i)=t(a_i)$, for $1 \leq i \leq n$. By construction, $\mu'|_{?A_1,\dots,?A_n} = \mu$. Then, $\mu' \in \ev{P}{g_{31}(D)}$ if and only if $\mu \in \ev{f_{31}(Q)}{g_{31}(D)}$.
    \item
      The fact that $\card(t, \ev{E}{D})=\card(\mu, \ev{f_{31}(E)}{g_{31}(D)})$ follows by construction; the cardinality of $t$ in the multiset relation $R^D$ is the same as the number of IRIs $u$ such that such that the triples $(u,u_b,u_r)$, $(u,u_1,t(A_1))$, \dots, $(u,u_n,t(A_n))$ belong to the RDF graph $g_{31}(D)$.
    \end{enumerate}
  \item
    If $E$ is a query $E_1 \cup E_2$, then $\widehat{E}_1=\widehat{E}$, $\widehat{E}_2=\widehat{E}$, and $f_{31}(E)$ is the {\SPARQL} query $(f_{31}(E_1) \UNION f_{31}(E_2))$.
    \begin{enumerate}
    \item
      Let $t$ be an {\MRA} tuple with $\hat{t}=\widehat{E}$, and $\mu$ be an {\SPARQL} mapping such that $h_{13}(\mu)=t$. By definition, $t \in \ev{E}{D}$ if and only if $t \in \ev{E_1}{D}$ or $t \in \ev{E_2}{D}$. By the induction hypothesis, $t \in \ev{E}{D}$ if and only if $\mu \in \ev{f_{31}(E_1)}{g_{31}(D)}$ or $\mu \in \ev{f_{31}(E_2)}{g_{31}(D)}$. By definition, $t \in \ev{E}{D}$ if and only if $\mu \in \ev{f_{31}(E)}{g_{31}(D)}$.
    \item
      Assume the respective answers $t$ and $\mu$ described in (a).
      By definition,
      \begin{align*}
        &\card(t, \ev{E}{D}) = \card(t, \ev{E_1}{D}) + \card(t, \ev{E_2}{D}),\\
        &\card(\mu, \ev{f_{31}(E)}{g_{31}(D)}) =
          \card(\mu, \ev{f_{31}(E_1)}{g_{31}(D)}) +
          \card(\mu, \ev{f_{31}(E_2)}{g_{31}(D)}).
      \end{align*}
      By the induction hypothesis, these two cardinalities are equal.
    \end{enumerate}
  \item
    If $E$ is a query $E_1 \Join E_2$  then $\widehat{E}_1 \cup \widehat{E}_2 = \widehat{E}$, and $f_{31}(E)$ is the {\SPARQL} query $(f_{31}(E_1) \AAND f_{31}(E_2))$.
    \begin{enumerate}
    \item
      Let $t$ be an {\MRA} tuple with $\hat{t}=\widehat{E}$, and $\mu$ be an {\SPARQL} mapping such that $h_{13}(\mu)=t$. By definition, $t \in \ev{E}{D}$ if and only if there exists two tuples $t_1$ and $t_2$ such that $t_1 \sim t_2$, $t=t_1 \cup t_2$, $t_1 \in \ev{E_1}{D}$ and $t_2 \in \ev{E_2}{D}$. By the induction hypothesis, $t_1 \in \ev{E_1}{D}$ and $t_2 \in \ev{E_2}{D}$ if and only there exist two {\SPARQL} mappings $\mu_1$ and $\mu_2$ such that $h_{13}(\mu_1)=t_1$, $h_{13}(\mu_2)=t_2$, $\mu_1 \in \ev{f_{31}(E_1)}{g_{31}(D)}$, and $\mu_1 \in \ev{f_{31}(E_1)}{g_{31}(D)}$. By construction, $\mu_1 \sim \mu_2$, $\mu_1 \cup \mu_2 = \mu$, and $\mu \in \ev{f_{31}(E)}{g_{31}(D)}$. Hence, $t \in \ev{E}{D}$ if and only if $\mu \in \ev{f_{31}(E)}{g_{31}(D)}$.
    \item
      Assume the respective answers $t$, $t_1$, $t_2$, $\mu$, $\mu_1$, and $\mu_2$ described in (a).
      By definition,
      \begin{align*}
        &\card(t, \ev{E}{D}) = \card(t_1, \ev{E_1}{D}) \times \card(t_2, \ev{E_1}{D}),\\
        &\card(\mu, \ev{f_{31}(E)}{g_{31}(D)}) =
          \card(\mu_1, \ev{f_{31}(E_1)}{g_{31}(D)}) \times
          \card(\mu_2, \ev{f_{31}(E_1)}{g_{31}(D)}).
      \end{align*}
      By the induction hypothesis, these two cardinalities are equal.
    \end{enumerate}
  \item
    If $E$ is a query $E_1 \setminus E_2$ then $\widehat{E}_1 = \widehat{E}$, $\widehat{E}_1 = \widehat{E}$, and $f_{31}(E)$ is the {\SPARQL} query $(f_{31}(E_1) \EXCEPT f_{31}(E_2))$.
    \begin{enumerate}
    \item
      Let $t$ be an {\MRA} tuple with $\hat{t}=\widehat{E}$, and $\mu$ be an {\SPARQL} mapping such that $h_{13}(\mu)=t$. By definition, $t \in \ev{E}{D}$ if and only if $t \in \ev{E_1}{D}$ and $t \notin \ev{E_2}{D}$. By the induction hypothesis, $t \in \ev{E}{D}$ if and only if $\mu \in \ev{f_{31}(E_1)}{g_{31}(D)}$ and $\mu \notin \ev{f_{31}(E_2)}{g_{31}(D)}$. Hence, $t \in \ev{E}{D}$ if and only if $\mu \in \ev{f_{31}(E)}{g_{31}(D)}$.
    \item
      Assume the respective answers $t$ and $\mu$ described in (a). By definition, $\card(t, \ev{E}{D}) = \card(t, \ev{E_1}{D})$ and $\card(\mu, \ev{f_{31}(E)}{g_{31}(D)}) = \card(\mu, \ev{f_{31}(E_1)}{g_{31}(D)})$. By the induction hypothesis, these two cardinalities are equal.
    \end{enumerate}
  \item
    If $E$ is a query $\pi_S(E_1)$ then $\widehat{E}=S$ and $S \subseteq \widehat{E_1}$, and $f_{31}(E)$ is a {\SPARQL} query $(\SELECT\;W\;f_{31}(E_1))$ such that $W$ is the corresponding set of {\SPARQL} variables for the set of attributes $S$.
    \begin{enumerate}
    \item
      Let $t$ be an {\MRA} tuple with $\hat{t}=\widehat{E}$. By definition, $t \in \ev{E}{D}$ if and only if there exists a tuple $t_1 \in \ev{E_1}{D}$ such that $t_1|_{\widehat{E}}=t$. By the induction hypothesis, $t_1 \in \ev{E}{D}$ if and only if there exists $\mu_1 \in \ev{f_{31}(E_1)}{g_{31}(D)}$ such that $h_{13}(\mu_1)=t_1$. Let $\mu$ be $\mu_1|_{W}$. By construction, $t \in \ev{E}{D}$ if and only if $\mu \in \ev{f_{31}(E)}{g_{31}(D)}$ and $h_{13}(\mu)=t$.
    \item
      Assume the respective answers $t$ and $\mu$ described in (a). By definition,
      \begin{align*}
        &\card(t, \ev{E}{D}) = \sum_{\substack{
          t_1 \in \ev{E_1}{D} \\ t_1|_{\widehat{E}}=t
        }} \card(t_1, \ev{E_1}{D}),\\
        &\card(\mu, \ev{f_{31}(E)}{g_{31}(D)}) = \sum_{\substack{
          \mu_1 \in \ev{f_{31}(E_1)}{g_{31}(D)} \\ \mu_1|_{W}=\mu
        }} \card(\mu_1, \ev{f_{31}(E_1)}{g_{31}(D)}).
      \end{align*}
      By the induction hypothesis, these two cardinalities are equal.
    \end{enumerate}
  \item
    If $E$ is a query $\rho_{A/B}(E_1)$ then $\widehat{E}=(\widehat{E}_1 \setminus \{A\}) \cup \{B\}$, $f_{31}(E)$ is the {\SPARQL} query that results from consistently renaming variable $?A$ as variable $?B$ in $f_{31}(E_1)$ (i.e., $\subs_{?A/?B}(A)$), and $?A$ and $?B$ are the corresponding {\SPARQL} variables for atributes $A$ and $B$.
    \begin{enumerate}
    \item
      By definition, $t \in \ev{E}{D}$ if and only if there exists a tuple $t_1 \in \ev{E_1}{D}$ where $t(C)=t_1(C)$ for every attribute $C \in \widehat{E} \setminus \{A\}$, and $t(A)=t_1(B)$. By the induction hypothesis, $t_1 \in \ev{E_1}{D}$ if and only if there exists a solution $\mu_1 \in \ev{f_{31}(E_1)}{g_{31}(D)}$ such that $h_{13}(\mu_1) = t_1$. Let $\mu$ be the {\SPARQL} mapping with domain $(\dom(\mu') \setminus \{?A\}) \cup \{?B\}$ such that $\mu(?C)=\mu_1(?C)$ for every variable $?C \in \dom(\mu') \setminus \{?A\}$, and $\mu(?A)=\mu_1(?B)$. By construction, $h_{13}(\mu)=t$. Hence, $t \in \ev{E}{D}$ if and only if $\mu \in \ev{f_{31}(E)}{g_{31}(D)}$. 
    \item
      Assume the respective query answers $t$, $t_1$, $\mu$, and $\mu_1$ described in (a). By definition,
      \begin{align*}
        &\card(t, \ev{E}{D}) = \card(t_1, \ev{E_1}{D}),\\
        &\card(\mu, \ev{f_{31}(E)}{g_{31}(D)}) = \card(\mu_1, \ev{f_{31}(E_1)}{g_{31}(D)}).
      \end{align*}
      By the induction hypothesis, these two cardinalities are equal.
    \end{enumerate}
  \item
    If $E$ is a query $\sigma_{A=B}(E_1)$ then $\widehat{E}=\widehat{E}_1$, $f_{31}(E)=(q(\widehat{E}),\Pi)$, and $f_{31}(E_1)$ is the SPARQ query $(P_1\; \FILTER\; {?A} = {?B})$ where $?A$ and $?B$ are the corresponding {\SPARQL} variables for the {\MRA} attributes $A$ and $B$.
    \begin{enumerate}
    \item
      By definition, $t \in \ev{E}{D}$ if and only if $t \in \ev{E_1}{D}$ and $t(A)=t(B)$. By the induction hypothesis, there is an answer $\mu \in \ev{f_{31}(E_1)}{g_{31}(D)}$ such that $h_{13}(\mu)=t$. By construction, $\mu(A)=\mu(B)$. Then, $t \in \ev{E}{D}$ if and only if $\mu \in \ev{f_{31}(E)}{g_{31}(D)}$.
    \item
      Assume the respective query answers $t$ and $\theta$ described in (a). By definition,
      \begin{align*}
        &\card(t, \ev{E}{D}) = \card(t_1, \ev{E_1}{D}),\\
        &\card(\mu, \ev{f_{31}(E)}{g_{31}(D)}) = \card(\mu_1, \ev{f_{31}(E_1)}{g_{31}(D)}).
      \end{align*}
      By the induction hypothesis, these two cardinalities are equal.
    \end{enumerate}
  \end{enumerate}
  Hence, the triple $(f_{31}, g_{31}, h_{13})$ is a simulation of {\MRA} in {\NRMDatalogSN}.
\end{proof}

\begin{claim}[{\SPARQL} to {\MRA}]
  \label{claim:sparql2mra}
  The triple $(f_{13},g_{13},h_{31})$ is a simulation of {\NRMDatalogSN} in {\SPARQL}.
\end{claim}

\begin{proof}
  To prove this claim we show that, for every {\SPARQL} query $Q$ and RDF graph $G$, it holds that $\ev{Q}{G} = h_{31}(\ev{f_{13}(Q)}{g_{13}})$. For simplicity, we write $D$ instead of $D$. We next show this identity by induction on the structure of a normalized {\SPARQL} query $Q$. In this proof we assume that $t$ is a {\MRA} tuple with the attributes of the {\MRA} expression $f_{13}(Q)$, and $\mu$ is the {\SPARQL} mapping $h_{31}(t)$. To show that $\ev{Q}{G} = h_{31}(\ev{f_{13}(Q)}{g_{13}})$, we prove that $\mu \in \ev{Q}{G}$ if and only if $t \in \ev{f_{13}(Q)}{D}$ and $\card(\mu, \ev{Q}{G}) = \card(t, \ev{f_{13}(Q)}{D})$.
  \begin{enumerate}
  \item
    Case $Q$ is a triple pattern.
    \begin{enumerate}
    \item
      By definition, every triple pattern is translated to a {\MRA} expression consisting of operations $\sigma$, $\rho$, and $\pi$ over the relation name $\TripRel$. For example, if $Q$ is the triple pattern $(\varX,p,\varX)$, then $f_{12}(Q)$ is the expression $\Pi_{X}(\rho_{S/X}(\sigma_{P=p \land S=O}(\TripRel)))$. It can be shown that the triple pattern $Q=(\varX,p,\varX)$ is equivalent to the {\SPARQL} query $Q'=(\SELECT\; \varX\; ((\varX,?P,?O) \FILTER (?P=p\, \land\, \varX=?O))$. Then, $\mu \in \ev{Q}{G}$ if and only if there exists a solution $\mu' \in \ev{(\varX,?P,?O)}{G}$ such that $\mu = \mu'|_{\{X\}}$, $\mu'(?P)=p$ and $\mu'(\varX)=\mu'(?O)$. Without loss of generality, let $\mu'(\varX)=a$. Such mapping $\mu'$ is a solution of the triple pattern $(\varX,?P,?O)$ if and only if $(a,p,a) \in G$. By construction, $(a,p,a) \in G$ if and only if $(a,p,a) \in \TripRel^D$, where $D$ is the {\MRA} database $D$. If $(a,p,a) \in \TripRel^D$ then $t \in \ev{f_{13}(Q)}{D}$. Hence, $\mu \in \ev{Q}{G}$ if and only if $t \in \ev{f_{13}(Q)}{D}$. So far, we showed that the claim follows for a particular triple pattern. This result can be extended for all the triple patterns following the same procedure.
    \item
      By construction, $\card(\mu, \ev{f_{1,1}(Q)}{g_{1,1}(G)})=1$ and $\card(\mu, \ev{Q}{G})=1$. Hence, $\card(t, \ev{f_{1,1}(Q)}{g_{1,1}(G)}) = \card(\mu, \ev{Q}{G})$.
    \end{enumerate}
  \item
    Case $Q$ is a query $(P_1 \AAND P_2)$. Without loss of generality assume that $\inScope(P_1)=\{\varX,\varY\}$ and $\inScope(P_2)=\{\varX,\varZ\}$. By definition, the {\MRA} expression for query $Q$ is:
    \begin{align*}
      f_{13}(Q)
      &= f_{13}(P_1) * f_{13}(P_2) \\
      &= \pi_{X,Y,Z}(\rho_{A_1/X_1}(\rho_{A_2/X_2}(\rho_{A/X}(\CompRel)))
        \Join \rho_{X/X_1}(f_{13}(P_1)) \Join \rho_{X/X_2}(f_{13}(P_2))).
    \end{align*}
    \begin{enumerate}
    \item
      If $t \in \ev{f_{13}(Q)}{D}$ then, there are {\MRA} tuples $t_1 \in \ev{f_{13}(P_1)}{D}$, $t_2 \in \ev{f_{13}(P_2)}{D}$, and $t_3 \in \ev{\CompRel}{D}$ such that $t(X)=t_3(A)$, $t(Y)=t_2(Y)$, $t(Z)=t_3(Z)$, and $t_1(X)=t_3(A_1)$, $t_2(X)=t_3(A_2)$. Let $\mu_1=h_{31}(t_1)$, $\mu_2=h_{31}(t_2)$, and $\mu_3=h_{31}(t_3)$. By the induction hypothesis in $P_1$ and $P_2$, $t_1 \in \ev{f_{13}(P_1)}{D}$ and $t_2 \in \ev{f_{13}(P_2)}{D}$ if and only if  $\mu_1 \in \ev{P_1}{G}$ and $\mu_2 \in \ev{P_2}{G}$. By the definition of $\CompRel^{D}$, it holds that $\mu_1 \sim \mu_2$ and
      $\mu_1 \cup \mu_2 = \mu$. By the semantics of the {\SPARQL}
      operator $\AAND$, it holds then that $\mu \in \ev{Q}{G}$. Hence, $t \in \ev{f_{13}(Q)}{D}$ if and only if $\mu \in \ev{Q}{G}$.
    \item
      By definition,
      \begin{align*}
        \card(\mu, \ev{Q}{G})
        &=
          \sum_{
            \substack{
              \mu_1 \in \ev{P_1}{G}\\
              \mu_2 \in \ev{P_2}{G}\\
              \mu_1 \sim \mu_2\\
              \mu=\mu_1 \cup \mu_2
            }
          }
        \card(\mu_1,\ev{P_1}{G}) \times
        \card(\mu_2,\ev{P_2}{G}),\\
        \card(t, \ev{f_{13}(Q)}{D})
        &=
        \sum_{
          \substack{
            t_1 \in \ev{f_{13}(P_1)}{D}\\
            t_2 \in \ev{f_{13}(P_2)}{D}\\
            t_3 \in \ev{\CompRel}{D}\\
            \varphi(t_1,t_2,t_3)
          }
        }
        \card(t_3,\ev{\CompRel}{D}) \times
        \card(t_1,\ev{P_1}{D}) \times
        \card(t_2,\ev{P_2}{D}),
      \end{align*}
      where $\varphi(t_1,t_2,t_3)$ is a condition coresponding to the compatibility, that is true if and only if the following statements hold:
      \begin{enumerate}
      \item $t_1(Y) = t(Y)$,
      \item $t_2(Z) = t(Z)$, and
      \item either
        \begin{enumerate}
        \item $(t_1(X) = t(X)$ and $t_2(X) = t(X)$,
        \item $(t_1(X) = t(X)$ and $t_2(X) = t(X)$, or
        \item $(t_1(X) = t(X)$ and $t_2(X) = t(X)$.
        \end{enumerate}
      \end{enumerate}
      By the induction hypothesis, $\card(\mu_1, \ev{P_1}{G}) = \card(t_1, \ev{f_{13}(P_1)}{D})$ and $\card(\mu_2, \ev{P_2}{G}) = \card(t_2, \ev{f_{13}(P_2)}{D})$. By construction, $\card(t_3, \ev{\CompRel}{D}) = 1$. Hence, $\card(\mu, \ev{Q}{G}) = \card(t, \ev{f_{13}(Q)}{D})$.
    \end{enumerate}
  \item
    Case $Q$ is a query $(P_1 \EXCEPT P_2)$. Let $?\bar{X}$ be the list of {\SPARQL} variables in set $\inScope(Q)$. The {\MRA} query $f_{13}(Q)$ is then $f_{13}(P_1) \setminus f_{13}(P_2)$.
    \begin{enumerate}
    \item
      By definition, $t \in \ev{f_{13}(Q)}{D}$ if and only if $t \in \ev{f_{13}(P_1)}{D}$ and $t \notin \ev{f_{13}(P_2)}{D}$. By the induction hypothesis, the last condition is equivalent to $\mu \in \ev{P_1}{G}$ and $\mu \notin \ev{P_2}{G}$. By the {\SPARQL} semantics, $t \in \ev{f_{13}(Q)}{G}$ if and only if $\mu \in \ev{Q}{G}$.
    \item
      By definition, $\card(\mu, \ev{Q}{G})=\card(\mu, \ev{P_1}{G})$ and $\card(t, \ev{f_{13}(Q)}{D}) = \card(t, \ev{f_{13}(P_1)}{D})$. By the induction hypothesis, $\card(\mu, \ev{P_1}{G}) = \card(t, \ev{f_{13}(P_1)}{D})$. Hence, $\card(\mu, \ev{Q}{G})=\card(t, \ev{f_{13}(Q)}{D})$.
    \end{enumerate}
  \item
    Case $Q$ is a {\SPARQL} query $(P_1 \UNION P_2)$. The {\MRA} expression $f_{13}(Q)$ is then $f_{13}(P_1) \cup f_{13}(P_2)$.
    \begin{enumerate}
    \item
      By definition, $t \in \ev{f_{13}(Q)}{G}$ if and only if $t \in \ev{f_{13}(P_1)}{D}$ or $t \in \ev{f_{13}(P_2)}{D}$. By the induction hypothesis, $t \in \ev{f_{13}(Q)}{G}$ if and only if $\mu \in \ev{P_1}{G}$ or $\mu \in \ev{P_1}{G}$. By the {\SPARQL} semantics, $t \in \ev{f_{13}(Q)}{G}$ if and only if $\mu \in \ev{Q}{G}$.
    \item
      By definition,
      \begin{align*}
        \card(t, \ev{f_{13}(Q)}{D})
        &= \card(t, \ev{f_{13}(P_1)}{D}) +
          \card(t, \ev{f_{13}(P_2)}{D}),\\
        \card(\mu, \ev{Q}{G})
        &= \card(\mu, \ev{P_1}{G}) +
          \card(\mu, \ev{P_2}{G}).
      \end{align*}
      By the induction hypothesis,
      \[\card(t, \ev{f_{13}(P_1)}{D}) = \card(\mu, \ev{P_1}{G})
      \quad\text{and}\quad
      \card(t, \ev{f_{13}(P_2)}{D}) = \card(\mu, \ev{P_2}{G}).\]
      Hence $\card(t, \ev{f_{13}(Q)}{D}) = \card(\mu, \ev{Q}{G})$.
    \end{enumerate}
  \item
    Case $Q$ is a {\SPARQL} query $(P \FILTER \varphi)$ where $\varphi$ is an atomic filter condition (i.e., a filter condition of the form $\varX=c$, $\varX=\varY$, or $\bound(\varX)$). The {\MRA} expression $f_{13}(Q)$ is then $\sigma_{\psi}(f_{13}(P))$ where $\psi$ is the {\MRA} selection condition defined as follows:
    \[
      \psi = \left\{
        \begin{array}{ll}
          X = c \land \neg(X=\bot) & \quad\text{if }\varphi \text{ is } \varX=c, \\
          X = Y \land \neg(X=\bot) \land \neg(Y=\bot) & \quad\text{if }\varphi \text{ is } \varX=\varY, \\
          \neg(X=\bot) & \quad\text{if } \varphi \text{ is } \bound(\varX).
        \end{array}
      \right.
    \]
    \begin{enumerate}
    \item
      By definition, $t \in \ev{f_{13}(Q)}{D}$ if and only if $t \in \ev{f_{13}(P)}{D}$ and $t$ satisfies condition $\psi$. It is not difficult to see that $t$ satisfies condition $\psi$ if and only if $\mu$ satisfies condition $\varphi$. By the induction hypothesis, $t \in \ev{f_{13}(P)}{D}$ if and only if $\mu \in \ev{P}{G}$. Hence, $\mu \in \ev{f_{13}(Q)}{D}$ if and only if $\mu \in \ev{Q}{G}$.
    \item
      By definition, if $t$ and $\mu$ satisfy the respective conditions, then:
      \begin{align*}
        \card(t, \ev{f_{13}(Q)}{D}) &= \card(t, \ev{f_{13}(P)}{D}),\\
        \card(\mu, \ev{Q}{G}) &= \card(\mu, \ev{P}{G}).
      \end{align*}
      By the induction hypothesis, $\card(t, \ev{f_{13}(P)}{D}) = \card(\mu, \ev{P}{G})$. Hence,
      \[\card(t, \ev{f_{13}(Q)}{D}) = \card(\mu, \ev{Q}{G}).\]
    \end{enumerate}
  \item
    Case $Q$ is a {\SPARQL} query $(\SELECT\; ?\bar{X} \; P)$. The {\MRA} expression $f_{13}(Q)$ is then $\pi_{\bar{X}}(f_{13}(P) \Join \Delta_{\bar{Y}})$, where $\bar{X}$ is the corresponding set of attributes for the variables $?\bar{X}$ and $\bar{Y}$ is the correponding set of attributes for the variables in set $\inScope(P) \setminus \inScope(Q)$.
    \begin{enumerate}
    \item
      By definition, $t \in \ev{f_{13}(Q)}{D}$ if and only if $t(Y)=\bot$ for every attribute name $Y \in \bar{Y}$ and there exists a solution $t' \in \ev{f_{13}(P)}{D}$ such that $t'(A)=t(A)$ for every attribute $A \in \bar{X} \setminus \bar{Y}$. Let $\mu' = h_{31}(t')$. By the induction hypothesis, $t' \in \ev{f_{13}(P)}{D}$ if and only if $\mu' \in \ev{P}{G}$. By construction $\mu = \mu'|_{?\bar{X}}$. Hence, $t \in \ev{f_{13}(Q)}{D}$ if and only if $t \in \ev{P}{G}$.
    \item
      By construction,
      \begin{align*}
        \card(t, \ev{f_{13}(Q)}{D})
        &= \sum_{\substack{
          t'|_{\inScope(Q)} = t \\
          t' \in \ev{f_{13}(P)}{D}}}
        \card(t', \ev{f_{13}(P)}{D}),\\
        \card(\mu, \ev{Q}{G})
        &= \sum_{\substack{
          \mu'|_{\inScope(Q)} = \mu \\
          \mu' \in \ev{P}{G}}}
        \card(\mu', \ev{P}{G}),
      \end{align*}
      By the induction hypothesis, $\card(t', \ev{f_{13}(P)}{D}) = \card(\mu', \ev{P}{G})$. Hence,
      \[\card(t, \ev{f_{13}(Q)}{D}) = \card(\mu', \ev{Q}{G}).\]
    \end{enumerate}
  \end{enumerate}
  Hence, the triple $(f_{13} , g_{13} , h_{31})$ is a simulation of {\SPARQL} in {\MRA}.
\end{proof}


\end{document}